\documentclass[12pt,usenames,dvipsnames]{article}
\PassOptionsToPackage{usenames,dvipsnames}{color}
\usepackage{amssymb}
\usepackage{amsfonts}
\usepackage[fleqn]{amsmath}
\usepackage{geometry}
\usepackage{graphicx}
\usepackage{subfigure}
\usepackage{natbib}
\usepackage{abraces}
\usepackage{mathtools}
\usepackage{subfigure}
\usepackage{float}
\usepackage{enumitem}
\usepackage[font=scriptsize,justification=centering]{caption}
\usepackage{setspace}
\usepackage{hyperref}
\usepackage{bigints}
\usepackage[usenames, dvipsnames]{color}
\usepackage{ed}      	
\hypersetup{ pdftitle= Sorting and Team Formation, pdfauthor=Job Boerma and Aleh Tsyvinski and Alexander Zimin,
colorlinks, citecolor=black,
filecolor=black,                                                linkcolor=Blue,                                                                                                pdftitle=title,                                                pdfauthor=author,                                                bookmarks,
citecolor= Blue,                                                pdftex}

\usepackage{tikz}
\usetikzlibrary{matrix,fit}
\pgfkeys{tikz/mymatrix/.style={matrix of math nodes,inner sep=0pt,row sep=0em,column sep=0em,nodes={inner sep=6pt}}}

\usepackage{amsthm}
\usepackage{mathrsfs}
\theoremstyle{definition}
\newtheorem{theorem}{Theorem}
\newtheorem*{theorem*}{Theorem}

\newtheorem{assumption}{Assumption}

\newtheorem{claim}{Claim}

\newtheorem{implication}{Implication}
\newtheorem{lemma}{Lemma}

\newtheorem{proposition}[theorem]{Proposition}

\usepackage{refstyle}
\usepackage[noabbrev]{cleveref}

\begin{document}



\title{Multidimensional Sorting: Comparative Statics\thanks{%
\baselineskip12.5pt \noindent We thank Anmol Bhandari, Felix Bierbrauer, Matteo Camboni, Hector Chade, Roberto Corrao, Jan Eeckhout, Joel Flynn, Philipp Kircher, Rasmus Lentz, Guido Lorenzoni, Ilse Lindenlaub, George Mailath, Paolo Martellini, Simon Mongey, Christian Moser, Soumik Pal, Alireza Tahbaz-Salehi, Fedor Sandomirskiy, Vasiliki Skreta, Stefan Steinerberger, Coen Teulings, and Venky Venkateswaran for useful comments.}}
\author{ Job Boerma \\
{\small { \hspace{2 cm} University of Wisconsin-Madison \hspace{2 cm} } }\\  \vspace{-0.5 cm}
\and  Andrea Ottolini \\
{\small { \hspace{4 cm}  University of Alabama-Birmingham\hspace{4 cm} } }\\  \vspace{-0.5 cm}
\and  Aleh Tsyvinski \\
{\small { \hspace{4 cm}  Yale University\hspace{4 cm} } }\\  \vspace{-0.5 cm}
}
\date{\vspace{0.4 cm} August 2026 \vspace{-0.2 cm} }
\maketitle

\begin{abstract}
\fontsize{12.3pt}{20.0pt} 


\noindent \selectfont Characterizing multidimensional sorting problems is notoriously difficult – solutions are known only for a small number of examples. Our main results completely characterize the comparative statics of sorting and earnings with respect to technological change for a general class of multidimensional assignment models. We show that symmetric technological change passes fully into worker earnings, while antisymmetric change results only in reallocation. In general, both margins adjust, with relative importance governed by the complementarity between workers and jobs and their distributions. To quantitatively illustrate our theory, we quantify responses in sorting and earnings to cognitive skill-biased technological change for U.S. data.






\vspace{0.6cm} \noindent {\bf JEL-Codes}: J01, D31, C78. 
\end{abstract}

\renewcommand{\thefootnote}{\fnsymbol{footnote}} \renewcommand{%
\thefootnote}{\arabic{footnote}} 
\thispagestyle{empty} \setcounter{page}{0}

\fontsize{12.3pt}{22.0pt} \selectfont
\newpage

\clearpage

\section{Introduction}

Recent empirical work in labor economics emphasizes that workers and jobs are multidimensional (see, for example, \citet{Acemoglu:2011} and \citet{Deming:2017}). When workers and jobs are multidimensional, how does technological change affect sorting and earnings? How do changes in earnings and labor reallocation depend on the direction of technological change? In order to answer these questions, we develop a general theory of comparative statics for multidimensional sorting models to characterize reallocation and earnings changes in response to technological change.

Assignment models, starting with \citet{Koopmans:1957}, \citet{Shapley:1971} and \citet{Becker:1973}, are an important framework for the analysis of markets with two-sided heterogeneity. Primarily, these models (see \citet{Chade:2017} and \citet{Eeckhout:2018} for comprehensive reviews) are unidimensional in worker and job  characteristics.

Characterizing the multidimensional sorting problem is notoriously difficult.\footnote{For example, \citet{Chade:2017} and \citet{Eeckhout:2018} focus on unidimensional models and point out the technical difficulty of analyzing multidimensional models.} Full solutions are known only for a small number of examples such as bilinear output functions with normally distributed workers and jobs (see \citet{Lindenlaub:2017} and references therein). Developing a tractable characterization of multidimensional assignment models for all other cases remains an important open problem. 



We develop a theory of reallocation and earnings changes in response to technological change. Our approach delivers a complete characterization of comparative statics for sorting and earnings with respect to technology for a general class of multidimensional assignment models.\footnote{We develop our theory under the classical assumptions of the optimal transport literature on the cost function and input distributions that ensure uniqueness and smoothness of equilibria \citep*{Villani:2009,Santambrogio:2015,Figalli:2017}. We discuss all assumptions in Section \ref{s:assumptions}.} 





Our main result (\Cref{eq:helmholtz} and \Cref{eq:changeinmearnings}) characterizes the equilibrium response of sorting and earnings to technological change for multidimensional sorting models with general technologies and worker and job distributions. We establish that any technological change can be decomposed into two orthogonal components: one that changes earnings and another that reallocates workers across jobs. Formally, this decomposition is given by the Helmholtz decomposition.\footnote{In short, the Helmholtz decomposition $-$ the fundamental theorem of vector calculus $-$ states that any vector field $f:\mathbb R^d\rightarrow \mathbb R^d$ can be written as the sum $f(x)= g(x)+h(x)$ of a gradient field $g(x)$, for which $\partial g_i/\partial x_j=\partial g_j/x_i$, and a divergence-free field $h(x)$, for which $\sum \partial h_i/\partial x_i = 0$.}

The first component corresponds to the part of technological change that is a gradient $-$ it fully passes through to earnings without changing the assignment. The second component is a divergence-free reallocation $-$ it reallocates workers across jobs without changing earnings. Given an exogenous change in technology, this decomposition characterizes reallocation and earnings changes and yields comparative statics for multidimensional assignment models.


We now provide intuition for our main result. Any assignment equilibrium has to satisfy two conditions. The market clearing condition states that all workers have to be assigned to jobs. The firm optimality condition equates the marginal output of hiring a worker to the marginal cost, given by the gradient of the earnings function. 

Consider a change in technology. Technological change leads to a change in earnings and to worker reallocation. This in turn changes the firm optimality condition as the changed marginal benefit of hiring a worker has to be equated to a new marginal earnings cost. It also changes the market clearing condition as workers reallocate. For the new technology, the new marginal earnings have to remain a gradient, which implies that the change in marginal earnings has to be a gradient. The new assignment has to satisfy the market clearing condition which we show restricts reallocation to be divergence free. We show that the technological change is equal to the sum of these orthogonal components: changes in marginal earnings (a gradient), and reallocation (divergence-free). This is exactly the Helmholtz decomposition of the technological change. 

The Helmholtz decomposition is unique and yields a complete characterization of how changes in technology translate into earnings and reallocation. Formally, solving for the earnings changes and sorting reduces to solving a Poisson equation.\footnote{The Poisson equation is one of the most studied partial differential equations. See, e.g., \citet{Strauss:2007} for an undergraduate level exposition and \citet{Evans:2022} for a graduate level treatment.} The solution depends only on information known at the initial equilibrium, thus fully characterizing equilibrium comparative statics with respect to technological change without solving the assignment itself. We show that the solution to the Poisson equation can be characterized as a solution to a simple quadratic optimization problem (Theorem \ref{t:alternchar}), which can be viewed as a gradient regression of technological change. This gradient regression determines the best fit of technological change in terms of a function that is restricted to be a gradient. This theorem also characterizes the role of production complementarities and the worker skill distribution in the passthrough of technological progress into earnings and allocations. Weak complementarities tilt equilibrium adjustments toward reallocation since firm output is not sensitive to changes in the worker type at the firm, while strong complementarities tilt equilibrium adjustments toward earnings changes. Reallocating worker types which are abundant is more costly on aggregate and marginal earnings adjusts more strongly to reflect technological change. 

The second main result shows that the symmetry features of technological change determine which margin $-$ earnings or labor allocation $-$ responds in equilibrium. Symmetric technological change passes fully into worker earnings, while antisymmetric change results only in reallocation.\footnote{The terms symmetric and antisymmetric originate from the Helmholtz decomposition of vector field $f(x)=Fx$ for some matrix $F$. In this case, $f$ is decomposed into the gradient field $g(x)=Gx$ and the divergence-free field $h(x)=Hx$ where $G=(F+F^\top)/2$ and $H=(F-F^\top)/2$ are the symmetric and antisymmetric part of $F$.} In general, both margins adjust, with their relative importance governed by the complementarity between workers and their job as well as by the distribution of worker types.

First, consider the case of \textit{symmetric technological change}. For example, an increase only in within-skill complementarities $-$ such as an increase in output produced by cognitive skills in cognitive jobs $-$ is symmetric. When technological change is symmetric, it can be represented as a gradient, so the Helmholtz decomposition has no divergence-free component. There is full passthrough of technological change into earnings and no labor reallocation.

Second, consider the case of \textit{antisymmetric technological change}. For example, an increase in output produced by cognitive skills in manual jobs and a decrease of output produced by manual skills in cognitive jobs. In this case, there is only the divergence-free component of the Helmholtz decomposition and no gradient component. Technological change leads only to reallocation and does not affect worker earnings.

Finally, general technological change is neither symmetric nor antisymmetric. In this case, both earnings and allocations change. The Helmholtz decomposition uniquely splits the technological change into two orthogonal components $-$ a gradient (symmetric) part that gives the change in earnings and a divergence-free (antisymmetric) part that gives labor reallocation. The importance of these two margins depends on production complementarities and on the distribution of worker skills.


We then provide a detailed closed-form example to demonstrate the Helmholtz decomposition and the role of production complementarities in determining comparative statics. We specialize the environment to a bilinear technology with a rotationally invariant worker skill distribution. In this setting, the Poisson equation that implements the Helmholtz decomposition simplifies to a linear algebra equation that decomposes any matrix into its symmetric and antisymmetric components: the Sylvester equation. This symmetric-antisymmetric decomposition of the matrix of technological change provides the analog of the Helmholtz decomposition: the symmetric part governs earnings changes, and the antisymmetric part governs reallocation. Changes in within-skill complementarity, represented by the symmetric part, change earnings while leaving the assignment unchanged. Changes in between-skill complementarity, represented by the antisymmetric part, induce reallocation. With two skill dimensions, reallocation implied by the antisymmetric component takes a particularly simple form $-$ a rotation. When between-skill complementarities change asymmetrically, workers rotate toward jobs that make greater use of their relatively strong skill. Within-skill complementarities govern the magnitude of this rotation.



We next quantitatively illustrate our model. We characterize reallocation and earnings changes in response to cognitive skill-biased technological change for an empirically relevant economy with heterogeneity in both manual and cognitive skills. We infer the distribution of worker skills using American Community Survey and O*NET data on earnings and occupation characteristics for all U.S. workers. The resulting worker skill distribution is sharply non-Gaussian: the skill distribution is multimodal, positively skewed, and has heavy tails.\footnote{\citet*{Kim:2026} adopt a semi-nonparametric estimation approach with sieve estimators and similarly find highly non-Gaussian distributions of workers skills and job characteristics using NLSY and O*NET data.} 

We consider cognitive skill-biased technological change that increases the marginal product of cognitive skills and leaves the marginal product of manual skills unchanged. We obtain three quantitative results.  First, cognitive skill-biased technological change induces workers with high (low) cognitive skills to replace workers with more (less) manual skills. Second, marginal earnings for cognitive skills rise for all workers and rise more sharply for high ability workers. This earnings change reflects the symmetric component of technological progress, which shifts the earnings schedule upward without altering allocations. Third, in stark contrast, marginal earnings for manual skills increase for workers with high cognitive skills and decreases for those with low cognitive skills $-$ even though the technological change leaves the marginal product of manual skills unchanged. 



Finally, we provide two additional results. First, we show that one can equivalently characterize reallocation by maximizing the aggregate output gain up to second order in response to changes in technology, which delivers the aggregate output gain due to reallocation as a byproduct. Second, we show how to apply the Helmholtz decomposition and Poisson equation along a trajectory of technological change to derive the equilibrium response to large changes in technology. Proposition \ref{t:flow} shows how to construct this trajectory of changes in earnings and allocation by iterating on the Helmholtz decomposition. In fact, this result suggests a novel approach to solve multidimensional assignment problems $-$ start from a known solution such as bilinear output functions with normally distributed workers and jobs and construct the trajectory of equilibrium changes in response to a trajectory of changes in the technology and the distributions of workers and jobs.\footnote{In Appendix \ref{a:distributionalch}, we extend our results to also account for changes in the distributions of workers and jobs.}

\vspace{0.4 cm}
\noindent \textbf{Literature}. The theory of markets with two-sided heterogeneity, surveyed in \citet{Chade:2017} and \citet{Eeckhout:2018}, provides detailed understanding of pairwise unidimensional assignment models in which supermodularity and submodularity yield
positive or negative sorting. Most literature either restricts attention to explicitly assortative cases or identifies conditions under which assortative predictions hold (see \citet{Eeckhout:2018b}, \citet{Chade:2020b,Chade:2017b}, and \citet{Calvo:2024} for recent prominent examples).

Explicit solutions for multidimensional assignment problems exist only in very
specific cases. Important papers by \citet{Tinbergen:1956} and \citet{Lindenlaub:2017} derive a complete characterization in the case of diagonal production and bilinear production with normally distributed workers and jobs. In contrast, we characterize comparative statics for general production technologies and input distributions.

Most closely related in terms of general structure is \citet{Fajgelbaum:2020}, who study transportation networks in spatial equilibrium. Their framework is a multidimensional assignment problem with a general cost
function that depends on both transport flows and on investment in infrastructure, and a general spatial distribution of the excess demand for goods. Formally,
goods satisfy balanced-flow constraints (a divergence condition) and local
prices clear the markets. This formulation
extends the continuous transportation problem in \citet{Beckmann:1952} to a
setting with endogenous costs determined by infrastructure investment.
\citet{Fajgelbaum:2020} use a duality approach to characterize equilibrium
conditions and solve the model computationally. Our paper studies a different type of
multidimensional assignment problem with general output and distribution functions,
and a divergence condition implied by labor market clearing, and we provide an analytical characterization of how technological change affects equilibrium prices and allocations. 

The equilibrium implications of technological change have also been studied in different types of multidimensional assignment frameworks that are related to our setting. \citet{Edmond:2019} analyzes the assignment of a continuum of workers into two occupations and provides comparative statics with respect to the parameters of the CES technology in each occupation. \citet{Ocampo:2018} examines a task assignment model with finite worker types and continuous tasks, and provides comparative statics for the task assignment for worker-replacing and skill-enhancing technological change. Our paper studies comparative statics for technological change with general output functions and continuous distributions of workers and jobs.

Our paper also relates to the mathematical literature on optimal transport. We develop our theory of comparative statics under the classical assumptions on the input distributions and cost function of this literature \citep*{Villani:2009,Santambrogio:2015,Figalli:2017}. Closest to this paper, \citet{Angenent:2003} constructs a gradient flow that moves a suboptimal assignment towards an optimum for a given technology. In contrast, our analysis starts from an optimal assignment and analyze equilibrium adjustments to technological change. Since the envelope theorem eliminates first-order variations at the optimum, comparative statics is inherently a second-order problem. We show how these second-order terms $-$ captured by complementarities and densities $-$ govern equilibrium adjustments. Our analysis is also related to the literature on the stability of optimal transport with respect to perturbations of the marginal distributions and the cost function (see, for example, \citet{Chen:2016} and \citet{Figalli:2017}), which shows that marginal earnings and allocations maintain the same regularity as a function of worker types, under suitable geometric conditions. Proposition \ref{t:flow} identifies the partial differential equation for the earnings changes and reallocation in response to technological change, which is shown to be optimal through the local-to-global criterion of in Theorem 12.46 of \citet{Villani:2009}.


This paper falls within a growing literature which builds on optimal transport theory to solve economic problems (see  \citet{Galichon:2018} and \citet{Sargent:2024} for comprehensive overviews). A recent example is work on unidimensional sorting problems with output functions which are neither supermodular nor submodular \citep{Boerma:2023,Echenique:2024,Strack:2025,Moscariello:2026}. Another example is work on team assignment problems \citep{Chade:2018,Eeckhout:2018b,BTZ:2021,Kolotilin:2025,Kolotilin:2026} in which multiple agents work together in a team as in \citet{Kremer:1993}.\footnote{\citet{Kremer:1996} studies a role assignment model where a single population of workers is split to work into teams as managers and assistants. Role assignment models fall within a distinct class of optimal transport problems which feature a fixed sum of marginal distributions \citep*{Rachev:1998}.} 
These papers depart from positive and negative sorting when agents differ along a single dimension by enriching the production technology. In contrast, this paper obtains rich sorting patterns due to multidimensional heterogeneity holding the canonical optimal transport conditions fixed.

A large literature has explored comparative statics for unidimensional assignment frameworks. \citet{Legros:2002} provides general sufficient conditions for assortative sorting and shows how changes in the surplus function yield insights about comparative statics. \citet{Mailath:2017} provide comparative statics with respect to premuneration values for an assignment model where workers make investments before matching similar to \citet{Cole:2001} and \citet{Noldeke:2015} with the difference that investments are unobservable. \citet{Boerma:2023} develop a theory of composite sorting when the output function is same-sided submodular and decreasing in skill gaps and provide comparative statics for the assignment with respect to this technology. \citet{Chade:2022} studies an assignment problem with risk and investment in skills, and provides detailed comparative statics for changes in technology, risk, and worker and firm heterogeneity. \citet{Chade:2024} extends this analysis to settings when agents on both sides of the market are risk averse. \citet{Anderson:2024} develop a general theory of increasing sorting, showing that sorting increases in the positive quadrant dependence partial order when cross-type synergy rises and is monotone in types. \citet{Teulings:2005} develops a sorting model with distance-dependent substitution and derives explicit comparative statics for changes in relative earnings in response to shifts in the skill distribution using a first-order Taylor expansion. \citet{Baccara:2020} analyzes dynamic unidimensional assignment models and provides comparative statics with respect to waiting costs. In addition, \citet{Perez:2022,Perez:2025} study a falsification-proof mechanism design problem that they represent as an optimal transport problem and provide comparative statics with respect to falsification costs. In contrast, we analyze comparative statics for static multidimensional assignment models.

 

Another class of assignment models with multiple dimensions of heterogeneity are models with unobserved heterogeneity \citep{Choo:2006}. \citet{Bojilov:2016} characterize the equilibrium and comparative statics in closed form when the output function is bilinear, workers and jobs are normally distributed, and the heterogeneity in tastes is of the continuous logit type. Our paper stays within the deterministic assignment problem and delivers equilibrium comparative statics for technological change without random shocks and for general production technologies and input distributions.

\section{Model}

We study an economy in which workers with multidimensional heterogeneous skills sort into jobs with multidimensional heterogeneous characteristics. We first provide an informal description of the model, and then state the exact assumptions in Section \ref{s:assumptions}.


\subsection{Environment}

The economy is populated by a unit measure of risk-neutral workers and firms. Each worker is endowed with a $d$-dimensional vector of skills $x=(x_1,\dots,x_d) \in X$, where $d$ denotes the number of skill dimensions (such as cognitive and manual skill). Firms offer a unit measure of jobs. Jobs are described by a $d$-dimensional vector of characteristics $z=(z_1,\dots,z_d) \in Z$. Workers and jobs are distributed according to probability density functions $f(x)$ and $g(z)$.

When workers and jobs are paired, they produce output using a production function $y_t(x,z)$. The dependence of the production function on technology is captured through $t$. Our main goal is to characterize how earnings and sorting respond to changes in technology $t$, that is, to $\dot{y}_t=\frac{d y_t}{d t}$. The initial value of the technology parameter is normalized to $t = 0$.\footnote{Throughout the paper, we omit the index $t$ when a function or its time derivative is evaluated at the initial value $t=0$ for ease of presentation.}




\subsection{Competitive Equilibrium}\label{s:eqconditions}

We next describe a competitive equilibrium for this economy given technology $t$. 


\vspace{0.3 cm}
\noindent A firm with job $z$ chooses a worker $x$ to maximize profits given worker earnings $w_t$: 
\begin{equation}
v_t(z) = \max\limits_{x \in X} \; \big\{  y_t(x,z) - w_t(x) \big\} \hspace{0.05 cm}. \label{eq:firmprob}
\end{equation}
A worker of type $x$ chooses a job $z$ to maximize earnings given firm profits $v_t$:
\begin{equation*}
w_t(x) = \max\limits_{z \in Z} \; \big\{ y_t(x,z) - v_t(z) \big\} \hspace{0.05 cm}. \label{eq:workerprob}
\end{equation*}
An assignment $\pi_{t}$ specifies a joint distribution of workers and jobs. Feasibility requires that the marginal distributions of $\pi_{t}$ equal the worker and job distributions, so that all workers and all jobs are paired. Feasibility of an assignment is thus equivalent to labor market clearing.\footnote{Formally, an assignment is a coupling between $f$ and $g$, that is, a probability measure on $X \times Z$ with marginal distributions $f$ and $g$.}

An equilibrium is a triple $(\pi_t,w_t,v_t)$ such that $(i)$ firms maximize profits, $(ii)$ workers maximize earnings, $(iii)$ the labor market clears, and $(iv)$ the resource constraint is satisfied. The resource constraint ensures that total output equals the sum of earnings and profits: 
\begin{equation*}
\int_{X \times Z} y_t(x, z) d \pi_t = \int_Z v_t(z) g(z) dz + \int_X w_t(x) f(x) dx \hspace{0.03 cm}.
\end{equation*}

\vspace{0.1 cm}
\noindent \textbf{Equilibrium Condition}. The first-order optimality condition to the firm problem (\ref{eq:firmprob}) is:
\begin{equation}
(\nabla_1 y_t)(x,z) = (\nabla w_t)(x) . \label{eq:focfirm}
\end{equation}
which must hold for all pairs $(x,z)$ on the support of the assignment.\footnote{The gradient of a function $h$ on $\mathbb{R}^d$, which is denoted by $(\nabla h)$, is the $d$-dimensional vector of partial derivatives of $h$: $(\nabla h)= ( \partial h/\partial x_1,\ldots, \partial h/\partial x_d )$. We adopt the notation $(\nabla_1 y_t)$ to denote the gradient of the output function $y_t$ with respect to its first argument only (which is also a vector). We omit the subscript for functions with a single argument. Firm optimality condition (\ref{eq:focfirm}) assumes differentiability of the earnings function, which is ensured by the assumptions in Section \ref{s:assumptions}.} Firms equate the marginal benefit of hiring worker type $x$, which is given by the marginal output $(\nabla_1 y_t)(x,z)$, to the marginal cost of hiring worker type $x$, given by the marginal earnings $(\nabla w_t)(x)$. Equation (\ref{eq:focfirm}) a standard pricing condition showing that marginal earnings equal marginal products in the assigned job in the setting of multidimensional skills. The second-order condition to the profit maximization problem (\ref{eq:firmprob}) is that $(\nabla^2_{11} y_t)(x,z) - (\nabla^2 w_t)(x)$ is negative semidefinite. 


\subsection{Assignment Problems}

We solve two problems to characterize an equilibrium.\footnote{We describe the known relationship between the planning problems and equilibrium in Appendix \ref{pf:prop:equilibrium}.}

\vspace{0.35cm}
\noindent \textbf{Primal Problem}. The planner chooses a feasible assignment $\pi_t$ to maximize aggregate output:
\begin{equation}
\max_{\pi_t \in \Pi} \int y_t (x,z ) d \pi_t . \label{e:pp}
\end{equation}
This is an optimal transport problem with multidimensional marginal distributions. 


\vspace{0.35cm}
\noindent \textbf{Dual Problem}. Equilibrium earnings $w_t$ and firm profits $v_t$ are characterized by the dual problem. The dual problem is to choose functions $w_t$ and $v_t$ that solve:
\begin{equation*}
\min_{w_t,v_t} \; \int w_t(x) f(x) d x + \int v_t(z) g(z) d z \label{e:pp_dual},
\end{equation*}
subject to the constraint $w_t(x) + v_t(z) \geq y_t(x,z)$ for any $(x,z)$. The constraint states that earnings and firm values do not exceed the surplus that would be produced by the worker-job pair. 

\vspace{0.35cm}
\noindent By Monge-Kantorovich duality, the primal problem and dual problem deliver the same optimal value.\footnote{See, for instance, Theorem 5.10 in \citet{Villani:2009}.} This implies that earnings and firm profits split the surplus on the support of the optimal assignment, $w_{t}(x)+v_{t}(z)=y_{t}(x,z)$. 

\subsection{Assumptions} \label{s:assumptions}

In this section, we discuss the model assumptions. 


\vspace{0.35 cm}
\noindent The set of workers $X$ and the set of jobs $Z$ are both compact, strictly convex, and have a smooth boundary. The probability density functions $f$ and $g$ are smooth and strictly positive on their support.\footnote{For simplicity, we assume that the measure of workers equals the measure of jobs. One possibility to incorporate unemployment is by adding positions that generate output below some threshold level for every worker type $x$ that is assigned to it, and consider these positions de facto as unemployment.} 



To ensure uniqueness and smoothness of equilibrium, we impose standard conditions from the mathematical literature on optimal transport \citep{Villani:2009,Santambrogio:2015,Figalli:2017}. Specifically, the production function satisfies both the twist condition and the Ma-Trudinger-Wang condition \citep{Ma:2005}, and the marginal product of worker $x$ and job $z$ is uniformly convex. We assume that the equilibrium is differentiable with respect to technological change.\footnote{For general discussions of regularity and differentiability in the mathematical literature on optimal transport we refer to Chapter $5$ in \citet{Figalli:2017}, where it is explained how this is connected to regularity theory for the associated elliptic problem. The elliptic problem refers to the Poisson equation in Theorem \ref{eq:changeinmearnings} and its generalization with respect to general technological change in Section \ref{s:largetc}. For recent developments, see \citet{Gonzalez:2026} and references therein.}  

 The twist condition is a standard condition in economics, also known as the Spence-Mirrlees condition \citep*{Mirrlees:1971,Spence:1973}. We assume the production technology $y_t(x,z)$ satisfies the twist condition, which means that if $(\nabla_1 y_t)(x,z) = (\nabla_1 y_t)(x,\hat{z})$, then $z = \hat{z}$ and similarly if $(\nabla_2 y_t)(x,z) = (\nabla_2 y_t)(\hat{x},z)$, then $x = \hat{x}$. Moreover, $(\nabla^2_{12} y_t)(x,z)$ is invertible. We show that the twist condition ensures that each worker type is paired with a unique job type in \Cref{pf:lemmamap}.

The Ma-Trudinger-Wang condition requires the output function to be four times continuously differentiable and the Ma-Trudinger-Wang tensor to be positive definite.\footnote{The explicit condition can be found, for instance, in Chapter $12$ of \cite{Villani:2009}, both in a coordinate dependent way (formula $12.20$) and in an intrinsic formulation (formula $12.21$).} Any bilinear production function, which is the most common choice of production function with multidimensional worker skills and job characteristics, satisfies this condition.\footnote{The bilinear technology is used in \citet{Tinbergen:1956}, \citet{Lindenlaub:2017}, \citet{Lise:2020}, \citet{BTZ2:2022} and \citet{Lindenlaub:2020}, amongst others. See \citet{Lindenlaub:2017} for discussion of other relevant production functions. In the closed-form example with bilinear production technology of Section \ref{s:closedform}, differentiability holds.} The relevance of the Ma-Trudinger-Wang condition is that it allows to rewrite the optimality for the dual solution $-$ that is, the $y_t$-concavity of the dual potential $-$ as a local differential condition involving the Hessian of the dual solution (see Theorem 12.46 in \citet{Villani:2009}). This guarantees that our candidate optimum is a global maximizer.\footnote{Examples showing that this condition is generally sharp can be found in \citet*{Loeper:2009}.}


To summarize, the twist condition ensures that each worker type is paired with a unique job type, while the Ma-Trudinger-Wang condition ensures smoothness of the earnings function. Then, the optimal assignment $\pi_t$ is given by a unique assignment function $\tau_t$. Moreover, the assignment  and its inverse are smooth.\footnote{See, for example, Theorem 5.6 and Theorem 5.7 in \citet{Figalli:2017}.} These results allow us to work with a smooth assignment function rather than with the joint distribution function $\pi_t$. We thus formulate the assignment problem  (\ref{e:pp}) as finding an optimal smooth assignment function $\tau_t$ in Section \ref{s:rearrangement}.



\section{Assignment and Rearrangement Problem} \label{s:rearrangement}

Equilibrium sorting is the solution to the planning problem that assigns heterogeneous workers to heterogeneous jobs (\ref{e:pp}). In order to understand how sorting and earnings respond to technological progress, it is useful to analyze an equivalent but simpler problem which we call a \textit{rearrangement problem}. The rearrangement problem assigns workers to replace other workers in the jobs they were assigned to under the initial equilibrium. 


\vspace{0.4 cm}
\noindent \textbf{Assignment Problem}. The assignment problem maximizes aggregate output by choosing a smooth assignment function $\tau_t$ between workers that are distributed according to density function $f$ and jobs that are distributed according to density function $g$. The planner chooses assignment $\tau_t: X \rightarrow Z$ to maximize aggregate output: 
\begin{equation*}
\max\limits_{\tau_t  \in \mathcal{T}} \int y_t(x, \tau_t(x)) f(x) dx \label{eq:outputmax} .
\end{equation*}

Feasibility requires that the mass of workers with skills in any set $A \subseteq X$ is equal to the mass of jobs they are assigned to under the assignment $\tau_{t}$. The feasibility condition ensures that every job is filled and every worker is assigned.\footnote{Formally, an assignment $\tau_t$ is feasible if it maps the worker distribution $f$ onto the job distribution $g$, meaning $\int_{\tau_t(A)} g(z) d z = \int_{A} f(x) dx$ for any Borel subset $A \subseteq X$. \label{eq:constraint}} This constraint can be written as:
\begin{equation}
\int_Z h(z) g(z) d z = \int_X h(\tau_t(x)) f(x) dx \label{eq:probconstraint}
\end{equation}
for any smooth function $h$. The set of feasible assignment functions $\mathcal{T}$ is the set of functions that satisfy (\ref{eq:probconstraint}). The optimal assignment given technology $t$ is denoted $\tau_t$, where $\tau_t(x)$ represents the job that worker  type $x$ is assigned to. 

\vspace{0.4 cm}
\noindent \textbf{Rearrangement Problem}. The assignment problem can also be described from a different perspective. Instead of asking which job each worker is assigned to, we ask which worker replaces a worker in the job they were assigned under the initial assignment. This rearrangement problem is equivalent to the assignment problem and is used in Section \ref{s:comparativestats} for analyzing comparative statics.

Let $\tau$ denote the optimal assignment for the initial technology $t=0$. The rearrangement problem fixes this assignment and asks how workers are reallocated across those different jobs. The rearrangement function $r_t(x)$ stipulates which worker is assigned to the job which was assigned to worker $x$ under the initial assignment $\tau$. That is, $r_{t}(x)$ is the replacement for worker $x$ under technology $t$ in job $\tau(x)$. 

The rearrangement problem is to choose a rearrangement $r_t$ to maximize aggregate output:
\begin{equation*}
\max\limits_{r_t \in \mathcal{R}} \int y_t ( r_t(x), \tau(x)) f(x) dx \label{eq:outputmaxrear}
\end{equation*}
over all rearrangements $r_t$ that map $f$ onto itself, or:
\begin{equation}
\int_X h(z) f(z) d z = \int_X h(r_t(x)) f(x) dx \label{eq:probconstraintrear}
\end{equation}
for any smooth function $h$. The set of feasible rearrangements, that is, the set of assignment functions satisfying (\ref{eq:probconstraintrear}) is denoted by $\mathcal{R}$. A rearrangement problem is an output maximization problem that asks: when the technology changes, which worker optimally takes job $\tau(x)$ that was assigned to worker $x$ under the initial technology? The rearrangement formulation specifies the problem in worker-worker terms rather than in worker-job terms.

\vspace{0.4 cm}
\noindent \textbf{Equivalence between the Assignment Problem and the Rearrangement Problem}. The next result shows that the assignment problem and the rearrangement problem are equivalent. 


\begin{lemma}{\textit{Equivalence between the Assignment Problem and the Rearrangement Problem.}}\label{l:reallocation}
The planning problem and the rearrangement problem deliver the same aggregate output: 
\begin{equation*}
\max\limits_{\tau_t \in \mathcal{T}} \int y_t(x, \tau_t(x)) f(x) dx = \max\limits_{r_t \in \mathcal{R}} \int y_t ( r_t(x), \tau(x)) f(x) dx 
\end{equation*}
with the optimal assignment and rearrangement linked as $\tau_t(r_t(x)) = \tau (x)$. 
\end{lemma}

\vspace{0.1 cm}
\noindent We prove this result in Appendix \ref{a:reallocationproblem}. For the initial technology $\tau_t = \tau$, implying that the optimal rearrangement at the initial technology is the identity $r_t(x) = x$. 

Lemma \ref{l:reallocation} proves that the rearrangement $r_{t}$ and the assignment $\tau_{t}$ describe equilibrium from complementary perspectives. The assignment $\tau_t$ specifies which job a worker holds for technology $t$. The rearrangement $r_{t}$ describes which worker fulfills the job initially held by worker $x$. The connection is formalized by $\tau_{t}(r_{t}(x))=\tau(x)$: under assignment $\tau_t$ worker $r_{t}(x)$ works in the job that worker $x$ worked in under the initial assignment $\tau$.

\vspace{0.4 cm}
\noindent \textbf{Equilibrium Condition for the Rearrangement Problem}. We use the equivalence between the assignment problem and the rearrangement problem in Lemma \ref{l:reallocation} to restate the firm optimality condition (\ref{eq:focfirm}). In the rearrangement formulation, the firm optimality condition takes the form:
\begin{equation}
(\nabla_1 y_t) (r_t(x),\tau(x)) = (\nabla_1 w_t) (r_t(x)) \hspace{0.04 cm}. \label{eq:hedonicw}
\end{equation}
This is a standard marginal pricing condition: the marginal compensation for replacement worker $r_t(x)$, which is denoted by $(\nabla_1 w_t) (r_t(x) )$, is equal to the marginal product of this worker in the job they take over, which is $(\nabla_1 y_t) (r_t(x) , \tau(x))$. 

For notational convenience, we define the output produced by replacement worker $\tilde{x}$ when they work in the job initially done by worker $x$ as: 
\begin{equation}
\mathsf{y}_t(\tilde{x},x) := y_t(\tilde{x},\tau(x)) \hspace{0.04 cm}. \label{eq:outputww}
\end{equation} 
The output complementarity between replacement worker $\tilde{x}$ and the original worker $x$ is then given by the cross-derivative $(\nabla^2_{12}\mathsf{y}_t)(\tilde{x},x)$. 

\section{Comparative Statics} \label{s:comparativestats}

This section derives the main result of the paper: comparative statics for multidimensional sorting. We proceed in four steps. \Cref{s:feasiblereall} establishes feasibility conditions for reallocation. \Cref{s:changeearnings} identifies two distinct margins of equilibrium adjustment $-$ changes in earnings and allocations $-$ by analyzing the firm optimality condition (\ref{eq:hedonicw}). \Cref{s:solvingrealloc} establishes that comparative statics is given by the Helmholtz decomposition. \Cref{s:symmetry} shows how symmetry and antisymmetry of technological change determine whether equilibrium adjustments occur in earnings, sorting, or both. 



\subsection{Feasible Labor Reallocation} \label{s:feasiblereall}

This section characterizes the set of feasible labor reallocations. As technology changes, workers are rearranged and this rearrangement must satisfy labor market clearing. Labor market clearing implies two restrictions on reallocation specific to assignment models. First, the worker inflow into each job must equal the worker outflow out of this job everywhere inside the type space $-$ a divergence-free condition. Second, workers in the boundary of the type space remain inside the boundary. Labor market clearing (\ref{eq:probconstraintrear}) requires the worker distribution to be preserved under rearrangement, and differentiating this market clearing condition yields the divergence-free and boundary conditions. Together with the firm optimality condition developed in \Cref{s:changeearnings}, these conditions enable our comparative statics analysis for assignment models.


\vspace{0.4 cm}
\noindent \textbf{Labor Reallocation}. We first differentiate the labor market clearing condition (\ref{eq:probconstraintrear}) with respect to technology: 
 \begin{equation*}
0 = \frac{d}{dt} \int_X h(r_t(x)) f(x) dx = \int_X (\nabla h) (r_t(x)) \cdot \dot{r}_t(x) f(x) dx .
\end{equation*}
At the initial technology $t=0$, where $r(x) = x$, this corresponds to $0 = \int (\nabla h) (x) \cdot \dot{r}(x) f(x) dx$.\footnote{In \Cref{s:largetc} we show optimal reallocation is  differentiable with respect to technological change. We also show how our analysis can also be used to characterize equilibrium comparative statics for large changes in technology.} Integration by parts then gives:
 \begin{equation*}
0 = \int_{\partial X} h(x) \dot{r}(x) f(x) \cdot n(x) dx - \int_X h (x) \nabla \cdot (\dot{r}(x) f(x)) dx \;, 
\end{equation*}
where $\nabla \cdot$ denotes the divergence and $n(x)$ is the unit outward normal vector at the boundary.\footnote{Given a vector field $v:\mathbb R^d\rightarrow \mathbb R^d$, its divergence is the function $\nabla \cdot v=\sum \frac{\partial v_i}{\partial x_i}$. The divergence measures the amount of mass flowing in and out of a point. A vector field is divergence-free if the divergence is zero everywhere.}  Since this must hold for all $h$, the reallocation $\dot{r}$ must satisfy two conditions:
\begin{align}
0 & = \nabla \cdot (\dot{r}(x) f(x)) \hspace{3.65 cm} \text{$ x \in X$} \label{eq:rearr0} \\
0 & = \dot{r}(x) f(x) \cdot n(x)  \hspace{3.5 cm} \text{$x \in \partial X$} \label{eq:rearr0b}
\end{align}
The first condition (\ref{eq:rearr0}) requires that reallocation is divergence-free in the interior of the worker types: the worker inflow into each job must equal the worker outflow out of this job. The second condition (\ref{eq:rearr0b}) requires that workers in the boundary of the worker types remain inside the boundary: a worker in the boundary cannot be replaced by a worker from the interior.\footnote{Alternatively, one could solve directly for the change in the optimal assignment by differentiating the assignment function $\tau_t$ with respect to technology. Under this approach, reallocations also play a prominent role even if not explicit. We introduce the reassignment problem to make this link explicit, and to provide clear expressions for the comparative statics. In \Cref{s:largetc}, we also present the direct approach of differentiating the assignment function $\tau_t$ to analyze the change in the assignment function with respect to large changes in technology.} We denote the set of reallocations that satisfy these conditions by $\mathcal{D}$.



\vspace{0.4 cm}
\noindent \textbf{Understanding Labor Reallocation}. The feasibility conditions (\ref{eq:rearr0}) and (\ref{eq:rearr0b}) can be understood in terms of reallocation. The vector $\dot{r}(x)$ describes how the worker type assigned to occupation $\tau(x)$ changes with technological progress. Under the initial technology, worker $x$ holds job $\tau(x)$. As technology changes, that job is taken over by replacement $r_{t}(x)$, so that $\dot{r}(x)$ gives the rate of replacement in occupation $\tau(x)$. Scaling by the density $f(x)$ gives the total replacement rate of worker type $x$ as $\dot{r}(x)f(x)$. The first reallocation condition (\ref{eq:rearr0}) means that the total inflow and outflow must be equal $-$ a worker that is replaced from one occupation must be assigned to another. Since the total replacement rate into and out of each job must balance, the replacement rate does not change at worker $x$, $\nabla \cdot (\dot{r}(x) f(x)) = 0$. 

The second reallocation condition considers the rate of replacement $\dot{r}(x) f(x)$ for all workers on the boundary of the worker type space, $x \in \partial X$. At the boundary, replacements cannot come from outside the support. This means that the replacement rate of workers in the outward direction of the worker type space is equal to zero, or $\dot{r}(x) \cdot n(x) = 0$. This implies that reallocation is parallel to the boundary: workers in the boundary remain in the boundary.\footnote{Workers in the boundary cannot move to the interior as no other worker can replace them when reallocation is smooth.}


\vspace{0.4 cm}
\noindent \textbf{Example: Labor Reallocation as an Infinitesimal Rotation}. In order to build intuition for labor reallocation, consider a setting where the skills distribution is rotationally invariant over a $d$-dimensional sphere.\footnote{Formally, the density $f(x) = \tilde{f} \big(\frac{1}{2} \lVert x \rVert ^2\big)$ only depends on the distance from the origin.} Moreover, consider reallocation that is linear in skills, $\dot{r}(x) = \dot{R} x $. Lemma \ref{l:antisymmetrics} shows that reallocation is feasible when the matrix $\dot{R}$ is antisymmetric, that is, $\dot{R} = - \dot{R}^\top$. 


\begin{lemma}{\textit{Antisymmetric Reallocations}.} \label{l:antisymmetrics}
Suppose the distribution of workers $f$ is rotationally invariant and that reallocation is linear in worker skills, $\dot{r}(x) = \dot{R} x$. If reallocation is feasible, then $\dot{R}$ is antisymmetric.
\end{lemma}


\noindent The proof is in Appendix \ref{a:antisymmetrics}. 

\begin{figure}[t!]
\begin{center}
\subfigure{\includegraphics[trim=0.0cm 0.0cm 0.0cm 0.0cm, width=0.58\textwidth,height=0.35\textheight,angle=0]{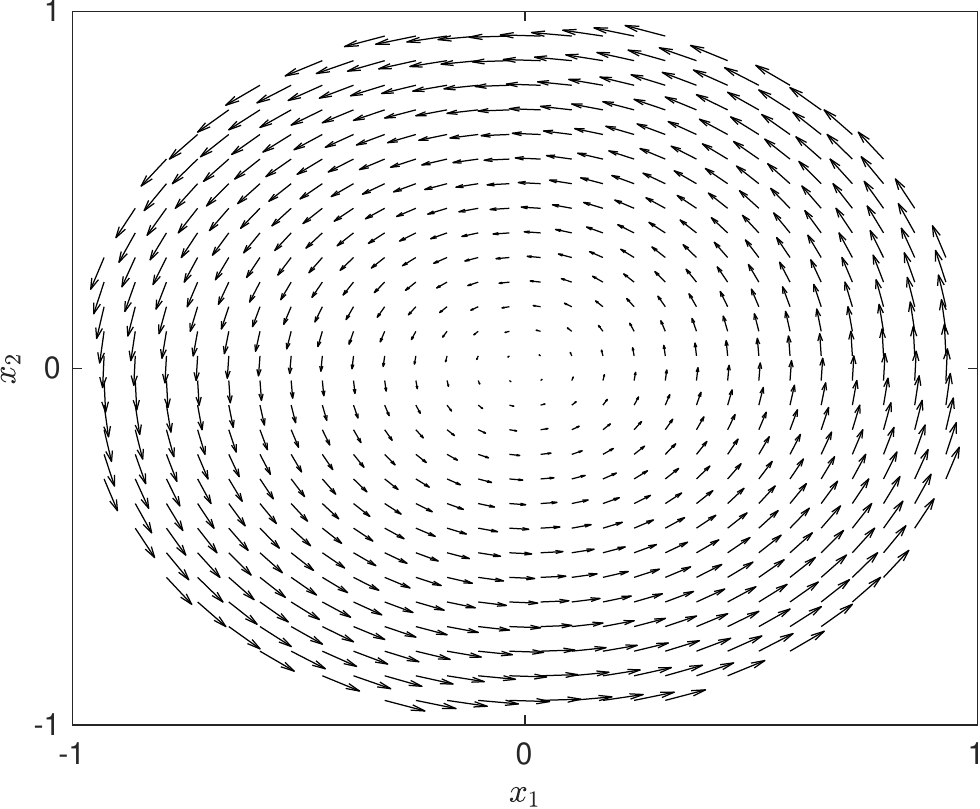}}
\end{center}\vspace{-0.75 cm}
\caption{An Example of Labor Reallocation} \label{f:rotation}
{\scriptsize \vspace{.2 cm} \Cref{f:rotation} shows a counterclockwise example of labor reallocation. The replacement of worker $x = (x_1,x_2)$ is given by $\dot{r}(x) = \theta (- x_2, x_1)$ where $\theta > 0$. For example, the arrow for worker $(1,0)$ points in the direction $(0,1)$. The interpretation is that the worker who replaces worker $(1,0)$ under the new technology comes from the direction $(0,1)$ $-$ which is upwards in the figure and implies that the replacement is more skilled in the second dimension.}
\end{figure}

In order to see this more clearly, consider the two-dimensional case. In two dimensions, labor reallocation being antisymmetric corresponds to reallocation being a rotation of the distribution $\dot{R} x$, where:
\begin{equation}
\dot{R} = \theta \Bigg( \begin{matrix} 0\phantom{-} & -1 \\ 1\phantom{-} & \phantom{-}0 \end{matrix} \Bigg) \;. \label{eq:rotationmat}
\end{equation}
When $\theta > 0$, reallocation is counterclockwise as shown in Figure \ref{f:rotation}. For worker $x = (x_1,x_2)$, the change in assignment is given by $\dot{r}(x) = \theta (- x_2 , x_1)$. In Figure \ref{f:rotation}, the arrow for worker $(1,0)$ points upwards in the direction $(0,1)$. That is, the worker who replaces worker $(1,0)$ under the new technology comes from the direction $(0,1)$ $-$ the replacement is more skilled in the second skill dimension. For worker $(-\frac{1}{2},\frac{1}{2})$, the arrow points in the direction $(-1,-1)$ $-$ the replacement of worker $(-\frac{1}{2},\frac{1}{2})$ is less skilled in both dimensions. When $\theta \leq 0$, reallocation is clockwise. The replacement of worker $x$ under the new technology is given by $\dot{r}(x)$. In short, $\dot{r}(x)$ is the change in the worker type that fulfills job $\tau(x)$.

\subsection{Changes in Earnings} \label{s:changeearnings}

We next return to the general case, and discuss how to determine optimal labor reallocation from the earnings condition. Differentiating the earnings condition (\ref{eq:hedonicw}) with respect to technology:
\begin{align}
(\nabla_1 \dot{\mathsf{y}}_t) (r_t(x) , x) + (\nabla^2_{11} \mathsf{y}_t) (r_t(x) , x) \cdot \dot{r}_t(x) = (\nabla_{1} \dot{w}_t) (r_t(x) ) + (\nabla_{11}^2 w_t) (r_t(x) ) \cdot  \dot{r}_t(x) . \label{e:diffhed}
\end{align}
This condition equates changes in marginal output on the left to changes in marginal labor expenses on the right. Each side has two terms. The \textit{direct effects} give the change in marginal output and marginal labor expenses holding fixed the rearrangement. The \textit{reallocation effects} yield the change in marginal output and labor expenses for each job when they are filled by different workers.

\vspace{0.4 cm}
\noindent \textbf{Changes in Marginal Output}. As technology changes, marginal output changes for two reasons, given by the left-hand side of equation (\ref{e:diffhed}). The first is a direct effect: technological change leads to changes in marginal output even if the same worker continues to work the job: 
\begin{equation}
\dot{\mathcal{A}}_t(x) := (\nabla_1 \dot{\mathsf{y}}_t) (r_t(x) , x)
\end{equation}
The second is a reallocation effect: marginal output changes because the firm is hiring a different worker. This effect is given by $(\nabla^2_{11} \mathsf{y}_t) (r_t(x) , x) \cdot \dot{r}_t(x)$, where $\dot{r}_t(x)$ captures the change in the worker that is hired to fulfill the job that was done by worker $x$ at the initial technology.\footnote{Similar to the firm problem yielding the marginal earnings condition (\ref{eq:hedonicw}), the worker problem yields the marginal firm profit condition $(\nabla_1)(v_t(\tau(x))) = (\nabla_2 y_t)(r_t(x),\tau(x))$. By analyzing the sensitivity of this condition with respect to technological change, $(\nabla_1)(\dot{v}(\tau(x))) = (\nabla_2 \dot{y})(x,\tau(x)) + (\nabla^2_{12} \dot{y})(x,\tau(x)) \dot{r}(x)$, we obtain the change in marginal profits given a change in earnings.} This term reflects the curvature (Hessian) of the output function. When curvature is large, small reallocations lead to large changes in marginal output, while reallocations have little effect on marginal output when curvature is small. The curvature is zero for the bilinear output function.

\vspace{0.4 cm}
\noindent \textbf{Changes in Marginal Labor Expenses}. As technology changes, the marginal expense of the firm on their employee also changes, given by the right-hand side of equation (\ref{e:diffhed}). As with the change in the marginal output, there are two effects on the marginal expense of the firm. The first is a direct effect: even if the same worker stayed in the job, marginal earnings directly respond to technological change. This effect is given by $(\nabla \dot{w}_t) (r_t(x))$. The second is a reallocation effect, the marginal expense of the firm changes because they employ a different worker for the position: when jobs are filled by different workers, the marginal expenses change through the curvature of the earnings function. This effect is given by $(\nabla^{2}w_{t})(r_{t}(x)) \cdot \dot{r}_{t}(x)$. Curvature plays the same role as with marginal output: it measures how marginal earnings change around a worker type. When curvature is large, marginal earnings change significantly when a worker is replaced by a nearby type.


\vspace{0.4 cm}
\noindent \textbf{Distinct Margins of Adjustment}. We next analyze the change in the earnings condition (\ref{e:diffhed}) at the initial technology $t=0$.\footnote{Since the initial reallocation is the identity, $r(x) = x$, the derivative with respect to the reallocation is equal to the derivative with respect to the worker's skill $x$, or $\nabla_1 = \nabla$.} In this case, the change in the earnings condition (\ref{e:diffhed}) is:
\begin{align}
\dot{\mathcal{A}}(x) + (\nabla^2_{11} \mathsf{y}) (x , x) \cdot \dot{r}(x) = (\nabla \dot{w}) (x ) + (\nabla^2 w) (x ) \cdot \dot{r}(x) \label{e:diffhed0}
\end{align}
Equation (\ref{e:diffhed0}) shows that two margins adjust in response to technological change. The first is the marginal earnings change $\nabla \dot{w}$, the second is the reallocation of labor $\dot{r}$. 

The two margins of adjustment are distinct in the following precise sense. The change in marginal earnings $\nabla \dot{w}$ is the gradient of a scalar function, specifically, of the change in earnings $\dot{w}$. We denote the space of vector fields which are gradients by $\mathcal{G}$. Reallocations, in contrast, belong to the space of vector fields which are divergence free $\mathcal{D}$. A key property is that these two spaces are orthogonal.\footnote{Specifically, the two spaces are orthogonal with respect to the inner product, $\langle a, b \rangle = \int a(x) b(x) f(x)dx = 0$ for every $a \in \mathcal{G}$ and $b \in \mathcal{D}$. In order to see this, note that the gradient $(\nabla w) \in \mathcal{G}$ and reallocation $\dot{r} \in \mathcal{D}$. Integrating by parts yields $\int_X (\nabla w)(x) \dot{r}(x) f(x) dx = \int_{\partial X} w(x) \dot{r}(x) f(x) \cdot n(x) dx - \int_X w(x) \nabla \cdot (\dot{r}(x) f(x)) dx = 0$ where the final equality comes from the reallocation restrictions (\ref{eq:rearr0}) and (\ref{eq:rearr0b}).} Orthogonality implies that technological change affects equilibrium along two distinct margins: changes in marginal earnings are a gradient and reallocation is divergence-free. 


\vspace{0.4 cm}
\noindent \textbf{Complementarity}. In equation (\ref{e:diffhed0}), reallocation affects changes in marginal output and marginal labor expenses. We equivalently can combine these two effects into one $-$ the complementarity between a worker and the replacement worker. We define \textit{complementarity} as:
\begin{equation}
\mathcal{C}(x) := (\nabla_{12}^2 \mathsf{y}) (x , x) . \label{eq:C0}
\end{equation}
Complementarity at skills $x$ is given by the matrix of cross-derivatives in the production function. Diagonal elements of the complementarity matrix capture complementarities within skills. They measure how the marginal product of a given skill for the replacement worker changes as the job changes to one that is fulfilled by a worker with a marginally higher level of that same skill under the initial output function. Non-diagonal terms capture complementarities between different skills. They measure how marginal output of a given skill for the replacement worker changes as the job is changed to one that is fulfilled by a worker with a marginally higher level of a different skill under the initial technology. 

In order to see the combined effect of reallocation on the change in the earnings condition (\ref{e:diffhed0}), differentiate the marginal earnings condition (\ref{eq:hedonicw}) at the initial technology $( (\nabla_1 \mathsf{y}) (x , x) = (\nabla w) (x ) )$ with respect to worker skill $x$: 
\begin{equation}
(\nabla_{11}^2 \mathsf{y}) (x , x) - (\nabla^2 w) (x ) = - \; \mathcal{C}(x) .\label{eq:difffocwrtx}
\end{equation}
The symmetric left-hand side captures the second-order loss in profits with respect to changes in the worker type employed at the firm.\footnote{Since the cross-derivatives of earnings $\nabla^2 w$ and the technology $\nabla_{11}^2 \mathsf{y}$ are symmetric, $\mathcal{C}$ is symmetric.} When the left-hand side is more negative semidefinite, profits decrease more when moving away from the optimum worker. The second-order condition to the firm problem implies the left-hand side in equation (\ref{eq:difffocwrtx}) is negative semidefinite, implying $\mathcal{C}(x)$ is positive semidefinite for every $x$. In fact, it is positive definite under our assumptions.\footnote{See, e.g., identity (12.2) in \citet{Villani:2009}. The right-hand side of the identity is strictly positive by the twist condition, because the optimal assignment is invertible, and since $f$ and $g$ are strictly positive. The left-hand side, the determinant of the complementarity matrix, is therefore also strictly positive.} This shows that $\mathcal{C}_f(x) := f(x) \mathcal{C}^{-1}(x)$ is also symmetric and positive definite. Equation (\ref{eq:difffocwrtx}) shows that profits decrease more when complementarity is large. In sum, the change in the marginal earnings condition (\ref{e:diffhed}) is given by:
\begin{align}
\dot{\mathcal{A}}(x) = (\nabla \dot{w}) (x ) +  \mathcal{C}(x)  \cdot \dot{r}(x) .  \label{e:equation2}
\end{align}
 
\subsection{Solving for Reallocation and Earnings Changes}  \label{s:solvingrealloc}
 
The previous sections establish two results. First, labor market clearing restricts reallocation to be divergence-free and parallel to the boundary (\Cref{s:feasiblereall}). Second, marginal earnings changes are a gradient, and the firm profit maximizing condition connects marginal earnings changes to reallocation (\Cref{s:changeearnings}). This structure allows us to represent equilibrium adjustments along two distinct margins $-$ earnings and reallocation $-$ that are orthogonal to each other. This section completely characterizes comparative statics for multidimensional assignment models.
 
The marginal earnings condition (\ref{e:equation2}) has two unknowns: earnings change $\dot{w}$ and reallocation $\dot{r}$. In order to characterize these unknowns, we show that technological change $\dot{\mathcal{A}}$ must be decomposed into two distinct components: a gradient that captures changes in marginal earnings, and a divergence-free reallocation. The relative relevance of these two components for each worker type is governed by the complementarity matrix. 

 
\vspace{0.4 cm}
\noindent We next establish the main theoretical results of the paper.

\begin{theorem}{\textit{Structure of Comparative Statics: the Helmholtz Decomposition of Technological Change}.} \label{eq:helmholtz}
For every technological change $\dot{\mathcal{A}}$, there exists a unique decomposition $\dot{\mathcal{A}}(x) = (\nabla \dot{w}) (x ) + \mathcal{C}(x) \cdot \dot{r}(x)$ into marginal earnings changes $(\nabla \dot{w}) \in \mathcal{G}$ and labor reallocation $\dot{r} \in \mathcal{D}$. 
\end{theorem}

\noindent Theorem \ref{eq:helmholtz} shows that any technological change is uniquely decomposed into two orthogonal components: a gradient component which corresponds to earnings changes, and a divergence-free component which corresponds to reallocation. This is the Helmholtz decomposition which states that any vector field can be written as the sum of a gradient component and divergence-free component (see, e.g., \citet{Arfken:2011}).\footnote{The Helmholtz decomposition is known to be unique. For completeness, we present a proof in Appendix \ref{a:helmholtz}.}


We next show how the Helmholtz decomposition of technological progress is obtained. Specifically, we show that the Poisson equation gives a direct way to characterize earnings changes. Once worker earnings are solved for, reallocation follows residually.

\begin{theorem}{\textit{Implementation of Comparative Statics: the Poisson Equation}.} \label{eq:changeinmearnings}
For any technological change $\dot{\mathcal{A}}$, the change in earnings $\dot{w}$ is uniquely determined as the solution to a Poisson equation:
\begin{align}
\nabla \cdot \big( \mathcal{C}_f(x) \nabla \dot{w} (x ) \big) & = \nabla \cdot \big( \mathcal{C}_f(x)  \dot{\mathcal{A}} (x) \big) \hspace{2.84 cm} \text{ for  $x \in X$} \label{e:csgen1} \\
\big( \mathcal{C}_f(x) \dot{\mathcal{A}}(x) \big) \cdot n(x) & = \big(  \mathcal{C}_f (x) \nabla \dot{w} (x ) \big) \cdot n(x)  \hspace{2.05 cm} \text{ for $x \in \partial X$} \label{e:csgen2} 
\end{align}
Given the unique characterization of the change in marginal earnings, the optimal reallocation is uniquely given using equation (\ref{e:equation2}).
\end{theorem}



\vspace{0.1 cm}
\noindent In order to prove Theorem \ref{eq:changeinmearnings}, we first rewrite the change in marginal earnings in equation (\ref{e:equation2}) by multiplying by the inverse of the complementarity matrix, and the density of worker skills $f$:
\begin{align*}
f(x) \dot{r}(x) = \mathcal{C}_f(x) \dot{\mathcal{A}}(x) - \mathcal{C}_f(x) \nabla \dot{w} (x) \;.  
\end{align*}
Recall that we established that $\dot{r}(x) f(x)$ is divergence free and parallel to the boundary. By taking divergences with respect to the previous equation, and inner products with respect to the outward normal vector at the boundary, we obtain the Poisson equation for the change in earnings $\dot{w}$ establishing Theorem \ref{eq:changeinmearnings}.\footnote{For the paper to be self-contained, we present well-known results on existence and uniqueness of the solution to the Poisson equation in Appendix \ref{a:poissoneq}.} By solving for the reallocation $\dot{r}(x)$, we characterize the worker that replaces $x$ under technology $t$ as $r_t(x) \approx r(x) + t \dot{r}(x) = x + t \dot{r}(x)$. 


\vspace{0.4 cm}
\noindent We now provide an alternative characterization of earnings changes as a solution to a quadratic optimization problem. This characterization enables our numerical analysis in Section \ref{s:quant}.

\begin{theorem}{\textit{Characterization of Earnings Changes}.} \label{t:alternchar}
The earnings change $\dot{w}$ solves the following quadratic minimization problem:
\begin{equation}
\min_{\varphi} \int ( \dot{\mathcal{A}}(x) - \nabla \varphi(x) )^\top \mathcal{C}_f(x) ( \dot{\mathcal{A}}(x) - \nabla \varphi(x) ) dx  \label{eq:quadproblem}
\end{equation}
That is, $\nabla \dot{w} = \nabla \varphi^*$. Given the change in marginal earnings, optimal reallocation is given using equation (\ref{e:equation2}).
\end{theorem}

\noindent Equilibrium reallocation $\dot{r}$ and earnings changes $\dot{w}$ are characterized by finding $\varphi$ that solves the quadratic problem (\ref{eq:quadproblem}). Intuitively, since the Poisson equation is linear in marginal earnings, we find an optimization problem that is quadratic with respect to marginal earnings that delivers the Poisson equation as an optimality condition. The proof is in Appendix \ref{a:alternativechar}. 

The quadratic minimization problem for earnings (\ref{eq:quadproblem}) in Theorem \ref{t:alternchar} can be interpreted as the \textit{gradient regression of technological change}. The objective minimizes the sum of squared differences between technological change and a gradient. Rather than fitting a line (as in the case of linear regression), or any other particular function (as in the case of nonlinear regression), the function is restricted to be a gradient. In other words, the gradient regression gives the best gradient fit. 


Theorem \ref{t:alternchar} is also useful to further describe the role of complementarities and the worker skill distribution in the passthrough of technological change to worker earnings and allocations. Weak complementarities tilt equilibrium adjustments toward reallocation since firm output is not sensitive to changes in the worker type at the firm, whereas strong complementarities tilt equilibrium adjustments toward earnings changes. Regions of the worker type space with many workers have strong earnings responses, while regions with few workers feature more reallocation as reallocation induces a smaller loss in the quadratic objective (\ref{eq:quadproblem}). When workers are abundant, residuals are costly and marginal earnings adjusts more strongly to reflect technological change. In sum, common worker types experience earnings responses, while rare worker types experience reallocation.

\vspace{0.5 cm}
\noindent \textbf{Discussion}. Theorems \ref{eq:helmholtz} and \ref{eq:changeinmearnings} provide both the structure and the solution to comparative statics for multidimensional sorting models. Theorem \ref{eq:helmholtz} shows that the structure of the comparative statics is given by the Helmholtz decomposition: technological change separates into two orthogonal, and thus distinct components, an earnings gradient and a divergence-free reallocation. Theorem \ref{eq:changeinmearnings} shows this decomposition is implemented through a Poisson equation. The earnings change is the solution to the Poisson equation, with reallocations determined residually. These results provide a full characterization of equilibrium comparative statics for the multidimensional assignment model for an arbitrary change in technology. In Appendix \ref{a:distributionalch}, we show that we can extend our characterization of comparative statics in Theorems \ref{eq:helmholtz} and \ref{eq:changeinmearnings} to incorporate changes in the distributions of workers and jobs. 

\subsection{Symmetric and Antisymmetric Technological Change} \label{s:symmetry}

This section develops the second main result of the paper. We show that the symmetry properties of technological change determine whether earnings, allocations or both change in equilibrium. When technological change is symmetric, there is a full passthrough to earnings. When technological change is antisymmetric, there is only reallocation. In the general case, technological progress results in changes along both margins, with the importance of each margin determined by production complementarities and the distribution of worker skills.

\vspace{0.4 cm}
\noindent \textbf{Symmetric Technological Change: Full Passthrough into Earnings}. First, suppose that technological change is symmetric, that is, $\frac{\partial \dot{\mathcal{A}}_i}{\partial x_j}(x) = \frac{\partial \dot{\mathcal{A}}_j}{\partial x_i}(x)$ for all $i,j$ and for every $x$.

\begin{proposition}{\textit{Symmetric Technological Change $-$ Full Passthrough into Worker Earnings}}. \label{p:symmetricse}
If technological change $\dot{\mathcal{A}}$ is symmetric then there is no reallocation and technological change fully passes through to earnings.
\end{proposition}

\noindent The symmetry condition requires that cross-effects of technological change across different skills coincide. Formally, this means that the technological change is the gradient of a function, that is, there exists a function $\zeta$ such that $\dot{\mathcal{A}} = \nabla \zeta$. In this case, the change in equilibrium is given only by changes in earnings without any reallocation: $\nabla \dot{w} = \nabla \zeta$ and $\dot{r} = 0$ satisfies equation (\ref{e:equation2}). In other words, there is full passthrough of technological change into earnings. The complementarity matrix $\mathcal{C}$ and the skill distribution $f$ are thus irrelevant in this case and matter only when part of the technological change cannot directly be absorbed into earnings. That is, the complementarity matrix and skill distribution only affect how technological change is divided between earnings and reallocation when technological progress is not symmetric.


A classic example of symmetric technological change is Hicks-neutral technological change that scales output by a constant factor, for example, $\mathsf{y}_t(\tilde{x},x) = A_t \mathsf{y}(\tilde{x},x)$ where $A_t > 0$. Technological change scales output of any given worker-job pair by the same constant $A_t$, and hence scales aggregate output by the same constant $A_t$, without changing tradeoffs. It thus follows immediately that the assignment that maximizes aggregate output is identical for all $A_t > 0$. 


\vspace{0.4 cm}
\noindent \textbf{Antisymmetric Technological Change: Full Passthrough into the Assignment}. Suppose next that complementarity-weighted technological change $\dot{\mathcal{A}}_{\mathcal{C}} := \mathcal{C}^{-1}\dot{\mathcal{A}}$ is antisymmetric, meaning that $\frac{\partial \dot{\mathcal{A}}_{\mathcal{C} i}}{\partial x_j}(x) = -\frac{\partial \dot{\mathcal{A}}_{\mathcal{C} j}}{\partial x_i}(x)$ for all $i,j$ and for every $x$. 

\begin{proposition}{\textit{Antisymmetric Technological Change $-$ Full Passthrough into Reallocation}}.
If the complementarity-weighted technological change $\dot{\mathcal{A}}_{\mathcal{C}}$ is antisymmetric, then there are no earnings changes and technological change fully passes through to reallocation.
\end{proposition}

\noindent Formally, the antisymmetry condition implies that the technological change $\dot{\mathcal{A}}(x) = \mathcal{C}(x) \rho(x)$, for some $\rho(x)$ such that $\rho(x)f(x)$ is divergence-free and parallel to the boundary. The Helmholtz decomposition of technological progress then results in no change in earnings, and changes sorting as $\dot{r} = \rho$. Complementarities $\mathcal{C}$ and the worker distribution $f$ govern how technological change translates into reallocation for different worker types. In this case, the technological change fully passes through to sorting, leaving earnings unaffected. This result sharply contrasts to the symmetric case, where production complementarities and the worker skill distribution are irrelevant. 

\vspace{0.4 cm}
\noindent \textbf{General Technological Change: Changes in Earnings and Assignments}. The symmetric and antisymmetric cases are the two polar cases that isolate the main forces of equilibrium comparative statics. More generally, technological progress has both symmetric and antisymmetric components and thus changes both earnings and allocations.

Consider a linear combination of the two previous cases and suppose that technological change $\dot{\mathcal{A}}$ can be written as $\dot{\mathcal{A}}(x) = \nabla \zeta (x) + \mathcal{C}(x) \rho(x)$, where $\rho(x)f(x)$ is divergence-free and parallel to the boundary. The marginal earnings condition (\ref{e:equation2}) holds when earnings change as $\dot{w} = \zeta$ and the assignment changes as $\dot{r} = \rho$. Alternatively phrased, we take a linear combination of the previous two cases such that technological change is absorbed both through earnings and the assignment. Theorem \ref{eq:helmholtz} shows that technological change separates uniquely into an earnings adjustment and reallocation and Theorem \ref{eq:changeinmearnings} solves for these changes. The earnings change reflects the symmetric part of technological change, whereas reallocation reflects the antisymmetric part. Production complementarities and the skill distribution govern the extent to which technological change is priced into earnings and how much is absorbed by reallocation.  


An example of such general technological change is automation, which increases the output of cognitive skills in more manual occupations (which now value more programming skills) and decreases the productivity of manual skills in more cognitive occupations (due to the increased automation of manual tasks). This component of automation is antisymmetric: it increases the marginal product of cognitive skills in more manual occupations while it decreases the marginal product of manual skills in more cognitive occupations. In response to automation both earnings and the allocation change and the Helmholtz decomposition dictates by how much.

We next continue developing intuition for the Helmholtz decomposition of technological change. First, by providing a closed-form characterization for bilinear output functions and $d$-dimensional rotationally invariant worker skill distributions (\Cref{s:closedform}). Second, we characterize reallocation and earnings responses to technological change for an empirically-relevant economy (\Cref{s:quant}). 


\section{Bilinear Technology and Sylvester Equation} \label{s:closedform}


We provide a detailed closed-form example to demonstrate \Cref{eq:changeinmearnings} and to highlight the role of complementarities in determining comparative statics. \Cref{s:bilinear} specifies a bilinear production technology. \Cref{s:compsylvester} establishes that the Poisson equation of \Cref{eq:changeinmearnings} reduces to a Sylvester equation. \Cref{s:twodimex} illustrates the results in two dimensions.


\subsection{Bilinear Technology} \label{s:bilinear}

Consider a bilinear technology $\mathsf{y}_t(\tilde{x},x) = \tilde{x}^\top \Sigma_t x$, where $\Sigma_t$ is the complementarity matrix.\footnote{Additional linear terms $A_t\tilde x+B_tx$ can immediately be incorporated. Since the additional terms depend only on distributions of workers and jobs $-$ not on their pairing $-$ the added terms do not affect sorting. The change of the optimal assignment is thus not affected by these terms, while the dual solutions change by explicit additive terms. For example, if technological progress changes the linear worker term as $\dot A_t \tilde{x}$, then earnings change at the same rate.} Output is linear in each input and depends on pairwise complementarities. Diagonal elements capture within-skill complementarity, while off-diagonal elements capture between-skill complementarity. From the profit maximization problem, it follows that $\Sigma$ is symmetric and positive semidefinite.\footnote{The first-order condition to the profit maximization problem (\ref{eq:hedonicw}) is $(\nabla_1 \mathsf{y}) (x , x) = \Sigma x = (\nabla w)(x)$, implying $\Sigma x$ is a gradient of a function. Since the second-order derivatives $(\nabla^2 w)$ are continuous, the mixed derivatives are the same, which implies that $\Sigma$ is symmetric. From the second-order condition to the firm problem $(\nabla^2_{11}  \mathsf{y}) (x , x) - (\nabla^2 w)(x)$ is negative semidefinite implying that $(\nabla^2 w)(x) = \Sigma$ is positive semidefinite since $(\nabla^2_{11}  \mathsf{y}) (x , x)=0$.} Technological change augments the complementarity matrix as $\dot{\Sigma}$. We further assume that the skill distribution is rotationally invariant over a $d$-dimensional sphere.

This setting delivers two results. First, we show that the Helmholtz decomposition reduces to a split of the matrix of technological change $\dot{\Sigma}$ into an antisymmetric and a symmetric component. Second, we show that reallocation is pinned down by a Sylvester equation.



\subsection{Comparative Statics and Sylvester Equation} \label{s:compsylvester}

The Helmholtz decomposition of technological change (\ref{e:equation2}) decomposes technological change into a gradient which corresponds to the change in earnings $(\nabla \dot{w})$ and a divergence-free reallocation $\dot{r}$. With a bilinear technology, the production complementarity in equation (\ref{eq:C0}) is given by $\mathcal{C}(x) = \Sigma$ and technological change is linear $\dot{\mathcal{A}} (x )=\dot{\Sigma}x$. The Helmholtz decomposition is then given by:
\begin{equation}
\dot{\Sigma} x = \nabla \dot{w} (x) + \Sigma \dot{r}(x)  .\label{eq:intuition0by}
\end{equation}

We suppose that reallocation is linear in worker skills $\dot{r}(x) = \dot{R} x$, where $\dot{R}$ is antisymmetric. This reallocation is feasible by Lemma \ref{l:antisymmetrics}. Since both the technological change and reallocation are linear in worker skills, the change in marginal earnings is also necessarily linear $(\nabla \dot{w}) (x) = \dot{W} x$, where the matrix governing the change in earnings $\dot{W}$ is symmetric. The Helmholtz decomposition (\ref{eq:intuition0by}) becomes $\dot{\Sigma} x = \dot{W} x + \Sigma \dot{R} x$. Since the Helmholtz decomposition is unique, it follows that $\dot{\Sigma} = \dot{W} + \Sigma \dot{R}$. Proposition \ref{p:sylvester} shows that the Helmholtz decomposition can then be fully characterized using the Sylvester equation.\footnote{The Sylvester equation has appeared in a range of macroeconomic models: to solve linear-quadratic dynamic programs \citep{Anderson:1996}, to characterize the stationary covariance matrix of a linear economy \citep{Fernandez:2007}, and recently in \citet{Bilal:2023a} to solve heterogeneous macro models and in \citet{Bilal:2023b} to solve climate change models.} 

\begin{proposition}{\textit{Comparative Statics and Sylvester Equation}}. \label{p:sylvester}
Suppose the technology is bilinear and the worker skill distribution is rotationally invariant on a $d$-dimensional sphere. Then, labor reallocation is determined by $\dot{R}$, which is the unique solution to the Sylvester equation
\begin{align}
\Sigma \dot{R} + \dot{R} \Sigma & = \dot{\Sigma} - \dot{\Sigma}^\top \hspace{0.03 cm}, \label{eq:syl1} \end{align}
and the corresponding change in the slope of marginal earnings is determined as $\dot{W} = \dot{\Sigma} - \Sigma \dot{R}$.
\end{proposition}

\noindent The proof uses that the matrix governing earnings changes is symmetric, $\dot{W} = \dot{\Sigma} - \Sigma \dot{R} = \dot{W}^\top = \dot{\Sigma}^\top - \dot{R}^\top \Sigma^\top = \dot{\Sigma}^\top + \dot{R} \Sigma$, where the final equality uses that  $\dot{R}$ is antisymmetric and $\Sigma$ is symmetric. Reorganizing this expression gives the Sylvester equation (\ref{eq:syl1}). We give a direct proof of existence and uniqueness of the solution to the Sylvester equation in Appendix \ref{a:sylvester}.


Proposition \ref{p:sylvester} implies that reallocation is determined by changes in complementarity between skills (the off-diagonal elements of $\dot{\Sigma}$). In contrast, changes in the complementarity within skills (the diagonal elements of $\dot{\Sigma}$) affect only earnings, not reallocation. This follows since the Sylvester equation does not depend on the diagonal elements of the matrix of technological change $\dot{\Sigma}$, which captures the change in the complementarity within skills, due to $\dot{\Sigma}-\dot{\Sigma}^\top$. Specifically, differential change in between-skill complementarities leads to reallocation $\dot{\Sigma} - \dot{\Sigma}^\top$ since the only non-zero elements of this matrix are off-diagonal and, moreover, those that capture the differential change in between-skill complementarity. The change in within-skill complementarity only affects worker earnings not the allocation of labor. 

In the special case where the complementarity matrix is given by the identity matrix $(\Sigma = I)$ reallocation is determined by $\dot{R} = \frac{1}{2} ( \dot{\Sigma} - \dot{\Sigma}^\top )$, and the change in marginal earnings is determined as $\dot{W} = \frac{1}{2} ( \dot{\Sigma} + \dot{\Sigma}^\top )$.  Earnings changes are determined by the symmetric part of the change in the complementarity matrix $\dot{\Sigma}$ which includes both the diagonal elements, which reflect within-skill complementarity, and the averages of non-diagonal elements, $\frac{1}{2}(\dot{\Sigma}_{ij}+\dot{\Sigma}_{ji})$, which reflect between-skill complementarity. In contrast, the antisymmetric part, $\frac{1}{2}(\dot{\Sigma}_{ij} - \dot{\Sigma}_{ji})$, represents between-skill changes in complementarity that cannot be captured in earnings and instead drive reallocation.\footnote{In \Cref{a:intuitionidentitysylvester}, we discuss the case where the complementarity matrix is the identity matrix for an economy with two dimensions of worker and job characteristics in more detail.}

The Sylvester equation simplifies comparative statics. The Helmholtz decomposition of Theorem \ref{eq:helmholtz} shows that technological change can be split into an earnings component and a reallocation component, characterized by a Poisson equation. With a bilinear production function and a rotationally invariant worker skill distribution, this structure reduces to a simple linear algebra condition: symmetric technological change translates into earnings changes, antisymmetric technological change translates into reallocation. 



\subsection{Two-Dimensional Example} \label{s:twodimex}

We continue the illustration of \Cref{eq:helmholtz} and \Cref{eq:changeinmearnings} using Proposition \ref{p:sylvester} by characterizing comparative statics for a two-dimensional setting. In this setting, the output produced by worker $\tilde{x}$ when they work in the job that is done by worker $x$ under the initial technology is: 
\begin{equation}
\mathsf{y}_t(\tilde{x}, x) = \alpha_{t} \tilde{x}_{1} x_1 + \beta_{t} \tilde{x}_1 x_2 + \gamma_{t} \tilde{x}_2 x_1 + \delta_{t} \tilde{x}_2 x_2 \;.\label{e:bilinearoutputworkers}
\end{equation}
The coefficients $\alpha_{t}$ and $\delta_t$ capture within-skill complementarities, whereas the coefficients $\beta_t$ and $\gamma_t$ capture the between-skill complementarities between workers and their job. Technological change is given by:
\begin{equation*}
\dot{\Sigma}= \bigg( \begin{matrix}\dot{\alpha} & \dot{\beta}\\
\dot{\gamma} & \dot{\delta}  
\end{matrix} \bigg) \;.
\end{equation*}
The diagonal elements $\dot{\alpha}$ and $\dot{\delta}$ capture changes in within-skill productivity, whereas non-diagonal elements $\dot{\beta}$ and $\dot{\gamma}$ capture changes in between-skill productivity.


We characterize reallocation in response to technological change using Proposition \ref{p:sylvester}. As discussed in Section \ref{s:comparativestats}, reallocation being antisymmetric corresponds to reallocation being a rotation (\ref{eq:rotationmat}) in two dimensions. Evaluating the Sylvester equation (\ref{eq:syl1}), $\Sigma \dot{R} + \dot{R} \Sigma = (\alpha + \delta) \dot{R} = \dot{\Sigma} - \dot{\Sigma}^\top$, we obtain:
\begin{equation*}
 \theta = \frac{\dot{\gamma} - \dot{\beta}}{\alpha + \delta} \;,
\end{equation*}
which is the angle of rotation parameter that determines reallocation as $\dot{r}(x) = \dot{R} x = \theta(x_2,-x_1)$.\footnote{In the two-dimensional environment, $\Sigma$ being symmetric is equivalent to the between-task complementarities being identical $\gamma = \beta$, while $\Sigma$ is positive semidefinite implies that the trace $\alpha + \delta$ and the determinant $\alpha \delta - \beta \gamma$ are positive. For the rearrangement problem to satisfy the twist condition (see Section \ref{s:assumptions}), the complementarity matrix has to be invertible, which restricts $\alpha \delta > \beta \gamma = \gamma^2$ such that both $\alpha$ and $\delta$ are strictly positive.} Labor reallocates when technological change affects the between-skill complementarities asymmetrically $(\dot{\gamma} \neq \dot{\beta})$. The \textit{direction} of reallocation is determined by the relative technological change in between-task complementarities, that is, $\theta \geq 0$ if and only if $\dot{\gamma} \geq \dot{\beta}$. An increase in one between-skill complementarity relative to the other rotates workers to exploit increased complementarity.\footnote{In order to build intuition for the angle of rotation, suppose the complementarity between a worker's skill $x_1$ and their replacement's skill $\tilde{x}_2$ $-$ captured by $\gamma_t$ in the worker-worker complementarity matrix $\Sigma_t$ $-$ grows large. When the reassignment problem with technology (\ref{e:bilinearoutputworkers}) is dominated by $\gamma_t \tilde{x}_2 x_1$, the optimal assignment positively sorts worker skills $x_1$ and the replacements' skills $\tilde{x}_2$. This is the counterclockwise reallocation in Figure \ref{f:rotation}.} The magnitude of this rotation is dampened by the sum of within-skill complementarities. When $\alpha+\delta$ is large, reallocation is limited. When $\alpha+\delta$ is small, small technological changes may induce large reallocations. Within-skill complementarity thus stabilizes sorting.

How does technological change pass through into marginal earnings? Using Proposition \ref{p:sylvester}, the change in the slope of marginal earnings is determined as:
\begin{equation*}
\dot{W} = \dot{\Sigma} - \Sigma \dot{R} = \Bigg( \begin{matrix} \dot{\alpha} - \frac{\beta}{\alpha + \delta}(\dot{\gamma} - \dot{\beta}) & & & \frac{\alpha}{\alpha + \delta}\dot{\gamma}  + \frac{\delta}{\alpha + \delta}\dot{\beta} \\  \frac{\alpha}{\alpha + \delta}\dot{\gamma}  + \frac{\delta}{\alpha + \delta}\dot{\beta} & & &  \dot{\delta} -  \frac{\gamma}{\alpha + \delta} (\dot{\gamma} - \dot{\beta})\end{matrix} \Bigg) \;. \label{eq:intuitionw2}
\end{equation*}
The diagonal elements of $\dot{W}$ capture the effect of changes in within-skill complementarities $\dot{\alpha}$ and $\dot{\delta}$ on earnings. An increase in within-skill complementarity increases returns to the corresponding skill. The off-diagonal entry reflects a weighted average of technological progress in between-skill complementarities, which stretches the earnings for workers with high-skill in both dimensions. These earnings changes are dampened by the sum of within-skill complementarities. In the special case where $\dot{\gamma} = \dot{\beta}$, technological change is symmetric and there is no reallocation and there is complete passthrough of technological change into earnings (Proposition \ref{p:symmetricse}).

\section{Quantitative Analysis} \label{s:quant}

In this section, we develop a quantitative illustration of the model. We conduct a quantitative evaluation of reallocation and earnings changes in response to technological change. We do so for an empirically-relevant economy where we infer the distribution of worker skills and the technology using data on earnings and occupations for all U.S. workers (\Cref{s:data} and \Cref{s:calibration}).\footnote{For ease of presentation, we adopt the pointwise identification approach in \citet{BK2:2020,BK1:2021}. \citet*{Kim:2026} develop a semi-nonparametric econometric approach to estimate the assignment model with bilinear production and general distributions of workers and jobs. They prove identification of the production function, equilibrium earnings, and assignment using optimal transport theory and propose efficient sieve estimators. One could similarly apply our tools to study the labor market response to cognitive skill-biased technological change at their estimated equilibrium.

The general approach is thus as follows. First, use empirical outcomes to observe where individuals work and what they earn under the current equilibrium to understand production technologies and input distributions. Second, use a technological perturbations that holds constant the worker, such as a randomized controlled trial, to learn the direction of a technological change. Third, use the Helmholtz decomposition of this technological change to understand how technological change will affect equilibrium sorting and wages.
}  We refer to the first skill dimension as manual skill $(x_1 = x_m)$ and to the second skill dimension as cognitive skill $(x_2 = x_c)$. In Section \ref{s:counterfact}, we use the quantitative model to illustrate how the labor market responds to cognitive skill-biased technological change.


\subsection{Data} \label{s:data}

We use data from the American Community Survey (ACS) for all individuals between 25 and 60 years of age. The final sample includes about 16 million individuals between 2000 and 2019. Our measure of earnings is wage and salary income before taxes over the past 12 months.


The ACS contains occupational information for every worker. We combine a worker with the task intensity for their occupation using O*NET task measures from \citet{Acemoglu:2011}. For identification, we construct task intensity by occupation exactly as in \citet{BTZ2:2022}. To obtain aggregated task intensity we use a Cobb-Douglas technology to map the intensity of subtasks into a final task intensity measure similar to \citet{Acemoglu:2011} and \citet{Deming:2017}:  
\begin{equation*}
q_s = \exp \bigg( \frac{1}{|\mathcal{V}_s|}\sum\limits_{\nu \in \mathcal{V}_s} \log q_{s \nu} \bigg) .
\end{equation*}
Letting $\log q_{s \nu}$ be the $Z$-score by subtask $\nu \in \mathcal{V}_s$, we construct task intensity levels $q_s$ for cognitive and manual skills $s \in \{m, c\}$. The resulting task intensity levels are approximately lognormally distributed across occupations (see Figure 1 and Figure 2  in \citet{BTZ2:2022}). We make an identification assumption that relative task intensity is equal to the relative skill level, $x_m/x_c = q_m/q_c$, which is hence also approximately lognormally distributed across occupations.


\subsection{Calibration} \label{s:calibration}

We now identify the worker skill distribution pointwise. The identification argument is similar to \citet{BK2:2020,BK1:2021} who use explicit characterizations of home production models to identify home productivity as well as permanent and transitory market productivity using data on consumption expenditures and time allocation.

Using the O*NET task measures, we have information on the relative task intensity for each worker $x_m/x_c = q_m/q_c$. In order to determine the level of worker skills, we use earnings equation (\ref{eq:hedonicw}) for the bilinear economy at the initial technology to write:
\begin{equation}
2 w(x) = x^\top \Sigma x = \alpha x_c^2 + 2 \beta x_c x_m + \delta x_m^2  = x_c^2 \left[ \alpha  + 2 \beta \frac{x_m}{x_c} + \delta \Big( \frac{x_m}{x_c} \Big)^2 \right] \label{e:earningsdata}
\end{equation}
where the second equality uses that the worker-worker complementary matrix $\Sigma$ is symmetric.\footnote{Since earnings are equal to half of the total surplus under the initial technology it follows that firm profits also equal half the surplus, that is, $v(z) = \frac{1}{2} x^\top \Sigma x$ for $z = \tau x$.}

Given the skill ratio for an individual's occupation, $\frac{q_m}{q_c}$, and an individual's earnings $w$, this equation uniquely characterizes the cognitive skills $x_c$ and the manual skills $x_m$ as:
\begin{equation}
x_c^2 = \frac{2 w(x)}{\alpha  + 2 \beta \frac{q_m}{q_c} + \delta \big( \frac{q_m}{q_c} \big)^2} \hspace{2 cm} \text{ and } \hspace{2 cm} x_m^2 = \frac{2 w(x)}{\alpha \big( \frac{q_c}{q_m} \big)^2  + 2 \beta \frac{q_c}{q_m} + \delta} \;. \label{e:xcxm}
\end{equation}
The key ingredients to our inference of the worker skill distribution are the quadratic earnings equation and the identification assumption that the relative skill level is a known function of the relative task intensity, $x_m/x_c = q_m/q_c$. We use equation (\ref{e:xcxm}) together with data on earnings and occupations to infer the worker skills distribution in the ACS.

\begin{table}[t!]
\def\arraystretch{1.4}%
\begin{center}
\caption{Example of Inference of Skills}\label{t:simple_example}
\begin{tabular}{clcccc}
\hline  \hline
 & & \multicolumn{1}{c}{Relative Skills} &  \multicolumn{1}{c}{Earnings} & \multicolumn{2}{c}{Task Intensity}    \\
 & & \hspace{0.85 cm}  $q_m/q_c$ \hspace{0.85 cm}  & \hspace{0.85 cm} $ w(x)$ \hspace{0.85 cm} & \hspace{1.10 cm}  $x^2_c$ \hspace{1.10 cm}  & \hspace{1.10 cm}  $x^2_m$ \hspace{1.10 cm}  \\
\hline
1 & \; Baseline		   			  & 		1 & 		$1/2$ & 		  $1/(\alpha + \delta)$     & $1/(\alpha + \delta)$   \\
2 & \; Earnings			  			  &		1 & 		1 & 	  $2/(\alpha + \delta)$     &  $2/(\alpha + \delta)$    \\ 
3 & \; Task intensity \hspace{0.10 cm} 	   &		2 & 		$1/2$ &   	$1/(\alpha + 4 \delta)$     &		          	$4/(\alpha + 4 \delta)$  \\ 
\hline \hline
\end{tabular}\end{center}
{\scriptsize \Cref{t:simple_example} illustrates the identification of workers' manual and cognitive skills through three examples. We infer higher levels of manual skills with higher earnings (Row 2), and higher manual task intensity (in Row 3).}
\end{table}

\vspace{0.4 cm}
\noindent \textbf{Examples}. In order to provide insight into the inference of worker skills, we consider a numerical example. We first consider an economy with no cross-skill complementarity, or $\beta = 0$.\footnote{When we estimate the parameters of the production technology (see \Cref{t:modelparam}) we find $\beta$ to be relatively small.}

Suppose a worker's occupational relative task intensity is one, $\frac{q_m}{q_c} = \frac{x_m}{x_c} = 1$, and earnings equal $1/2$. Following equation (\ref{e:xcxm}), the worker's cognitive and manual skills are both equal to $1/(\alpha + \delta)$. This worker is presented in the first row of \Cref{t:simple_example}. Inferred skills increase with earnings. For a worker with the same relative task intensity, but higher earnings levels, the skill level of each task is greater. This worker is presented in the second row of \Cref{t:simple_example}. Finally, inferred manual skill increases with manual task intensity. Consider some worker with relative manual task intensity equal to two, $\frac{q_m}{q_c} = 2$, and earnings of $1/2$. By equation  (\ref{e:xcxm}), the worker's manual skills increases, while the cognitive skills decreases. This worker is presented in the third row of \Cref{t:simple_example}.

\begin{table}[t!]
\def\arraystretch{1.4}%
\begin{center}
\caption{Illustration of Identification}\label{t:example}
\begin{tabular}{lccccc}
\hline  \hline
\multicolumn{1}{l}{Occupation} &  \multicolumn{1}{c}{Relative} &  \multicolumn{1}{c}{Earnings} & \multicolumn{1}{c}{Manual}  &  \multicolumn{1}{c}{Cognitive}  &  \multicolumn{1}{c}{SOC Code}  \\
  & \hspace{0.45 cm}  $\log \frac{q_m}{q_c}$ \hspace{0.45 cm}  & \hspace{0.45 cm} $\mathbb{E} w(x)$ \hspace{0.45 cm} & \hspace{0.45 cm}  $\log x_m$ \hspace{0.45 cm}  &\hspace{0.45 cm}   $\log x_c$ \hspace{0.45 cm}  \\
\hline
Gardeners 				\hspace{2.1 cm} & \phantom{-}1.7 & \phantom{1}23 & \phantom{-}0.73     & 		  -0.99  &37$-$3010  \\ 
Cashiers				 	 			 &\phantom{-}0.7 & \phantom{1}19 & \phantom{-}0.13     & 		  -0.55  &41$-$2010  \\ 
Police officers 			 				 & 		     -0.1 & \phantom{1}64 & \phantom{-}0.74     & \phantom{-}0.88  &33$-$3050 \\ 
Physicians 	 			 			 &		     -0.2 & 			184 & \phantom{-}1.73     & \phantom{-}1.94  &29$-$1060  \\ 
Financial managers 		 	 			 & 		     -1.8 & \phantom{1}92 & 		  -0.49     & \phantom{-}1.34   &11$-$3030  \\ 
Chief executives  						 & 		     -2.1 & 			149 & 		  -0.25     & \phantom{-}1.82  &11$-$1010  \\ 
\hline \hline
\end{tabular}\end{center}
{\scriptsize \Cref{t:example} shows the identification of worker skills for a number of occupations. Holding constant the relative manual skill intensity, high earnings identify high skill levels. This can be seen by contrasting the manual and cognitive skills of police officers and physicians, and of financial managers and chief executives. Holding constant earnings, high manual task intensity identifies high manual skills. This can, for example, be seen by comparing the skills of gardeners and cashiers.}
\end{table}

 Having illustrated the identification with examples, we turn to identification using earnings data. \Cref{t:example} illustrates the identification of underlying skills for representative workers in occupations listed in the first column. The second column shows the relative manual task intensity for these occupations from O*NET task measures. The third column shows average earnings of the workers by occupation in the ACS. 
 

We identify manual and cognitive skills using equation (\ref{e:xcxm}). First, we show that higher earnings identify higher levels of skills, all else equal. Consider an example of police officers and physicians. The relative task intensity for police officers and physicians is comparable, while the earnings of physicians exceed the earnings of police officers. This implies a higher level of both cognitive and manual skills for physicians, which we document in the fourth and fifth column of \Cref{t:example}. 

Second, we consider two occupations with similar earnings to show that high manual task intensity identifies high manual skill all else equal. While the earnings of gardeners and cashiers are similar, gardening is more demanding in manual skills. By equation (\ref{e:xcxm}), it follows that a gardener has more manual skills than a cashier, but less cognitive skills. The fourth and fifth column in \Cref{t:example} displays this pattern.

   \begin{figure}[!t]
   \begin{center}
    \begin{subfigure}{}
        \includegraphics[trim=0.0cm 0.0cm 0.0cm 0.0cm, width=0.485\textwidth,height=0.29\textheight]{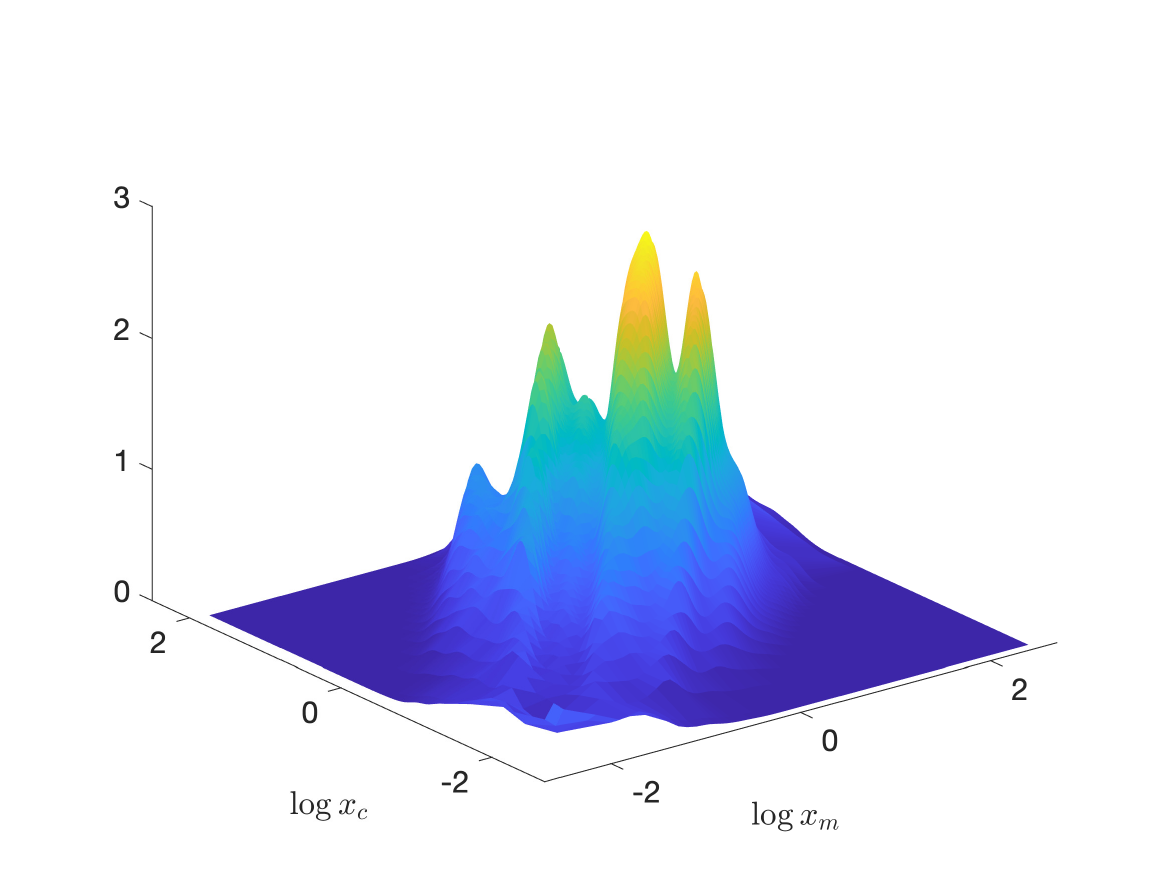}  
    \end{subfigure}%
    \begin{subfigure}{}
        \includegraphics[trim=0.0cm 0.0cm 0.0cm 0.0cm, width=0.485\textwidth,height=0.29\textheight]{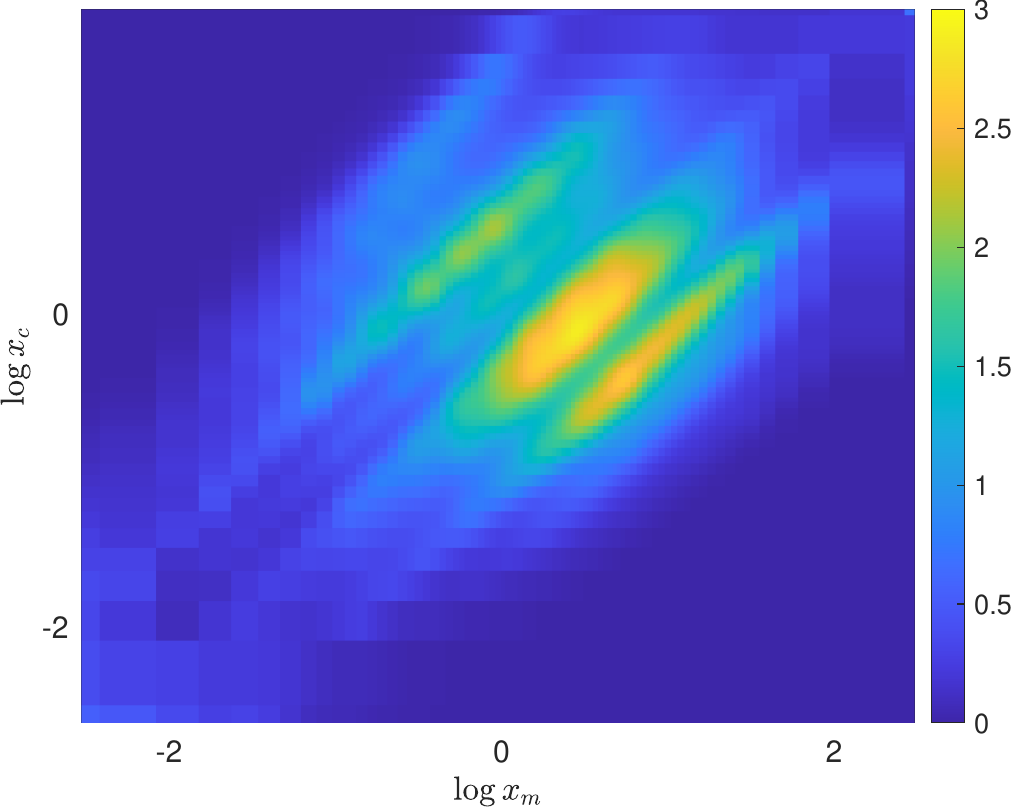}
    \end{subfigure}
    \par\end{center}
    \vspace{-0.6 cm}
    \caption{Distribution of Worker Types} \label{f:worker_types}
{\scriptsize \vspace{.2 cm} Figure \ref{f:worker_types} shows the distribution of manual and cognitive skills across workers. The left-hand panel plots the distribution function, the right-hand panel displays the distributions as a heatmap. The worker skill distribution sharply differs from rotationally invariant distributions such as Gaussians: the skill distribution is multimodal, positively skewed, and has heavy tails.}
    \end{figure}

We apply the identification argument to all workers in the ACS to infer the skill distribution. By identifying skills at the worker level, we allow for skill heterogeneity within occupations driven by earnings variation within that occupation. As in the example, workers with high earnings have higher cognitive and manual skills than a worker with low earnings in the same occupation. \Cref{f:worker_types} shows the distribution of cognitive and manual skills, after 98 percent winsorization and after smoothing the pointwise identified distribution using a kernel density estimation. The worker skill distribution sharply differs from rotationally invariant distributions such as Gaussians: the skill distribution is multimodal, positively skewed, and has heavy tails. \citet*{Kim:2026} similarly estimate highly non-Gaussian distributions of workers skills and job characteristics using a semi-nonparametric estimation approach.

\begin{table}
\global\long\def\arraystretch{1.4}%

\begin{centering}
\caption{Quantitative Model Parameters}
\label{t:modelparam} %
\begin{tabular}{cclcc}
\hline \hline 
\hspace{0.2 cm} Parameter \hspace{0.2 cm}  & Value & Moment & Data & Model\tabularnewline
\hline 
 $\alpha$ & \hspace{0.7cm} 0.239 \hspace{0.7cm} & Mean of log manual skill & \hspace{0.4cm} -0.036 \hspace{0.4cm} & \hspace{0.4cm} -0.039 \hspace{0.4cm} \tabularnewline
 $\beta$  & 0.020 & Mean of log cognitive skill & \phantom{-}0.132 & \phantom{-}0.129 \tabularnewline
 $\delta$ & 0.036 & Variance of log manual skill & \phantom{-}0.234 & \phantom{-}0.263\tabularnewline
&   & Variance of log cognitive skill & \phantom{-}0.389 & \phantom{-}0.393\tabularnewline
&  & Covariance between log skills \hspace{0.6 cm} & \phantom{-}0.326 & \phantom{-}0.257\tabularnewline
\hline \hline 
\end{tabular}
\par\end{centering}
{\scriptsize \vspace{0.1cm}
 \Cref{t:modelparam} presents the parameter values chosen to minimize the distance between occupational statistics in the data and their model counterpart. The first two columns present the parameters of the earnings equation (\ref{e:earningsdata}) and their estimated values. The third column describes the moments across occupations that jointly inform the parameter values. The fourth column presents the empirical moment, which we calculate using the data of  \citet{Acemoglu:2011}, while the fifth column presents the model counterpart. }
\end{table}

We estimate the three relevant model parameters $(\alpha,\beta,\delta)$ that appear in the earnings equation (\ref{e:earningsdata}) by making sure that the model-implied occupational statistics align with the occupational statistics in the data of \citet{Acemoglu:2011}. For each occupation, we calculate the mean of cognitive and manual skills in logarithms across workers in that occupation. We use the means by occupation to calculate the means, the variances, and the covariance of logarithmic skills across occupations. We then choose the model parameters to minimize the squared distance between the model and data in terms of the means and variances of log manual and log cognitive skills and the covariance between manual and cognitive skill across occupations. The occupation statistics for the model and the data as well as the estimated parameter values are presented in Table \ref{t:modelparam}.

\subsection{Cognitive Skill-Biased Technological Change} \label{s:counterfact}

We next use this empirically-relevant model economy to evaluate how the labor market responds to technological change. In particular, we illustrate our framework by analyzing reallocation and earnings changes in response to an increase in the importance of cognitive skills of employees, that is, \textit{cognitive skill-biased technological change}. If cognitive skills of workers (say, denoted by $\tilde{x}_2$) are becoming increasingly important for production, this is represented by $\dot{\gamma}_t > 0$ and $\dot{\delta}_t > 0$. Specifically, we consider cognitive skill-biased technological change such that $\dot{\gamma} = \dot{\delta}$. 

   \begin{figure}[!t]
   \begin{center}
    \begin{subfigure}{}
        \includegraphics[trim=0.0cm 0.0cm 0.0cm 0.0cm, width=0.625\textwidth,height=0.35\textheight]{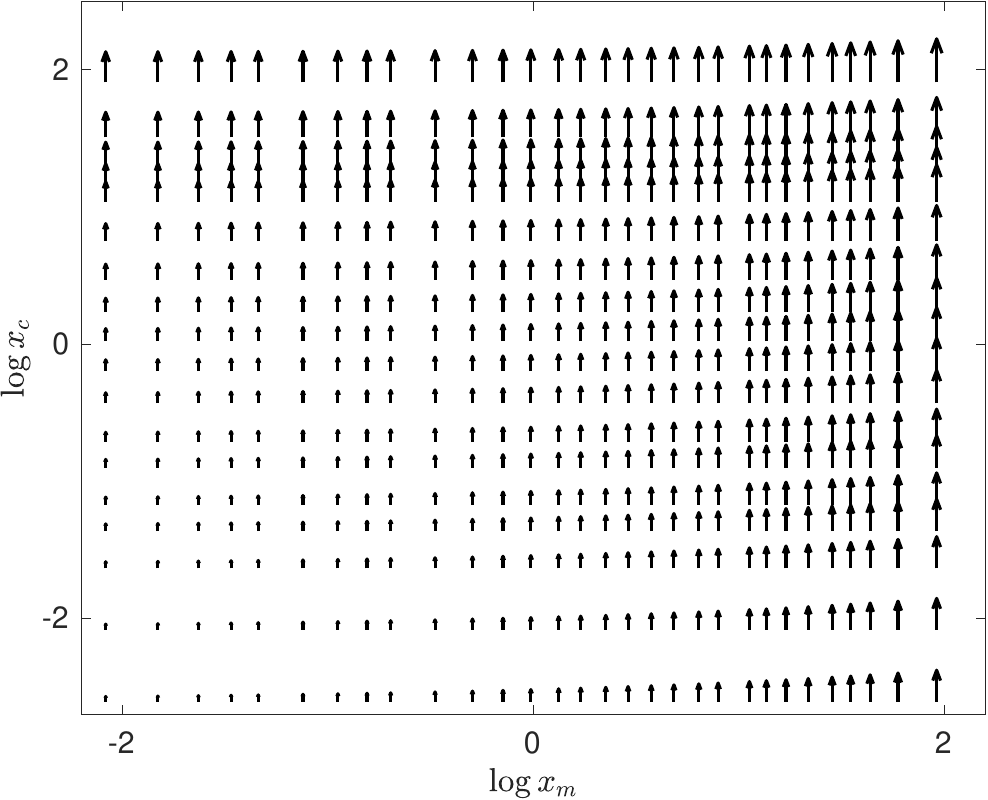}  
    \end{subfigure}%
    \par\end{center}
    \vspace{-0.8 cm}
    \caption{Cognitive Skill-Biased Technological Change}
{\scriptsize \vspace{.2 cm} Figure \ref{fig:techchange} depicts cognitive skill-biased technological change, which increases the marginal product of cognitive skills (indicated by the upward motion of the vectors) while leaving the marginal product of manual skills unaffected (indicated by the absence of horizontal motion). Technological change is skill-biased since the marginal product of high skill workers increases more than the marginal product of low skills workers.}
    \label{fig:techchange}
    \end{figure}

Figure \ref{fig:techchange} presents the technological change over the worker types space. Technological change is the change in marginal product holding fixed the initial assignment, that is, $\dot{\mathcal{A}}(x) = \dot{\Sigma} x$. Since $\dot{\alpha} = \dot{\beta} = 0$ there is no change in the marginal product of manual skills $x_m$. In Figure \ref{fig:techchange} this follows since the vectors do not point  left or right. Due to cognitive skill-biased technological change, the marginal product of cognitive skills increases by $\dot{\gamma} x_m + \dot{\delta} x_c$. Cognitive skill-biased technological change increases the marginal product of cognitive skills leaving the marginal product of manual skills unaffected. In Figure \ref{fig:techchange}, the increase in the marginal product of cognitive skills is seen from the upward pointing vectors. Since technological change is linear in worker skill, marginal product of cognitive skills increases in $x_c$ and $x_m$, which is visually captured by the increasing length of the vectors with $x_c$ and $x_m$.\footnote{We only plot a subset of vectors and vectors are not drawn to scale in order to facilitate reading the figures.} Technological change is thus skill-biased $-$ the marginal product of high skill workers increases more than the marginal product of low skill workers.

Using the Helmholtz decomposition (\Cref{eq:helmholtz}), we decompose the technological change in Figure \ref{fig:techchange} into a divergence-free reallocation and a gradient component that captures the change in marginal earnings. The Helmholtz decomposition of cognitive skill-biased technological change is presented in Figure \ref{fig:techdec} and Figure \ref{fig:techdecwag}. 


   \begin{figure}[!t]
   \begin{center}
    \begin{subfigure}{}
        \includegraphics[trim=0.0cm 0.0cm 0.0cm 0.0cm, width=0.625\textwidth,height=0.35\textheight]{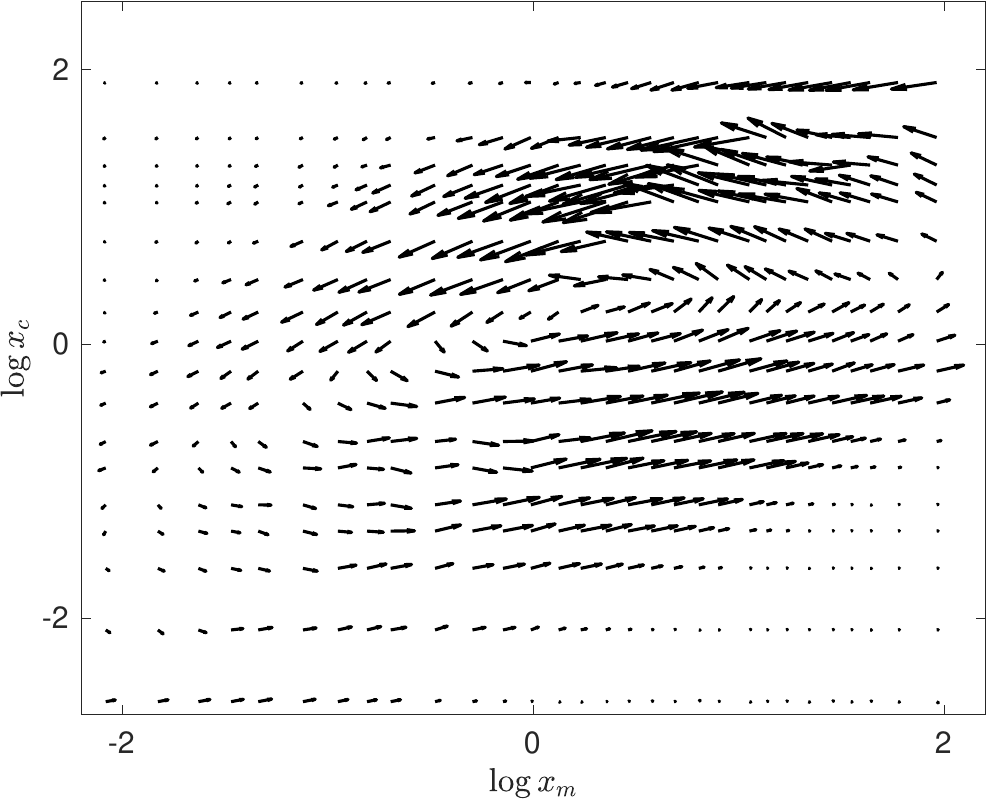}
    \end{subfigure}
    \par\end{center}
    \vspace{-0.8 cm}
    \caption{Labor Reallocation in Response to Cognitive Skill-Biased Technological Change} 
{\scriptsize \vspace{.2 cm} Figure \ref{fig:techdec} shows reallocation in response to cognitive skill-biased technological change. Labor reallocates in a counterclockwise fashion to increase the assortativeness of sorting between workers with high cognitive skill and replacement workers with high manual skill.}
    \label{fig:techdec}
    \end{figure}

\vspace{0.4 cm}
\noindent \textbf{Reallocation}. In response to cognitive skill-biased technological change, labor reallocates in a counterclockwise fashion as shown in Figure \ref{fig:techdec}. In order to build intuition for reallocation of labor, suppose that the complementarity between the worker's manual skills $x_m$ and their replacement's cognitive skills $\tilde{x}_c$, captured by $\gamma_t$ in the complementarity matrix $\Sigma_t$, grows large. When the reallocation problem with technology (\ref{e:bilinearoutputworkers}) is dominated by $\gamma_t \tilde{x}_c x_m$, the planner positively sorts workers' manual skills and the replacements' cognitive skills. The replacement for a worker with high manual skills should be a worker with high cognitive skills. Hence, when $\gamma_t$ increases the replacement for worker $(2,1)$ in logarithms comes from direction $(-1,1)$. Similarly, the replacement for worker $(0,1)$ in logarithms with low manual skills should be a worker with low cognitive skills from the direction $(-1,-1)$. Optimal reallocation is thus counterclockwise as shown in Figure \ref{fig:techdec}. We note that $\dot{\delta}$ does not influence reallocation since the change in $\delta_t$ is a symmetric technological change. Moreover, reallocation is non-linear since the worker skill distribution sharply differs from a Gaussian distributions.\footnote{Since the feasible reallocations (\ref{eq:rearr0}) and (\ref{eq:rearr0b}) depend on the worker skill distribution $f$, reallocation necessarily differs from what one would obtain by approximating the data with a normal distribution. In other words, linear reallocation is not feasible under the empirical worker distribution.}

\begin{implication}{\textit{Labor Reallocation and Skill Complementarity}.} \label{c:complrl}
Changes in sorting are exclusively driven by changes in between-skill complementarity. Cognitive skill-biased technological change induces workers with high (low) cognitive skills to replace workers with more (less) manual skills. 
\end{implication}


   \begin{figure}[!t]
   \begin{center}
    \begin{subfigure}{}
        \includegraphics[trim=0.0cm 0.0cm 0.0cm 0.0cm, width=0.625\textwidth,height=0.35\textheight]{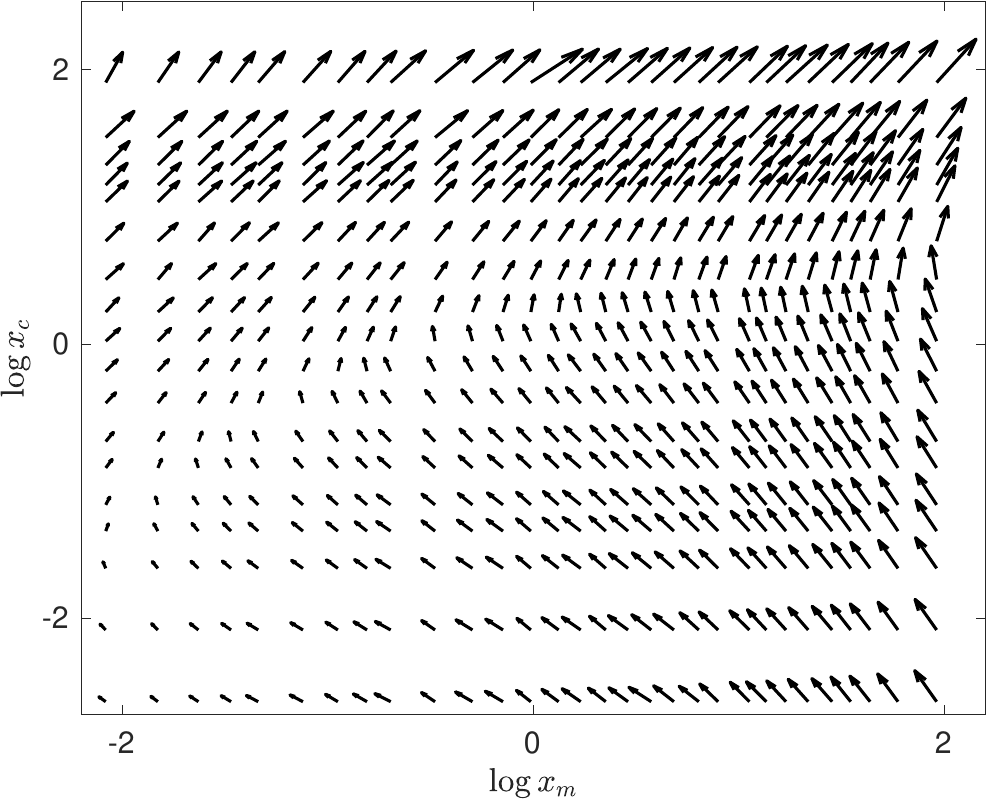}
    \end{subfigure}
    \par\end{center}
    \vspace{-0.8 cm}
    \caption{Earnings Change in Response to Cognitive Skill-Biased Technological Change} 
{\scriptsize \vspace{.2 cm} Figure \ref{fig:techdecwag} illustrates the change in marginal earnings in response to cognitive skill-biased technological change. Since cognitive skill-biased technological change is not a gradient, the change in marginal earnings is not equal to the technological change (\Cref{fig:techchange}).}
    \label{fig:techdecwag}
    \end{figure}

\noindent \textbf{Earnings Changes}. Since cognitive skill-biased technological change is not a gradient, the change in marginal earnings is not equal to the technological change. This can be observed by comparing technological change in Figure \ref{fig:techchange} to marginal earnings changes in Figure \ref{fig:techdecwag}. So, how does cognitive skill-biased technological change pass through to earnings? 

The Helmholtz decomposition when the production technology is bilinear shows that the change in marginal earnings is:
\begin{equation}
(\nabla \dot{w})(x) = \dot{\Sigma} x - \Sigma \dot{r} (x)  \label{eq:changeearn}
\end{equation}
where $ - \dot{r}(x)$ captures the direction of the worker that worker $x$ replaces under the new technology. First, marginal earnings vary with the change in marginal products, $\dot{\Sigma} x$. Second marginal earnings also vary due to marginal productivity being altered because worker $x$ replaces other workers, and marginal products change as $- \Sigma \dot{r}(x)$. 


In Figure \ref{fig:techdecwag}, we illustrate the effect of cognitive skill-biased technological change on marginal earnings. It follows that an increase in $\gamma_t$ and $\delta_t$ directly increase the marginal product of cognitive skills, which is given by $\dot{\gamma} x_1 + \dot{\delta} x_2$, while there is no direct effect on the marginal product of manual skills. However, the marginal product also changes due to reallocation $- \dot{r}(x)$, which is captured by the second term in equation (\ref{eq:changeearn}). 

Consider worker $\log x = (2,1)$ with high manual and cognitive skills. Due to the direct effect of technological change, the marginal product of cognitive skills changes as $\dot{\gamma} x_1 + \dot{\delta} x_2$, with no direct effect on the marginal product of manual skill (Figure \ref{fig:techchange}). However, this worker also replaces the worker in the direction $(1,-1)$, or $\dot{r}(x) = (-1,1)$, changing its marginal product. Specifically, the marginal product of manual skill indirectly changes as $- \alpha \dot{r}_1(x) - \beta \dot{r}_2(x)$, while marginal product of cognitive skill indirectly changes as $- \gamma \dot{r}_1(x) - \delta \dot{r}_2(x)$. Since $\alpha$ exceeds $\beta$ (\Cref{t:modelparam}) the marginal product of manual skills increase, and since $\delta$ exceeds $\gamma$ the increase in the marginal product of cognitive skills is muted. This is seen in Figure \ref{fig:techdecwag}, where changes in marginal earnings with respect to manual skills are given by horizontal movements. The right-and-upward vector for worker $(2,1)$ in logarithms shows that the marginal earnings for manual and cognitive skill both increase.

Similarly, for worker $\log x = (1,-1)$ with high manual and low cognitive skills the direct effect of technological change on marginal output is $\dot{\gamma} x_1 + \dot{\delta} x_2$, which is smaller than for worker $\log x = (2,1)$ because the worker is less skilled. Since this worker replaces the worker in direction $(-1,-1)$ with less manual skills and less cognitive skills, or $\dot{r}(x) = (1,1)$ and since all technology parameters are positive, the increase in the marginal product of both manual and cognitive skills is reduced. Since the direct effect of cognitive skill-biased technological change on the marginal product of manual skills is zero, this introduces a decline in the marginal product of manual skills. 

The changes in marginal earnings for cognitive skills are captured in Figure \ref{fig:techdecwag} by the vertical displacement. Figure \ref{fig:techdecwag} shows that marginal earnings for cognitive skills inherit two properties of the cognitive skill-biased technological change displayed in Figure \ref{fig:techchange}. First, the marginal earnings for cognitive skills increase for all workers. This is given by the fact that all vectors are directed upwards. Second, the increase in marginal earnings is more pronounced for high skill workers, that is, the change in marginal earnings is skill-biased. This is given by the increasing norm of the vectors from the bottom-left (low skill workers) to the top-right (high skill workers) corner.

\begin{implication}{\textit{Cognitive Skill-Biased Technological Change and the Value of Cognitive Skills}.} \label{c:mwc}
Marginal earnings for cognitive skills ($i$) increase for all workers; and ($ii$) increase more for high skill workers, in response to cognitive skill-biased technological change.
\end{implication}

In Figure \ref{fig:techdecwag}, the change in marginal earnings in the manual skill dimension is captured by the horizontal displacement. Strikingly, even though cognitive skill-biased technological change only affects the marginal product of cognitive skills, marginal earnings for manual skills change. Specifically, cognitive skill-biased technological change decreases the marginal earnings of manual skills for workers with low cognitive skills (bottom half of the panel) while increasing the marginal earnings of manual skills for workers with high cognitive skills (top half of the panel) in the manual skill dimension. 


\begin{implication}{\textit{Cognitive Skill-Biased Technological Change and the Value of Manual Skills}.} \label{c:mwm}
Marginal earnings for manual skills increase for workers with high cognitive skills and decrease for workers with low cognitive skills in response to cognitive skill-biased technological change.
\end{implication}



   \begin{figure}[!t]
   \begin{center}
    \begin{subfigure}{}
        \includegraphics[trim=0.0cm 0.0cm 0.0cm 0.0cm, width=0.625\textwidth,height=0.35\textheight]{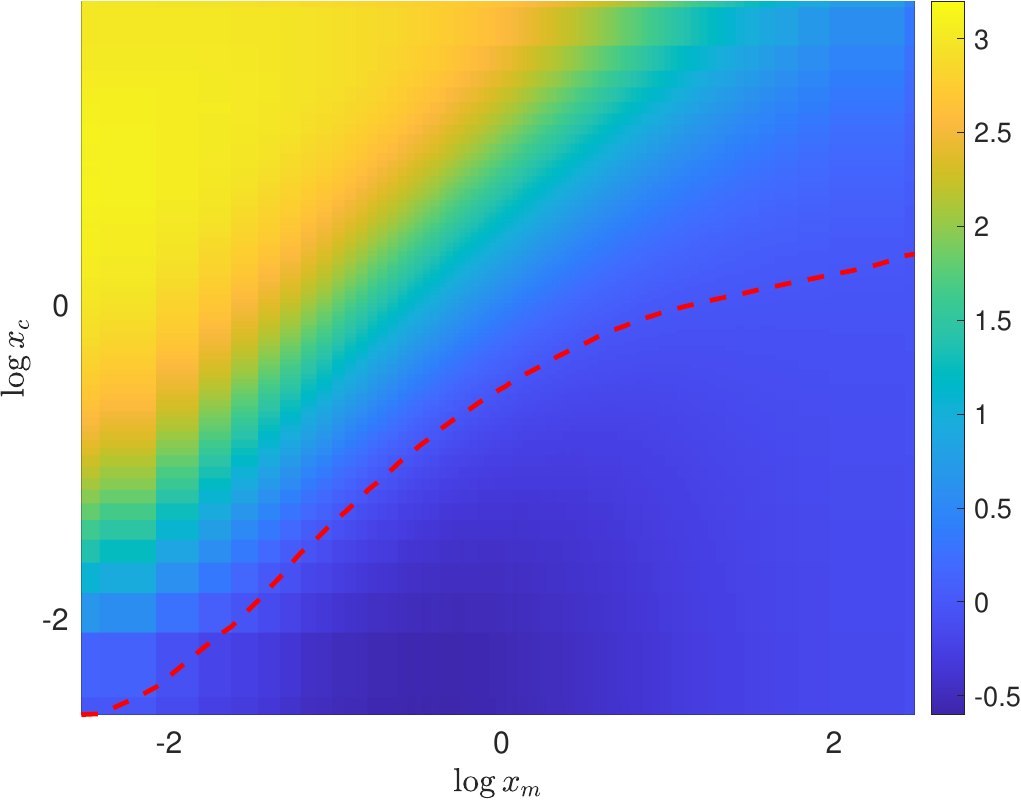}  
    \end{subfigure}%
    \par\end{center}
    \vspace{-0.8 cm}
    \caption{Earnings Changes with Cognitive Skill-Biased Technological Change}
{\scriptsize \vspace{.2 cm} Figure \ref{fig:earnchange} depicts the percentage change in earnings for different worker types in response to cognitive skill-biased technological change. The dashed red line represents an isocurve along which workers do not experience changes in earnings.}
    \label{fig:earnchange}
    \end{figure}
    

Having quantified the change in marginal earnings, we integrate the change in marginal earnings to obtain the earnings changes. Figure \ref{fig:earnchange} shows the resulting percentage change in earnings for all workers, where the dashed red line represents an isocurve of worker types that do not experience overall changes in earnings. Cognitive skill-biased technological change mostly benefits workers with high cognitive skills and low manual skills as shown by the bright-yellow color in the top-left of Figure \ref{fig:earnchange}. While workers with high cognitive and manual skills benefit in terms of the level increase in earnings, the percentage increase is dampened by the higher initial level of earnings. The earnings of workers with high manual skills and low cognitive skills decrease relative to the workers with high cognitive and low manual skills. 

\section{Additional Characterization}

In this section, we provide additional characterization of comparative statics.

\subsection{Second-Order Perspective on Reallocation}

We next provide an alternative interpretation of reallocation in response to technological change. In order to derive Theorem \ref{eq:helmholtz}, we  analyzed the sensitivity with respect to technological progress of two equilibrium conditions $-$ labor market clearing (\ref{eq:probconstraintrear}) and firm optimality (\ref{eq:hedonicw}). By combining the restrictions (\ref{eq:rearr0}), (\ref{eq:rearr0b}) and (\ref{e:diffhed}) from these two equilibrium conditions, we characterized reallocation and changes in earnings through the Helmholtz decomposition and the Poisson equation. In other words, the equilibrium comparative statics start with first-order necessary conditions, and examines their response to technological change $-$ a first-order expansion of the first-order necessary conditions.

An alternative approach is to directly consider the second-order expansion of the increase in aggregate production with respect to technological change around the initial equilibrium, and to characterize reallocation that maximizes the increase in aggregate production up to second-order. We pursue this alternative approach in \Cref{a:2ndorderapp}. The main result (Proposition \ref{p:socharacterization}) shows that reallocation $\dot{r}$ in \Cref{eq:helmholtz}, where $\nabla \dot{w}$ is characterized using the Poisson equation in Theorem \ref{eq:changeinmearnings}, maximizes aggregate production up to second order in response to technological change. These two approaches are thus equivalent and provide complimentary views. 
  
The result states that a first-order expansion of the equilibrium conditions $-$ corresponding to the first-order conditions of the planning problem $-$ and a second-order expansion of the planning problem yield the same optimal reallocation. A byproduct of the proof is that it characterizes the increase in aggregate output due to optimal reallocation $\dot{r}$ as:
\begin{equation*}
\frac{1}{2} \int \dot{r}(x)^\top \mathcal{C}(x) \dot{r}(x) f(x) d x, 
\end{equation*}
which shows that more reallocation is indicative of larger increases in output. The complementarity matrix $\mathcal{C}$ weights how valuable reallocation is in terms of additional aggregate output: reallocation along dimensions with strong complementarity generates large output gains. The worker distribution $f$ scales these effects, showing that reallocating many workers has a larger aggregate impact than moving few. This result complements Theorem \ref{eq:helmholtz} and Theorem \ref{eq:changeinmearnings} to show not only how technological progress passes through to earnings and reallocation but also how much output increases.

\subsection{Large Technological Change} \label{s:largetc}

So far, we have characterized comparative statics for small technological change around the initial equilibrium. We next discuss equilibrium adjustments in earnings and sorting in response to large technological change. We show how to use the Helmholtz decomposition and Poisson equation to understand the response to large technological change. We then use the prior analysis for each technology $t$ to derive a differential equation that gives the relation between earnings and sorting under technology $t$, and equilibrium responses to technological change. For each technology, technological change decomposes into changes in marginal earnings and sorting that are orthogonal to one another.

\vspace{0.4 cm}
\noindent The objective is to characterize an assignment $\tau_t$ and earnings schedule $w_t$ such that the marginal earnings condition (\ref{eq:focfirm}) given by $(\nabla_1 y_t)(x, \tau_t(x)) = (\nabla w_t)(x)$ is satisfied for every technology $t \in [0,T]$ given an initial equilibrium.

For a change in the assignment function $\dot{\tau}_t(x)$ to clear the labor market, $\upsilon_t(x) f(x)$ is necessarily divergence free and parallel to the boundary, where $(\nabla \tau_t(x)) \upsilon_t(x) + \dot{\tau}_t(x) = 0$. By differentiating the marginal earnings condition with respect to technology $t$, we obtain:
\begin{equation}
(\nabla_1 \dot{y}_t)(x, \tau_t(x)) + (\nabla^2_{12} y_t)(x, \tau_t(x)) \dot{\tau}_t (x)= (\nabla \dot{w}_t)(x) \hspace{0.03 cm}.
\end{equation}
This condition equates changes in marginal production on the left to changes in marginal earnings on the right. The direct effect gives the change in marginal output $(\nabla_1 \dot{y}_t)(x, \tau_t(x))$ holding fixed the assignment $\tau_t$. The reallocation effect yields the change in the marginal production of worker $x$ when they are assigned to a different job. The right-hand side captures the change in marginal earnings due to a change in technology $(\nabla \dot{w}_t)$.

By differentiating the marginal earnings condition with respect to worker skills $x$, $
(\nabla^2 w_t)(x) = (\nabla^2_{11} y_t)(x, \tau_t(x)) + (\nabla^2_{12} y_t)(x, \tau_t(x)) \nabla \tau_t (x)$, the change in the marginal labor earnings with respect to technology is:
\begin{equation}
(\nabla_1 \dot{y}_t)(x, \tau_t(x)) + (\nabla^2_{11} y_t)(x, \tau_t(x))\upsilon_t(x)  = (\nabla \dot{w}_t)(x)+  (\nabla^2 w_t)(x) \upsilon_t(x) , \label{eq:flow} 
\end{equation}
where $\upsilon_t \in \mathcal{D}$. This equation shows that two margins adjust in response to technological change for each technology $t$: marginal earnings change and labor reallocates. The change in marginal earnings is a gradient $(\nabla{\dot{w}}_t \in \mathcal{G})$, while reallocation is divergence-free and parallel to the boundary $(\upsilon_t \in \mathcal{D})$. We recall that the key property is that these two spaces are orthogonal to one another. In sum, equation (\ref{eq:flow}) is thus the generalization of equation (\ref{e:diffhed0}) which we used in Section \ref{s:solvingrealloc} to characterize comparative statics for technological change around the initial equilibrium. We use similar ideas to those in Theorem \ref{eq:helmholtz} to characterize how earnings and sorting respond to large technological change in Proposition \ref{t:flow}.

\begin{proposition} \label{t:flow}
Let $(\tau,w)$ be an equilibrium for the initial technology $t=0$. Let $(\tau_t,w_t)$ satisfying equation (\ref{eq:flow}) for each technology $t \in [0,T]$ for some $T >0$ with $(\nabla^2 w_t)(x)- (\nabla^2_{11}y_t)(x,\tau_t(x))$ positive definite, where $(\nabla \tau_t(x)) \upsilon_t(x) + \dot{\tau}_t(x) = 0$ and $\upsilon_t \in \mathcal{D}$. The pair $(\tau_t,w_t)$ is the equilibrium given technology $t$.
\end{proposition}

\noindent The proof of this result is given in Appendix \ref{a:largechange}. 
 
We finally remark that Proposition \ref{t:flow} suggests a novel approach to characterize multidimensional assignment problems. One can start from a closed-form solution, such as bilinear output functions with normally distributed workers and jobs, and construct a trajectory of equilibrium changes in response to a trajectory of technology changes that moves one from the initial technology to the target technology.\footnote{For completeness, we discuss the unidimensional case in Appendix \ref{a:unidimensional}.} 

\section{Conclusion}

This paper develops a general theory of comparative statics for multidimensional sorting models. A main difficulty in multidimensional assignment models is that, outside of a few highly symmetric environments, the model does not admit a tractable characterization. Our contribution is to show that the comparative statics can be completely characterized.




\clearpage

{ {
\bibliographystyle{econometrica}	
\bibliography{BOT_bib}
} }

\clearpage

\pagebreak

\renewcommand{\theequation}{A.\arabic{equation}} \setcounter{equation}{0}
\renewcommand{\thefigure}{A.\arabic{figure}}\setcounter{figure}{0}
\renewcommand{\thetable}{A.\arabic{table}}\setcounter{table}{0}
\setcounter{page}{1}
\newpage \appendix

\newpage

\begin{center}
{\Large Multidimensional Sorting: Comparative Statics\\}
\bigskip
{\Large Appendix \\}
\bigskip
{\large  Job Boerma, Andrea Ottolini, and Aleh Tsyvinski \\}
\bigskip
{\large August 2026}
\end{center}

\section{Mathematical Background}

This appendix provides a short introduction into some mathematical results that we use repeatedly.

\subsection{Transport Equations} \label{a:transporteq}

We derive the transport equation for a reallocation problem and for a general assignment problem between distributions $f_t$ and $g_t$.

\vspace{0.4 cm}
\noindent \textbf{Reallocation}. For the reallocation problem, we differentiate the labor market clearing condition (\ref{eq:probconstraintrear}) with respect to technological change in order to obtain:
\begin{align*}
0 & = \int_X \nabla_r (h(r_t(x))) \dot{r}_t(x) f(x) dx  = \int_X \nabla_x (h(r_t(x))) (\nabla_x r_t(x))^{-1} \dot{r}_t(x) f(x) dx ,
\end{align*}
where the second equality follows since $\nabla_x (h(r_t(x))) = \nabla_r (h(r_t(x))) \nabla_x r_t(x)$. By defining $\upsilon_t(x) := - (\nabla_x r_t(x))^{-1} \dot{r}_t(x)$, and by integrating the previous expression by parts:
\begin{align*}
0 & = \int_X h(r_t(x)) \nabla \cdot (\upsilon_t(x) f(x)) dx  - \int_{\partial X} h(r_t(x))) \upsilon_t(x) f(x) \cdot n(x) dx ,
\end{align*}
where $n(x)$ is the outward normal vector.

Since the previous expression has to hold for all smooth functions $h$, it follows that for a change in the reallocation $\dot{r}_t(x)$ to be feasible, it has to be that:
\begin{align}
0 & = \nabla \cdot (\upsilon_t(x) f(x)) \hspace{3.65 cm} \text{$ x \in X$} \label{eq:rearr} \\
0 & = \upsilon_t(x) f(x) \cdot n(x)  \hspace{3.5 cm} \text{$x \in \partial X$} \label{eq:rearrb}
\end{align}
A particular instance of this is at the initial point of technology, $r(x) = x$ and hence $\nabla_x r(x) = I$ and $\upsilon_t(x) = - \dot{r}_t(x)$. The reallocation restrictions (\ref{eq:rearr}) and (\ref{eq:rearrb}) simplify to the reallocation restrictions (\ref{eq:rearr0}) and (\ref{eq:rearr0b}) at the initial technology in the main text.

\vspace{0.4 cm}
\noindent \textbf{General Case}. When $\tau_t$ is an assignment between the distributions $f_t$ and $g_t$, it follows by the labor market clearing condition (\ref{eq:constraint}) that:\footnote{See, for example, Definition 1.2 in \citet{Villani:2009}.}
\begin{equation*}
\int h(z) g_t(z) dx = \int h(\tau_t(x)) f_t(x) dx
\end{equation*}
for any smooth function $h$. 

We differentiate this expression with respect to technological change in order to obtain:
\begin{align*}
\int h(z) \dot{g}_t(z) dx & = \int \nabla_1 (h(\tau_t(x))) \dot{\tau}_t(x) f_t(x) dx + \int h(\tau_t(x)) \dot{f}_t(x) dx \\
& = \int \nabla_x (h(\tau_t(x))) (\nabla_x \tau_t(x))^{-1} \dot{\tau}_t(x) f_t(x) dx + \int h(\tau_t(x)) \dot{f}_t(x) dx  ,
\end{align*}
where the second equality follows since $\nabla_x (h(\tau_t(x))) = \nabla_1 (h(\tau_t(x))) \nabla_x \tau_t(x)$. By defining $\upsilon_t(x) := - (\nabla_x \tau_t(x))^{-1} \dot{\tau}_t(x)$, and by integrating the previous expression by parts:
\begin{align*}
\int h(z) \dot{g}_t(z) dx & = \int_X h(\tau_t(x)) \big( \nabla \cdot (\upsilon_t(x) f_t(x)) + \dot{f}_t(x) \big) dx  - \int_{\partial X} h(\tau_t(x)) \upsilon_t(x) f(x) \cdot n(x) dx ,
\end{align*}
where $n(x)$ is the outward unit normal vector. We rewrite the left-hand side of the equation using a change of variables $z=\tau_t(x)$ so that, after rearranging:
\begin{align*}
\int_X h(\tau_t(x)) \big( \dot{f}_t(x) + \nabla \cdot (\upsilon_t(x) f_t(x))   - \dot{g}_t(\tau_t(x)) | \text{det} \nabla \tau_t(x) | \big) dx =  \int_{\partial X} h(\tau_t(x)) \upsilon_t(x) f(x) \cdot n(x) dx .
\end{align*}

Since the previous expression has to hold for all smooth functions $h$, it follows that for a change in the assignment $\dot{\tau}_t(x)$ to be feasible, it has to be that:
\begin{align}
0 & = \nabla \cdot (\upsilon_t(x) f_t(x)) - \dot{g}_t(\tau_t(x)) | \text{det} \nabla \tau_t(x) |  + \dot{f}_t(x) \hspace{3.65 cm} \text{$ x \in X$} \label{eq:rearrgen} \\
0 & = \upsilon_t(x) f_t(x) \cdot n(x)  \hspace{9.09 cm} \text{$x \in \partial X$} \label{eq:rearrbgen}
\end{align}
where $\upsilon_t(x) := - (\nabla_x \tau_t(x))^{-1} \dot{\tau}_t(x)$, or equivalently, $(\nabla_x \tau_t(x)) \upsilon_t(x) + \dot{\tau}_t(x) = 0$. 

\subsection{Gradient} \label{a:gradient}


In this appendix, we characterize what vector fields are the gradient of function in the case when the domain of the vector field is a convex set.\footnote{In fact, this assumption can be considerably relaxed by requiring only that the set is simply connected.}

\begin{lemma}{\citet{Poincare:1887}.} \label{l:poincare}
Let the support $X$ of the function be a convex set and let $w(x) = (w_1(x),\dots,w_n(x))$ where $x = (x_1,\dots,x_n) \in X$. Assume that all $w_i$ are continuously differentiable functions. Then $w(x) = \nabla \varphi(x)$ for some twice differentiable function $\varphi$ with continuous derivative if and only if $\frac{\partial w_i(x)}{\partial x_j} = \frac{\partial w_j(x)}{\partial x_i}$ for all $(i,j)$.
\end{lemma}


\begin{proof}
``$\Longrightarrow$'': In this direction, we evaluate $\frac{\partial w_i(x)}{\partial x_j}$ as:
\begin{equation*}
\frac{\partial w_i(x)}{\partial x_j} = \frac{\partial}{\partial x_j} \frac{\partial \varphi}{\partial x_i} = \frac{\partial}{\partial x_i} \frac{\partial \varphi}{\partial x_j} =  \frac{\partial w_j(x)}{\partial x_i} ,
\end{equation*}
where the first and third equality follow since $w(x) = \nabla \varphi(x)$, while the second equality follows by the symmetry of continuous second derivatives. 


``$\Longleftarrow$'': Consider the function specified from some origin point:
\begin{equation*}
\varphi(x) := \int_0^1 w(tx) \cdot x dt  = \sum_{i=1}^n \int_0^1 w_i(t x) x_i dt 
\end{equation*}
this function is well-defined since $w(x)$ is defined over the segment $[0,x]$ as the support is connected. 

Let $z_j = t x_j$, then we can write:
\begin{equation*}
\frac{\partial \varphi(x)}{\partial x_j} = \sum_{i=1}^n \int_0^1 t \frac{\partial w_i}{\partial z_j}(t x) x_i dt + \int_0^1 w_j(t x) dt = \sum_{i=1}^n \int_0^1 t \frac{\partial w_j}{\partial z_i}(t x) x_i dt + \int_0^1 w_j(t x) dt
\end{equation*}
where the final equality follows since $\frac{\partial w_j}{\partial x_i} = \frac{\partial w_i}{\partial x_j}$. This expression can be further simplified noting that $\sum \frac{\partial w_j}{\partial z_i}(t x) x_i dt = \sum \frac{\partial w_j}{\partial z_i}(t x) \frac{\partial z_i}{\partial t} dt = \frac{\partial}{\partial t} w_j(tx)$ so that:
\begin{equation*}
\frac{\partial \varphi(x)}{\partial x_j} = \int_0^1 t \frac{\partial}{\partial t} w_j(tx) dt + \int_0^1 w_j(t x) dt = t w_j(tx) \Big|_{t=0}^{t=1}  = w_j(x) \hspace{0.04 cm},
\end{equation*}
where the second equality follows from integration by parts. As a result, $\nabla \varphi(x) = w(x)$. Moreover, since $w$ has a continuous derivative $\varphi$ is twice differentiable with continuous derivative.\end{proof}

\section{Proofs}

In this appendix, we formally prove the results in the main text.

\subsection{Proof to Lemma \ref{l:lemmamap}} \label{pf:lemmamap}

In this appendix, we establish that the optimal assignment is such that the job type $z$ is uniquely determined by the worker type $x$ for all possible types $x$. That is, there is only one job $z$ that employs workers of type $x$. 


In order to ensure that a single firm type employs workers of type $x$, we require that the firm's optimality condition (\ref{eq:focfirm}) is satisfied by a single job type $z$ for each worker type $x$. This is guaranteed by the twist condition (Section \ref{s:assumptions}). As a result, the optimal coupling $\pi$ is supported on $(x, \tau(x))$, where $x \in \text{supp } f$. 

\begin{lemma}\label{l:lemmamap}
If the output function $y_t(x,z)$ satisfies the twist condition, then the optimal plan in the Kantorovich sense is of the form $(x,\tau_t(x))$ for a function $\tau_t$ from the worker space $X$ to the job space $Z$. 
\end{lemma}

\begin{proof}
Suppose not, then there exist jobs $z_1$ and $z_2 \neq z_1$ which are both on the support of the optimal assignment $\pi$ for some worker type $\hat{x}$. Thus, we have an equilibrium in which the solution to the firm problem (\ref{eq:firmprob}) for both job $z_1$ and job $z_2$ is to employ worker $\hat{x}$. Equivalently, the profit for both jobs, $y_t(x,z_1) - w_t(x)$ and $y_t(x,z_2) - w_t(x)$, is maximized at worker type $x=\hat{x}$. By the necessary condition to firm's problem, this implies:
\begin{equation*}
\nabla_1 y_t(x,z_1) = \nabla_1 y_t(x,z_2) ,
\end{equation*}
which contradicts the twist condition in Section \ref{s:assumptions}.
\end{proof}

\subsection{Proof to Lemma \ref{l:reallocation}} \label{a:reallocationproblem}

\noindent In order to prove Lemma \ref{l:reallocation}, it is useful to note that given an assignment $z = \tau(x)$, by change of variables, we can write: 
\begin{equation*}
\int_A g(z) d z = \int_{\tau^{-1}(A)} g(T(x)) | \det \nabla \tau(x) | dx .
\end{equation*} 
By combining this equation with the assignment constraint (\ref{eq:constraint}), we write: 
\begin{equation}
f(x) = g(\tau_t(x)) | \det  \nabla \tau(x) | \hspace{0.04 cm}. \label{eq:constraintma}
\end{equation}


For any assignment function $\tau_t$, consider the reallocation function $r$ such that $\tau_t(r_t(x)) = \tau(x)$. We change variables $x=r_t(z)$ in order to write:
\begin{align}
\int y_t(x,\tau_t(x)) f(x) dx & = \int y_t(r_t(z),\tau_t(r_t(z))) f(r_t(z)) | \det \nabla r_t(z) | dz = \int y_t(r_t(z),\tau_t(r_t(z))) f(z) dz \notag \\
& = \int y_t(r_t(z),\tau(z)) f(z) dz \label{eq:outputeq1}
\end{align}
where the second equality follows by the characterization of the constraint (\ref{eq:constraintma}) for the case of a reallocation and the third equality follows since $\tau_t(r_t(x)) = \tau(x)$. Output is identical under the assignment function and the associated reassignment. 

In order to establish \Cref{l:reallocation}, we establish two inequalities. First,
\begin{align*}
\max\limits_{\tau_t} \int y_t(x, \tau_t(x)) f(x) dx & = \int y_t(x, \hat{\tau}_t(x)) f(x) dx = \int y_t(\hat{r}_t(x),\tau(x)) f(x) dx \\
& \leq \max\limits_{r_t} \int y_t ( r_t(x), \tau(x)) f(x) dx
\end{align*}
where the second equality follows from (\ref{eq:outputeq1}). The converse is established in the same fashion:
 \begin{align*}
\max \limits_{r_t} \int y_t ( r_t(x), \tau(x)) f(x) dx & = \int y_t ( \tilde{r}_t(x), \tau(x)) f(x) dx = \int y_t(x, \tilde{\tau}_t(x)) f(x) dx \\
& \leq \max\limits_{\tau_t} \int y_t(x, \tau_t(x)) f(x) dx 
\end{align*}
which concludes the proof.

\subsection{Proof to Lemma \ref{l:antisymmetrics}} \label{a:antisymmetrics}

First, recall that the set of worker types $X$ is bounded, convex, and has a smooth boundary. If the distribution of worker types is rotationally invariant, then $\partial X = \{ x : \lVert x \rVert^2 = b \text{ for some } b \geq 0 \}$.\footnote{The density of a rotationally invariant distribution $f(x) = \tilde{f} \big(\frac{1}{2} \lVert x \rVert ^2\big)$ only depends on the distance from the origin.} Since the density functions for worker skill $f$ is smooth, it follows that $(\nabla f)(x) = \frac{d \tilde{f}}{d \tilde{x}} \big(\frac{1}{2} \lVert x \rVert ^2\big) x$. In words, this is simply saying that the gradient of $f$ must be pointing in a radial direction.

In order to prove Lemma \ref{l:antisymmetrics}, we show that the reallocation conditions (\ref{eq:rearr0}) and (\ref{eq:rearr0b}) are satisfied. First, we consider the reallocation condition (\ref{eq:rearr0}), which can be written as:
\begin{align*}
\nabla \cdot (f(x) \dot{r}(x)) & = f(x) \nabla \cdot \dot{r}(x) + \nabla f(x) \cdot \dot{r}(x)  \\
& = f(x) \nabla \cdot ( \dot{R} x ) + \frac{d \tilde{f}}{d \tilde{x}} \Big(\frac{1}{2} \lVert x \rVert ^2\Big) x \cdot \dot{R} x
\end{align*}
where the second equality uses that the distribution is rotationally invariant and that the labor reallocation is linear. The first term is equal to zero since $\nabla \cdot ( \dot{R} x ) = 0$ when $\dot{R}$ is antisymmetric because the trace of the reallocation matrix $\dot{R}$ equals zero. In order to see this:
\begin{equation*}
\nabla \cdot (\dot{R} x)= 0 \hspace{1.6 cm} \iff \hspace{1.6 cm}  \sum_{i=1}^n \frac{\partial}{\partial x_i} \Big( \sum_j \dot{R}_{ij} x_j \Big) = \sum_{i=1}^n \dot{R}_{ii} = \text{Tr}(\dot{R}) = 0 .
\end{equation*}
on $X$. The second term is also equal to zero since $x \cdot \dot{R} x = ( \dot{R}^\top x) \cdot x = - (\dot{R} x) \cdot x = - x \cdot \dot{R} x$, where the second equality follows since $\dot{R}$ is antisymmetric, which can only be satisfied when $x \cdot \dot{R} x = 0$.

Second, we consider the reallocation condition (\ref{eq:rearr0b}), where the outward unit vector is $n(x) = \frac{x}{\lVert x \rVert}$, which yields:
\begin{equation*}
\dot{r}(x) \cdot n(x) = 0 \hspace{2 cm} \iff \hspace{2 cm}  \dot{R} x \cdot x = 0 .
\end{equation*}
Since $\dot{R} x \cdot x = 0$ when $\dot{R}$ is antisymmetric as established in the previous paragraph, this concludes the proof.

\subsection{Proof to Theorem \ref{eq:helmholtz}} \label{a:helmholtz}

In this appendix, we establish that the Helmholtz decomposition in Theorem \ref{eq:helmholtz} is unique. More precisely, we show that the solution is unique up to a constant.

\vspace{0.4 cm}
\noindent In order to establish uniqueness, suppose by contradiction that there are two distinct Helmholtz decompositions (\ref{e:equation2}) given by $\dot{\mathcal{A}}(x) = (\nabla \dot{w}_1) (x ) + \mathcal{C}(x) \dot{r}_1(x)$ and by $\dot{\mathcal{A}}(x) = (\nabla \dot{w}_2) (x ) + \mathcal{C}(x) \dot{r}_2(x)$. Given the linearity of equation (\ref{e:equation2}), the differences $\dot{w}(x) = \dot{w}_2(x) - \dot{w}_1(x)$ and $\dot{r}(x) = \dot{r}_2(x) - \dot{r}_1(x)$ satisfy $0= (\nabla \dot{w}) (x) + \mathcal{C}(x) \dot{r}(x)$. 

We next show $\dot{r}(x) = 0$, which then implies $(\nabla \dot{w}) (x) =0$. In order to show $\dot{r}(x) = 0$, consider:
\begin{equation*}
\int_X \dot{r}(x) \cdot \mathcal{C}(x) \dot{r} (x ) dx = - \int_X \dot{r}(x) \cdot (\nabla \dot{w}) (x) dx = \int_X (\nabla \dot{r}(x)) \cdot \dot{w} (x )  dx
 - \hspace{-0.1 cm} \int_{\partial X} \hspace{-0.05 cm} (\dot{r}(x) \dot{w} (x )) \cdot n(x) d x = 0
 \end{equation*}
 Since $\mathcal{C}(x)$ is positive definite, this implies $\dot{r}(x) = 0$ for all $x$, which concludes the proof.

\subsection{Proof to Theorem \ref{eq:changeinmearnings}} \label{a:poissoneq}

\noindent In order to show uniqueness, consider equations (\ref{e:csgen1}) and (\ref{e:csgen2}) in generalized form:
\begin{align}
\nabla \cdot \big( \mathcal{C}_f(x) \nabla \dot{w} (x ) \big) & = k_1(x) \hspace{2.84 cm} \text{ for  $x \in X$} \tag{\ref{e:csgen1}} \\
 \big(  \mathcal{C}_f(x) \nabla \dot{w} (x ) \big)\cdot n(x) & = k_2(x)  \hspace{2.82 cm} \text{ for $x \in \partial X$} \tag{\ref{e:csgen2}}
\end{align}
where $\mathcal{C}_f(x) = \mathcal{C}^{-1}(x) f(x) $ and $k_1(x) := \nabla \cdot \big( \mathcal{C}_f(x) \dot{\mathcal{A}} (x)  \big)$ and $k_2(x) :=\big(\mathcal{C}_f(x) \dot{\mathcal{A}}(x) \big)  \cdot n(x)$.

Suppose, by contradiction that we have two distinct solutions to equations (\ref{e:csgen1}) and (\ref{e:csgen2}), say $\dot{w}_1(x)$ and $\dot{w}_2(x)$. Given the linearity of both conditions the difference $\dot{w}(x) = \dot{w}_2(x) - \dot{w}_1(x)$ then also satisfies (\ref{e:csgen1}) and (\ref{e:csgen2}) with $k_1(x) = k_2(x) = 0$.

We multiply equation (\ref{e:csgen1}) by $\dot{w}(x)$ and integrate by parts so that:
\begin{equation*}
0 = \int_X \dot{w}(x) \nabla \cdot  ( \mathcal{C}_f(x) \nabla \dot{w} (x ) ) dx = \hspace{-0.1 cm} \int_{\partial X} \hspace{-0.05 cm} \dot{w}(x) \mathcal{C}_f(x) \nabla \dot{w} (x ) \cdot n(x) d x  - \int_X (\nabla \dot{w}(x))^\top ( \mathcal{C}_f(x) \nabla \dot{w} (x ) ) dx
\end{equation*}
Since the boundary term is also equal to zero by equation (\ref{e:csgen2}) with $k_2(x) = 0$, we obtain that 
\begin{equation*}
0 = \int_X (\nabla \dot{w}(x))^\top \mathcal{C}_f(x) \nabla \dot{w} (x )  dx \hspace{0.03 cm}.
\end{equation*}
Since $\mathcal{C}_f(x)$ is positive definite, this implies $\dot{w}(x) = 0$ for all $x$, which implies $\dot{w}_2(x) = \dot{w}_1(x)$, which concludes the proof.

In order to establish existence of a solution, we first provide a necessary condition. A necessary condition for a function satisfying equations (\ref{e:csgen1}) and (\ref{e:csgen2}) to exist is that $k_1$ and $k_2$ satisfy:
\begin{equation}
\int_X k_1(x) dx = \int_{\partial X} k_2(x) dx .
\end{equation}
In order to understand why, we note that:
\begin{equation}
\int_X k_1(x) dx = \int_X \nabla \cdot \big( \mathcal{C}_f(x) \nabla \dot{w} (x ) \big)  dx = \int_{\partial X} \big( \mathcal{C}_f(x) \nabla \dot{w} (x ) \big) \cdot n(x) dx = \int_{\partial X} k_2(x) dx
\end{equation}
where the first equality follows by equation (\ref{e:csgen1}), the second equality by integration by parts, and the third equality follows from equation (\ref{e:csgen2}). 

The necessary condition is also sufficient for existence. For more discussion, see Chapter 6 in \citet{Evans:2022} on the theory of second-order uniformly elliptic operators.

\subsection{Proof to Theorem \ref{t:alternchar}} \label{a:alternativechar}

To prove the result, we consider a first-order variation of some function $\phi(x)$ to $\phi(x) + \varepsilon \upsilon(x)$ for some differentiable function $\upsilon$. Any variation around the optimizer should yield a strictly lower objective value. For brevity, let $\mathcal{C}_f(x) = \mathcal{C}^{-1}(x) f(x)$, which is symmetric because $\mathcal{C}^{-1}(x)$ is symmetric. We first evaluate the term under the first-order variation as:
\begin{equation*}
\int  \big( \dot{\mathcal{A}}(x) - \nabla \phi(x) - \varepsilon \nabla \upsilon(x))^\top \mathcal{C}_f(x) ( \dot{\mathcal{A}}(x) - \nabla \phi(x) - \varepsilon \nabla \upsilon(x) \big) dx .
\end{equation*}
Letting $\xi(x) := \dot{\mathcal{A}}(x) - \nabla \phi(x)$, the objective can be rewritten as:
\begin{align*}
\int \Big( \xi(x)^\top \mathcal{C}_f(x) \xi(x) -2  \varepsilon \nabla \upsilon(x)^\top \mathcal{C}_f(x) \xi(x)  + \varepsilon^2 \nabla \upsilon(x)^\top \nabla \upsilon(x) \Big) dx . 
\end{align*}
The first term is independent of the variation, while the final term is of second-order importance. For the function $\phi(x)$ to be optimal, it has to be that the change in the objective with respect to $\varepsilon$ is zero  for any $\upsilon(x)$, or:
\begin{equation}
0 = \varepsilon \int_X \nabla \upsilon(x)^\top \mathcal{C}_f(x) \xi(x) dx = \varepsilon \int_{\partial X}  \upsilon(x) \mathcal{C}_f(x) \xi(x) \cdot n(x) dx - \varepsilon \int_X \upsilon(x) \nabla \cdot ( \mathcal{C}_f(x) \xi(x) )dx
\end{equation}
where the second equality follows from integration by parts. 

For the change in the objective value to be zero for any function $\upsilon(x)$, we thus require that:
\begin{align*}
0 & = \nabla \cdot ( \mathcal{C}_f(x) \xi(x) )  \hspace{3.84 cm} \text{ for  $x \in X$} \\
0 & = \big( \mathcal{C}_f(x) \xi(x) \big) \cdot n(x)  \hspace{3.3 cm} \text{ for $x \in \partial X$}
\end{align*}
which are the same conditions as for the change of equilibrium earnings to technological change as in \Cref{eq:changeinmearnings}. This implies that the solution to quadratic minimization problem $\varphi^*$ indeed provides the change in the marginal earnings schedule as $\nabla \varphi^* = \nabla \dot{w}$.

\subsection{Proof to Proposition \ref{p:sylvester}} \label{a:sylvester}

We complete the proof to Proposition \ref{p:sylvester} by showing that the solution to the Sylvester equation is unique.

\vspace{0.4 cm}
\noindent We are interested in solving the Sylvester equation (\ref{eq:syl1}), where $\dot{R}$ is the unknown, $\Sigma$ is symmetric and given, and $\dot{\Sigma}$ is given and arbitrary. 

\begin{claim}
The labor reallocation matrix $\dot{R}$ that solves the Sylvester equation (\ref{eq:syl1}) exists, is unique, and is antisymmetic.
\end{claim}

\begin{proof}
Regarding existence and uniqueness, there is general theory (see, for example, \citet{Bartels:1972}) saying that $A X + X B = C$ has a unique solution $X$ provided that $A$ and $-B$ do not share eigenvalues. In our case, $A = B = \Sigma$ is positive definite, meaning that all eigenvalues of $A$ and $B$ are positive. 

In order to show that $\dot{R}$ is antisymmetric, we show that when $A$ and $B$ are symmetric, while $C$ is antisymmetric, then $X$ is antisymmetric. This suffices since $A = B= \Sigma$ is symmetric, while $ \dot{\Sigma} - \dot{\Sigma}^\top$ is antisymmetric since $(\dot{\Sigma} - \dot{\Sigma}^\top)^\top = - (\dot{\Sigma} - \dot{\Sigma}^\top)$. 

We next show that $X$ is antisymmetric, or $X = - X^\top$. By uniqueness, it suffices to show that $Y = -X^\top$ solves $A Y + Y B = C$ when $X$ solves $A X + X B = C$. In order to see this, note that:
\begin{equation}
(AY + YB)^\top = -(A X^\top + X^\top B)^\top = - (XA + X B ) = - C 
\end{equation}
where the first equality follows by the definition of $Y$, the second equality uses the symmetry of $A$ and $B$, and the third equality uses that $A = B$ and $A X + X B = C$. By transposing this equation, using that $C$ is antisymmetric, one obtains $AY + YB = - C^\top = C$. This means that $Y = X$ or, equivalently, $X = - X^\top$.
\end{proof}

\section{Additional Characterization}

In this appendix, we present additional results.

\subsection{Planning Problem and Competitive Equilibrium} \label{pf:prop:equilibrium}

The following lemma, which is well-known in the literature, shows the relation between equilibrium and the primal and dual problem.\footnote{A similar argument is presented in Proposition 2.3 of \citet{Galichon:2018} and discussed in \citet{Chade:2017} for one-to-one assignment problems, and in \citet{BTZ:2021} for multimarginal sorting problems.}



\vspace{0.05cm}
\begin{lemma} \label{prop:equilibrium}
Earnings schedule $w_t$, firm value function $v_t$, and assignment $\pi_t$ are an equilibrium if and only if the assignment $\pi_t$ solves the primal problem and $(w_t,v_t)$ solve the dual problem.
\end{lemma}

\vspace{0.02 cm}
\noindent The equilibrium assignment is thus the allocation of workers to jobs that maximizes total output. 

\vspace{0.4 cm}
\noindent We provide the proof to Lemma \ref{prop:equilibrium} by separately proving both directions.

\vspace{0.05cm}
\begin{lemma} \label{prop:equilibrium1}
Let $\pi$ solve the planning problem and let $(w,v)$ solve the dual problem. Then, wage schedule $w$, firm value function $v$, and assignment $\pi$ are an equilibrium.
\end{lemma}

\vspace{-0.15cm}
\begin{proof}
The proof follows from the constraints on the dual problem, or $w(x) \geq y(x,z)  - v (z)$, which imply: \vspace{-0.2 cm}
\begin{equation}
w (x) \geq  \max\limits_z \hspace{0.03 cm} \big(  y(x,z)  - v (z) \big). \label{eq:wproofeq}
\end{equation} 
In order to show that (\ref{eq:wproofeq}) holds with equality, suppose instead that $w (x) > \max\limits_z \hspace{0.03 cm} (  y(x,z) - v (z) )$. In this case, one could reduce $w(x)$ to $\max\limits_z \hspace{0.03 cm} (  y(x,z) - v (z) )$, which satisfies the constraints to the dual problem, and reduces the aggregate payments to workers and firms, thus contradicting that functions $(w,v)$ solve the dual problem. The argument for the firm profit function $v$ is identical. By Monge-Kantorovich duality, it follows that the resource constraint is satisfied.\end{proof}

\noindent Lemma \ref{prop:equilibrium1} implies that assignment function $\pi$, wage schedule $w$, and firm profit function $v$ are an equilibrium. Given this connection, we call $w$ the wage schedule and $v$ the firm value. 

Moreover, Lemma \ref{prop:equilibrium2} shows that the decentralized equilibrium coincides with the solution to the planning problem so the competitive equilibrium is efficient.

\vspace{0.05cm}
\begin{lemma} \label{prop:equilibrium2}
Let $(\pi,w,v)$ be an equilibrium. Then, $\pi$ solves the planning problem and $(w,v)$ solves the dual problem.
\end{lemma}

\begin{proof}
Since the assignment is feasible, it is a candidate solution to the primal planning problem. Following the worker and the firm problem:
\begin{equation}
w (x) \geq  \max\limits_z \hspace{0.03 cm} \big(  y(x,z)  - v (z) \big) \hspace{1.35 cm} \implies \hspace{1.25 cm} w (x) \geq   y(x,z)  - v (z) \hspace{1.25 cm} \forall \hspace{0.03 cm} z \in Z
\end{equation} 
and similarly
\begin{equation}
v (z) \geq  \max\limits_x \hspace{0.03 cm} \big(  y(x,z)  - w (x) \big) \hspace{1.35 cm} \implies \hspace{1.25 cm} v (z) \geq   y(x,z)  - w (x) \hspace{1.25 cm} \forall \hspace{0.03 cm} x \in X
\end{equation} 
it follows that the wage function and the firm profit function satisfy the dual constraints: $w(x) + v(z) \geq y(x,z)$ for all $(x,z) \in X \times Z$.

Finally, it follows from the resource constraint that the dual gap is equal to zero, and hence that the assignment $\pi$ solves the planning problem and the wage function and firm value function solve the dual problem. 
\end{proof}

\subsection{Changing the Distributions of Workers and Jobs} \label{a:distributionalch}

Thus far, we considered comparative statics with respect to changes in the production technology. In this section, we show that our comparative statics methodology extends to analyzing comparative statics with respect to changes in the distributions of workers and jobs. In this section, we illustrate this by providing comparative statics for a change in the worker skills distribution.\footnote{Adding changes in the job distributions requires additional notation, but is conceptually identical to introducing changes in the distribution of workers.}


\vspace{0.4 cm}
\noindent \textbf{Changing the Skill Distribution}. We first discuss how we formally change the distribution of workers. Parallel to our parameterization of the technology function by $t$ to characterize comparative statics for changes in technology, we parameterize the distribution of worker skills by $t$ to characterize comparative statics with respect to the worker distribution. 

We formalize the change in the worker skill distribution by introducing an assignment function, which we denote $M^{-1}_t(x)$, which states that a worker with skills $x$ under the original distribution $f$ has skills $M^{-1}_t(x)$ under the distribution of worker skills $f_t$. Equivalently, a worker with skills $\tilde{x}$ under the distribution $f_t$ has skills $M_t(\tilde{x})$ in the original distribution. The assignment function $M^{-1}_t$ assigns the initial skill distribution $f$ to the new skill distribution $f_t$ and thus satisfies:
\begin{equation}
\int_X \phi(x) f(x) dx = \int_X \phi(M_t(x)) f_t(x) dx  \label{eq:reallocationssdist}
\end{equation}
for all smooth functions $\phi$, similar to the feasibility constraint for assignment $\tau_t$ in equation (\ref{eq:probconstraint}).

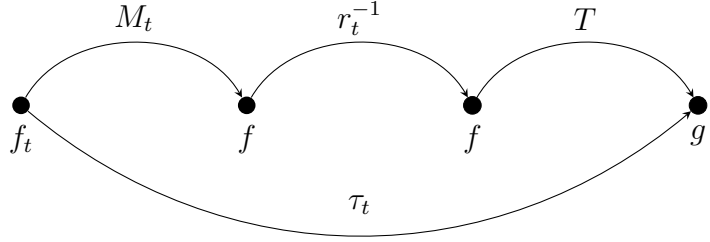
\begin{figure}[!t]
\begin{centering}
\begin{tikzpicture}

\node[circle,fill=black, minimum size=0.1pt,scale=0.6,label=below:{$f_t$}] (A) at  (0,0){};
\node[circle,fill=black, minimum size=0.1pt,scale=0.6,label=below:{$f$}] (B) at (3,0) {};
\node[circle,fill=black,draw, minimum size=0.1cm,scale=0.6,label=below:{$f$}] (C) at  (6,0) {};
\node[circle,fill=black,draw, minimum size=0.1cm,scale=0.6,label=below:{$g$}] (D) at  (9,0){};

\path[->,>=stealth, scale=125,every node/.style={font=\sffamily\small}] (A) edge[bend left=60] node [left] {} (B);
\path[->,>=stealth,every node/.style={font=\sffamily\small}] (B) edge[bend left=60] node [left] {} (C);
\path[->,>=stealth,every node/.style={font=\sffamily\small}] (C) edge[bend left=60] node [left] {} (D);
\path[->,>=stealth] (A) edge[bend right=40] node [yshift=0.2cm,below] {} (D);

\node[circle,draw=none, minimum size=0.1pt,scale=0.6,label=below:{$M_t$}] (F) at  (1.5,1.6){};
\node[circle,draw=none, minimum size=0.1pt,scale=0.6,label=below:{$r_t^{-1}$}] (G) at  (4.5,1.65){};
\node[circle,draw=none, minimum size=0.1pt,scale=0.6,label=below:{$T$}] (H) at  (7.5,1.55){};
\node[circle,draw=none, minimum size=0.1pt,scale=0.6,label=below:{$\tau_t$}] (I) at  (4.5,-0.93){};

        \end{tikzpicture} 
        
\par\end{centering}
\caption{Assignment with Changing Distributions of Workers}
\label{f:changingdist} {\scriptsize{}{}\vspace{0.2cm}
Figure \ref{f:changingdist} visualizes the assignment function $\tau_t$ between $f_t$ and $g_t$ as a composition of the maps $\tau$, $r_t^{-1}$ and $M_t$.}
\end{figure}


Given a distribution of workers parameterized by $t$, the optimal assignment pairs workers with skills $\tilde{x}$ distributed according to distribution $f_t$ to jobs $\tilde{z}$ distributed according to the distribution of job requirements $g$. This assignment function is $z = \tau_t(\tilde{x})$. We construct the optimal assignment as: $\tau_t(\tilde{x}) = \tau (r_t^{-1} (M_t(\tilde{x})))$. The map $M_t(\tilde{x})$ assigns type $\tilde{x}$ in the distribution $f_t$ to a type, say $\hat{x}$, in the original type distribution $f$, where $r_t^{-1}$ is the inverse reallocation which specifies whose job worker $\hat{x}$ will take under the new technology and worker distribution. The job that they will be replacing worker $x = r_t^{-1}(M_t(\tilde{x}))$ in is given by $z = \tau(x)$.  

\vspace{0.4 cm}
\noindent \textbf{Comparative Statics}. By the optimality condition to the profit maximization problem:
\begin{equation*}
(\nabla_1 y_t) (\tilde{x}, \tau_t(\tilde{x})) = (\nabla w_t) (\tilde{x}) 
\end{equation*}
for any distributions and production technology $t$, where the function $\tau_t$ assigns workers from the distribution $f_t$ to jobs in the distribution $g$. We evaluate the optimality condition to the profit maximization for worker $x = r_t^{-1}(M_t(\tilde{x}))$, or equivalently, $\tilde{x} = M_t^{-1} (r_t(x))$, to obtain:
\begin{equation}
(\nabla_1 y_t) (M_t^{-1} (r_t(x)), T (x)) = (\nabla w_t) (M_t^{-1} (r_t(x))) \label{e:focwithdistt} .
\end{equation}
We differentiate (\ref{e:focwithdistt}) with respect to a change in the production technology and evaluate at the initial technology in order to derive comparative statics with respect to changes in the worker skill distribution. 

\begin{theorem}{\textit{Comparative Statics for Changing Skill Distribution}.} \label{t:changeinmearningsdist}
The change of equilibrium earnings to technological change, which is denoted $\nabla \dot{w} (x )$, is uniquely determined as:
\begin{align}
\dot{f}(x) + \nabla \cdot \big( f(x)  \mathcal{C}^{-1}(x) \nabla \dot{w} (x ) \big) & = \nabla \cdot \big(f(x) \mathcal{C}^{-1}(x)  \dot{\mathcal{A}} (x) \big) \hspace{2.84 cm} \text{ for  $x \in X$} \label{e:csgen1dist} \\
 \big(f(x)  \mathcal{C}^{-1}(x) \dot{\mathcal{A}}(x) \big) \cdot n(x) & = \big(  f(x) \mathcal{C}^{-1}(x) \nabla \dot{w} (x ) \big) \cdot n(x)  \hspace{2.05 cm} \text{ for $x \in \partial X$} \label{e:csgen2dist} 
\end{align}
This is a system with a single unknown given by $\dot{w}$.
\end{theorem}


\begin{proof} Before analyzing the sensitivity of equation (\ref{e:focwithdistt}) with respect to a change in the worker distribution and the technology, recall the transport equations (\ref{eq:rearrgen}) and (\ref{eq:rearrbgen}) in Appendix \ref{a:transporteq}. Since $M_t$ assigns workers from the distribution $f_t$ to jobs in the distribution $g$, the transport equations are: 
\begin{align*}
0 & = \nabla \cdot (m_t(x) f_t(x)) + \dot{f}_t(x) \hspace{7.84 cm} \text{$ x \in X$}  \\
0 & = m_t(x) f_t(x) \cdot n(x)  \hspace{9.09 cm} \text{$x \in \partial X$} 
\end{align*}
where $(\nabla M_t)(x) m_t(x) - \dot{M}_t(x) = 0$, which evaluated at the original technology yields:
\begin{align}
0 & = \nabla \cdot (m(x) f(x)) + \dot{f}(x) \hspace{8.18 cm} \text{$ x \in X$}  \\
0 & = m(x) f(x) \cdot n(x)  \hspace{9.34 cm} \text{$x \in \partial X$} .
\end{align}
For the reallocation between workers and their replacement, the reallocation conditions (\ref{eq:rearr0}) and (\ref{eq:rearr0b}) hold. Given these transport equations, we can also write $\dot{M^{-1}}(x) = - m(x)$.\footnote{By definition $M_t(M_t^{-1}(x)) = x$, which we differentiate with respect to technlogical change and evaluate at the initial technology to obtain: $ \dot{M^{-1}}(x) + \dot{M} (x) = 0$ since $M(x) = x$ and hence $(\nabla M)(x) = I$ as well as $x = M^{-1}(x)$. As a consequence $ \dot{M^{-1}}(x) = - \dot{M} (x) = - m(x)$.}

\vspace{0.4 cm}
\noindent We differentiate the marginal earnings condition (\ref{e:focwithdistt}) with respect to a change in the production technology and the worker skill distribution and evaluate at the initial technology to derive:
\begin{align*}
& (\nabla_1 \dot{y} ) (x,\tau(x)) + (\nabla^2_{11} y) (x, \tau(x)) ( r(x) - m(x) )) = (\nabla_1 \dot{w})(x) + (\nabla^2 w)(x)  (r(x) - m(x) ) \hspace{0.05 cm}.
\end{align*}
By using the marginal earnings condition and equation (\ref{eq:difffocwrtx}) this expression simplifies to:
\begin{equation*}
(\nabla_1 \dot{w})(x) = (\nabla_1 \dot{y} ) (x,\tau(x)) - (\nabla^2_{12} y) (x, \tau(x)) \cdot \nabla \tau(x) ( r(x) - m(x) ) .
\end{equation*}
We multiply this expression by $f(x) \mathcal{C}^{-1}(x)$, where $\mathcal{C}^{-1}(x)$ is defined by (\ref{eq:C0}), to write:
\begin{equation*}
f(x) ( r(x) - m(x) )  = f(x) \mathcal{C}^{-1}(x) \big( (\nabla_1 \dot{y} ) (x,\tau(x)) - (\nabla_1 \dot{w})(x) \big)  .
\end{equation*}
We take the divergence on both sides, where we observe $\nabla \cdot (r(x) f(x)) = 0$ and $\nabla \cdot (m(x) f(x)) = \dot{f}(x)$, so that overall we obtain:
\begin{equation*}
\dot{f}(x)  = \nabla \cdot \big( f(x) \mathcal{C}^{-1}(x) \big( \dot{\mathcal{A}}(x) - (\nabla_1 \dot{w})(x) \big) \big) .
\end{equation*}

\end{proof}

\subsection{Intuition Identity Case} \label{a:intuitionidentitysylvester}


\noindent In order to build first intuition for the Helmholtz decomposition of technological change, consider the setting where the complementarity matrix $\Sigma$ is the identity matrix so that the initial output function is $\mathsf{y}(\tilde{x},x) = \tilde{x}_1 x_1 + \tilde{x}_2 x_2$. Moreover, suppose workers are uniformly distributed on a unit disk. In this setting, the change in marginal earnings in response to technological change (\ref{eq:intuition0by}) is:
\begin{equation}
\dot{\Sigma} x = \nabla \dot{w} (x) + \dot{r}(x)   .\label{eq:intuition0by1}
\end{equation}
We now decompose the technological change $\dot{\Sigma} x$ into a gradient and a divergence-free reallocation. We consider technological change $\dot{\mathsf{y}}(\tilde{x},x) = \dot{\alpha} \tilde{x}_1 x_1 + \dot{\beta} \tilde{x}_1 x_2 + \dot{\gamma} \tilde{x}_2 x_1 + \dot{\delta} \tilde{x}_2 x_2$, or:
\begin{equation}
\dot{\mathcal{A}}(x) = \dot{\Sigma} x = \Bigg(  \begin{matrix} \dot{\alpha} x_1 + \dot{\beta} x_2 \\ \dot{\gamma} x_1 + \dot{\delta} x_2  \end{matrix} \Bigg) . \label{eq:techchangesimp}
\end{equation}

The Helmholtz decomposition decomposes the technological change (\ref{eq:techchangesimp}) into a gradient and a divergence-free reallocation. First, observe that technological change $\dot{\Sigma} x$ is generally not a gradient (\ref{eq:techchangesimp}), and hence labor reallocates. In order for technological change (\ref{eq:techchangesimp}) to be the gradient of a function, its cross-derivatives have to agree, or equivalently, the matrix of technological change $\dot{\Sigma}$ has to be symmetric. When $\dot{\beta} \neq \dot{\gamma}$, the cross-derivatives do not align and hence labor reallocates. 

\vspace{0.3 cm}
\noindent \textit{Symmetric Technological Change}. First, consider the case where the technological change $\dot{\mathcal{A}}$ is the gradient of a function $(\dot{\beta} = \dot{\gamma})$. The marginal earnings change is $(\nabla \dot{w})(x) = \dot{\Sigma} x$, which changes earnings as $\dot{w}(x) =  \dot{\alpha} x_1^2 + \dot{\beta} x_1 x_2 + \dot{\delta} x^2_2$. An increase in within-task complementarities ($\dot{\alpha} > 0$ and $\dot{\delta} > 0$) increases the compensation for cognitive and manual skills, while $\dot{\beta} = \dot{\gamma}$ increases compensation for workers with correlated skills. Symmetric technological change does not change sorting, $\dot{r}(x) = 0$.
 
 \vspace{0.3 cm}
\noindent \textit{Antisymmetric Technological Change}. Second, consider the case where $\dot{\beta} = - \dot{\gamma} > 0$ and within-task complementarities do not change, or $\dot{\alpha} = \dot{\delta} = 0$. The direct effect of technological change $\dot{\mathcal{A}}(x)$ is not the gradient of a function but is instead written as the product of the complementarity matrix $\mathcal{C}(x) = \Sigma = I$, and a vector $\rho(x)f(x)$ that is divergence-free and parallel to the boundary. Since workers are  distributed uniformly on a unit disk $f(x) = 1$, the technological change (\ref{eq:intuition0by}) is:
\begin{equation*}
\dot{\beta} \Bigg(  \begin{matrix} \phantom{-} x_2 \\ - x_1 \end{matrix} \Bigg) = \dot{r}(x)   .
\end{equation*}
The reallocation is indeed divergence free $\nabla \cdot \dot{r}(x) = 0$ and parallel to the boundary as $\dot{r}(x) \cdot n(x) =  \dot{r}(x) \cdot (x_1,x_2) = 0$. A change of between-task complementarities in opposite directions ($\dot{\beta} = - \dot{\gamma}$) reallocates workers without a change in compensation.

\vspace{0.3 cm}
\noindent \textit{General Technological Change}. Finally, we combine the prior two cases to provide the Helmholtz decomposition of the general direct effect of technological change (\ref{eq:techchangesimp}). In order to split the technological change into a symmetric component $-$ the change in earnings $-$ and a divergence-free component, we note that the change in between-task complementarities can be written as:
\begin{equation*}
\dot{\beta} = \frac{\dot{\beta} + \dot{\gamma}}{2} + \frac{\dot{\beta} - \dot{\gamma}}{2}  \hspace{2.5 cm} \text{ and } \hspace{2.5 cm}  \dot{\gamma} = \frac{\dot{\beta} + \dot{\gamma}}{2} - \frac{\dot{\beta} - \dot{\gamma}}{2} .
\end{equation*}
Using this, we can write the matrix of technological change as:
\begin{equation}
\dot{\mathcal{A}}(x) = \Bigg(  \begin{matrix} \dot{\alpha}  &  \frac{\dot{\beta} + \dot{\gamma}}{2} \\ \frac{\dot{\beta} + \dot{\gamma}}{2}  & \dot{\delta}   \end{matrix} \Bigg)  \Bigg(  \begin{matrix} x_1 \\ x_2 \end{matrix} \Bigg) + \frac{\dot{\beta} - \dot{\gamma}}{2} \Bigg(  \begin{matrix} \phantom{-} x_2 \\ - x_1 \end{matrix} \Bigg) \; .
\end{equation}
The first term on the right is the change in marginal earnings, as discussed in the first case above. The change in within-task complementarity only affects worker earnings. The average technological change in between-task complementarity increases earnings for workers with correlated skills. The second term on the right side is the reallocation, as discussed in the second case. The difference in the change of between-task complementarities leads to labor reallocation.

\subsection{Second-Order Approximation} \label{a:2ndorderapp}

\noindent We want to find the reallocation of labor $r_t(x)$ that maximizes the increase in aggregate output. We consider all possible reallocations $r_t(x)$ that satisfy reallocation restrictions (\ref{eq:rearr}) and (\ref{eq:rearrb}), where $(\nabla_x r_t(x)) \upsilon_t(x) + \dot{r}_t(x) = 0$. We label the increase in aggregate output given a change $\upsilon_t(x)$ by $\mathcal{I}_t$:
\begin{equation*}
\mathcal{I}_t(\upsilon) = \int \mathsf{y}_t(r_t(x),x) f(x) dx - \int \mathsf{y}_t(x,x) f(x) dx \hspace{0.05 cm}.
\end{equation*}
The objective is to maximize the increase in aggregate output as $t \rightarrow 0$. At the initial technology $t=0$, we have $\mathcal{I}(\upsilon)= 0$ since $r(x) = x$.


\vspace{0.4 cm}
\noindent \textbf{First-Order Effect}. We first analyze how the change in aggregate production due to technological change varies with the choice of the reallocation $\upsilon$. That is, we evaluate:
\begin{equation*}
\frac{d \mathcal{I}_t}{d t} (\upsilon) =  \int \dot{\mathsf{y}}_t(r_t(x),x) f(x) dx + \int (\nabla_1 \mathsf{y}_t) (r_t(x),x) \dot{r}_t(x) f(x) dx - \int \dot{\mathsf{y}}_t(x,x) f(x) dx \hspace{0.05 cm}. 
\end{equation*}
Evaluating at the initial technology,  
\begin{equation*}
\frac{d \mathcal{I}_t}{d t} \bigg\vert_{t=0} (\upsilon) = \int (\nabla_1 \mathsf{y}) (x,x) \dot{r}(x) f(x) dx \hspace{0.05 cm}. 
\end{equation*}
From the optimality condition at the initial technology, $(\nabla_1 \mathsf{y})(x, x) = (\nabla w)(x)$, and since $\dot{r}(x) = - \upsilon(x)$, we can write the first-order change as:
\begin{equation*}
\frac{d \mathcal{I}_t}{d t} \bigg\vert_{t=0} (\upsilon) = \int_X \nabla w(x) \dot{r}(x) f(x) dx = - \int_{\partial X} w(x) \upsilon(x) f(x) \cdot n(x) dx + \int_X w(x) \nabla \cdot ( \upsilon(x) f(x) ) dx = 0 ,
\end{equation*}
where the final equality follows by the reallocation restrictions (\ref{eq:rearr0}) and (\ref{eq:rearr0b}). To a first-order, the change in aggregate output due to technological change due to the choice of $\upsilon$ equals zero. This is not surprising, it is simply a manifestation of the envelope theorem.

\vspace{0.4 cm}
\noindent \textbf{Second-Order Effect}. We next analyze how the change in output varies with the choice of the reallocation $\upsilon$ to a second-order. In order to do so, we examine the second-order effects:
\begin{align*}
\frac{d^2 \mathcal{I}_t}{d t^2} (\upsilon) & =  \int \ddot{\mathsf{y}}_t(r_t(x),x) f(x) dx +  \int (\nabla_1 \dot{\mathsf{y}}_t)(r_t(x),x) \dot{r}_t(x) f(x) dx \\ & +  \int (\nabla_1 \dot{\mathsf{y}}_t) (r_t(x),x) \dot{r}_t(x) f(x) dx + \int \dot{r}_t(x)^\top (\nabla_{11} \mathsf{y}_t) (r_t(x),x) \dot{r}_t(x) f(x) dx \\ & + \int (\nabla_1 \mathsf{y}_t) (r_t(x),x) \ddot{r}_t(x) f(x) dx - \int \ddot{\mathsf{y}}_t(x,x) f(x) dx \hspace{0.05 cm}. 
\end{align*}
Evaluating at the initial technology:
\begin{align}
\frac{d^2 \mathcal{I}_t}{d t^2}  \bigg\vert_{t=0} (\upsilon) & = 2 \int (\nabla_1 \dot{\mathsf{y}}) (x,x) \dot{r}(x) f(x) dx + \int \dot{r}(x)^\top (\nabla_{11} \mathsf{y}) (x,x) \dot{r}(x) f(x) dx \notag \\ & + \int (\nabla_1 \mathsf{y}) (x,x) \ddot{r}(x) f(x) dx \notag \\
& = - 2 \int (\nabla_1 \dot{\mathsf{y}}) (x,x) \cdot \upsilon(x) f(x) dx + \int \upsilon(x)^\top (\nabla_{11} \mathsf{y}) (x,x) \upsilon(x) f(x) dx \notag  \\ & + \int (\nabla w) (x) \ddot{r}(x) f(x) dx \label{eq:soa1}  \hspace{0.05 cm} ,
\end{align}
which features the second-order change in the reallocation $\ddot{r}(x)$. We next establish that the second-order aggregate output gain is independent of this second-order response of the assignment.

In order to simplify the second-order gain in aggregate production, we vary the reallocation definition $0 = (\nabla_x r_t(x)) \upsilon_t(x) + \dot{r}_t(x)$ with respect to technology:
\begin{equation*}
(\nabla_x \dot{r}_t(x)) \upsilon_t(x) + (\nabla_x r_t(x)) \dot{\upsilon}_t(x) + \ddot{r}_t(x) = 0 \;.
\end{equation*}
Evaluated at the initial technology,
\begin{equation*}
- (\nabla_x \upsilon(x)) \upsilon(x) + \dot{\upsilon}(x) + \ddot{r}(x) = 0 ,
\end{equation*}
where we use that $\dot{r}(x) = - \upsilon(x)$. Moreover, we differentiate both of the reallocation conditions (\ref{eq:rearr}) and (\ref{eq:rearrb}) with respect to technology to write:
\begin{align*}
0 & = \nabla \cdot (\dot{\upsilon}_t(x) f(x)) \hspace{3.65 cm} \text{$ x \in X$} \\
0 & = \dot{\upsilon}_t(x) f(x) \cdot n(x)  \hspace{3.5 cm} \text{$x \in \partial X$} 
\end{align*}
so that $\dot{\upsilon}_t(x)$ is divergence-free and parallel to the boundary. At the initial technology this gives $0 = \nabla \cdot (\dot{\upsilon}(x) f(x))$ and $0 = \dot{\upsilon}(x) f(x) \cdot n(x)$. 

We can use that $\dot{\upsilon}$ is divergence-free and parallel to the boundary in order to simplify the final term in equation (\ref{eq:soa1}) as:
\begin{align*}
\int (\nabla w) (x) \ddot{r}(x) f(x) dx & = \int (\nabla w) (x) (\nabla \upsilon)(x) \upsilon(x) f(x) dx - \int (\nabla w) (x) \dot{\upsilon}(x) f(x) dx ,
 \end{align*}
where the second term is equal to zero through integration by parts:
\begin{equation*}
\int_X (\nabla w) (x) \dot{\upsilon}(x) f(x) dx =  \int_X w (x) \nabla \cdot (\dot{\upsilon}(x) f(x)) dx - \int_{\partial X} w(x) (\dot{\upsilon}(x) f(x)) \cdot n(x) dx = 0 \;.
\end{equation*}
We next simplify the first term. We first consider the identity $\nabla (\nabla w(x) \cdot \upsilon (x)) = (\nabla^2 w(x)) \upsilon (x) + (\nabla w)(x) (\nabla \upsilon) (x)$ so that we can write the first term as:
\begin{equation*}
\int (\nabla w) (x) (\nabla \upsilon)(x) \upsilon(x) f(x) dx = \int \nabla (\nabla w(x) \cdot \upsilon (x)) \upsilon(x) f(x) dx - \int (\nabla^2 w)(x) \upsilon (x) \cdot \upsilon(x) f(x) dx \hspace{0.05 cm}.
\end{equation*}
Since $\nabla (\nabla w(x) \cdot \upsilon (x))$ is a gradient with respect to $x$, thus the first integral on the right gives no contribution from integration by parts. In summary, the final term in the second-order change of aggregate output (\ref{eq:soa1}) simplifies to:
\begin{align*}
\int (\nabla w) (x) \ddot{r}(x) f(x) dx & = - \int \upsilon(x)^\top (\nabla^2 w)(x) \upsilon (x)  f(x) dx \hspace{0.05 cm}.
 \end{align*}
In turn, this implies that the second-order change in aggregate output can be written as:
\begin{align*}
\frac{d^2 \mathcal{I}_t}{d t^2}  \bigg\vert_{t=0} (\upsilon) & = - 2 \int (\nabla_1 \dot{\mathsf{y}}) (x,x) \cdot \upsilon(x) f(x) dx  + \int \upsilon(x)^\top \big( (\nabla^2_{11} \mathsf{y}) (x,x) - (\nabla^2 w)(x) \big) \upsilon(x) f(x) d x \hspace{0.05 cm} .
\end{align*}
which can be further simplified using the definition of technological change $\dot{\mathcal{A}}$ and equations (\ref{eq:difffocwrtx}) and (\ref{eq:C0}):
\begin{align*}
\frac{d^2 \mathcal{I}_t}{d t^2}  \bigg\vert_{t=0} (\upsilon) & = - 2 \int \dot{\mathcal{A}}(x) \cdot \upsilon(x) f(x) dx - \int \upsilon(x)^\top \mathcal{C}(x) \upsilon(x) f(x) d x 
\hspace{0.05 cm} .
\end{align*}

\vspace{0.4 cm}
\noindent \textbf{Maximizing the Output Gain}. The realized aggregate output gain due to technological change with a choice of reallocation $\upsilon$ is given by $\mathcal{I}_t(v) \approx \mathcal{I}(v) + t \frac{d \mathcal{I}_t}{d t} \big\vert_{t=0} (\upsilon) + \frac{t^2}{2} \frac{d^2 \mathcal{I}_t}{d t^2} \big\vert_{t=0} (\upsilon) + O(t^3)$. The first-order effects equal zero, the second-order approximation yields: 
\begin{align}
\frac{1}{t^2} \Big[ \mathcal{I}_t(v) - \mathcal{I}(v) \Big]\approx \frac{1}{2} \frac{d^2 \mathcal{I}_t}{d t^2}  \bigg\vert_{t=0} (\upsilon) & = - \int \dot{\mathcal{A}}(x) \cdot \upsilon(x) f(x) dx - \frac{1}{2} \int \upsilon(x)^\top \mathcal{C}(x) \upsilon(x) f(x) d x \notag \\
& = \int \dot{\mathcal{A}}(x) \dot{r}(x) f(x) dx - \frac{1}{2} \int \dot{r}(x)^\top \mathcal{C}(x) \dot{r}(x) f(x) d x, \label{eq:soa4}
\end{align}
where the second equality follows as $\upsilon(x) = - \dot{r}(x)$. Proposition \ref{p:socharacterization} shows that the optimal local reallocation maximizes aggregate output to the second order in response to technological change.

\begin{proposition} \label{p:socharacterization}
The reallocation that maximizes the increase in aggregate output up to a second order $\dot{\tilde{r}}(x)$ is given by: 
\begin{equation*}
\dot{\mathcal{A}}(x) = (\nabla \dot{w})(x) + \mathcal{C}(x) \dot{\tilde{r}}(x) \tag{\ref{e:equation2}} ,
\end{equation*}
where the change in the gradient of earnings $(\nabla \dot{w})(x)$ uniquely solves:
\begin{align}
\nabla \cdot \big(  \mathcal{C}_f(x) \nabla \dot{w} (x ) \big) & = \nabla \cdot \big( \mathcal{C}_f(x)  \dot{\mathcal{A}} (x) \big) \hspace{2.84 cm} \text{ for  $x \in X$} \tag{\ref{e:csgen1}} \\
\big( \mathcal{C}_f(x) \dot{\mathcal{A}}(x) \big) \cdot n(x) & = \big(  \mathcal{C}_f (x) \nabla \dot{w} (x ) \big) \cdot n(x)  \hspace{2.05 cm} \text{ for $x \in \partial X$} \tag{\ref{e:csgen2}}
\end{align}
\end{proposition}

\noindent Proposition \ref{p:socharacterization} establishes that the reallocation introduced in equation (\ref{e:equation2}), together with the solution to the generalized Poisson equation, is the reallocation maximizing aggregate output up to second order. 

\begin{proof}
First, by construction, $\dot{\tilde{r}}(x) f(x)$ is divergence free and parallel to the boundary. It follows from equation (\ref{e:equation2}) that $\dot{\tilde{r}}(x) =\mathcal{C}^{-1}(x) \dot{\mathcal{A}}(x) - \mathcal{C}^{-1}(x) (\nabla \dot{w})(x)$ and, thus, by substitution into equations (\ref{e:csgen1}) and (\ref{e:csgen2}) it follows that $\dot{\tilde{r}}(x)f(x)$ is divergence free and parallel to the boundary. 

We remain to show that the candidate optimal reallocation maximizes aggregate output up to second order. Given the specification of the candidate local reallocation (\ref{e:equation2}), by multiplying by any feasible local reallocation $\dot{r}(x)$, we write:
\begin{equation*}
 \int (\nabla \dot{w})(x) \dot{r}(x) f(x) dx  =\int \dot{\mathcal{A}}(x) \dot{r}(x) f(x) dx - \int \dot{\tilde{r}}(x)^\top \mathcal{C}(x) \dot{r}(x) f(x) dx .
\end{equation*}
Since the optimal reallocation is divergence free and parallel to the boundary, the left-hand side can shown to equal zero through integrated by parts.

As a result, it follows that the increase of aggregate output up to a second order (\ref{eq:soa4}) is:
\begin{align}
\frac{1}{2} \frac{d^2 \mathcal{I}_t}{d t^2}  \bigg\vert_{t=0} (\upsilon) & = \int \dot{\mathcal{A}}(x) \dot{r}(x) f(x) dx - \frac{1}{2} \int \dot{r}(x)^\top \mathcal{C}(x) \dot{r}(x) f(x) d x \notag \\ 
& = \int \tilde{b}(x)^\top \frac{\mathcal{C}(x)}{f(x)} b(x) dx - \frac{1}{2} \int b(x)^\top \frac{\mathcal{C}(x)}{f(x)} b(x) d x \notag \\ 
& = - \frac{1}{2} \hspace{-0.1 cm} \int (b(x) - \tilde{b}(x))^\top \frac{\mathcal{C}(x)}{f(x)} (b(x) - \tilde{b}(x)) dx + \frac{1}{2} \int \tilde{b}(x)^\top \frac{\mathcal{C}(x)}{f(x)} \tilde{b}(x) d x \label{eq:soa5}
\end{align}
where the second equality follows by defining $b(x) := \dot{r}(x) f(x)$ and $\tilde{b}(x) := \dot{\tilde{r}}(x) f(x)$, and the third equality follows by completing the square and recalling $\mathcal{C}_f(x)$. Given the quadratic objective and since $\mathcal{C}$ is positive definite following the equilibrium conditions in Section \ref{s:eqconditions} and the definition of $\mathcal{C}$ in equation (\ref{eq:C0}), it follows that the optimum is to set $b(x)$ to $\tilde{b}(x)$, that is, $\dot{r}(x) = \dot{\tilde{r}}(x)$.\end{proof}

\vspace{0.2 cm}
\noindent Following the proof to Proposition \ref{p:socharacterization}, setting $b$ to $\tilde{b}$ in equation (\ref{eq:soa5}), the second-order increase in aggregate output is given by: 
\begin{align}
\mathcal{I}_t(v^*) \approx  \frac{1}{2} \int \tilde{b}(x)^\top \frac{\mathcal{C}(x)}{f(x)} \tilde{b}(x) d x + O(t^3) = \frac{1}{2} \int \dot{r}(x)^\top \mathcal{C}(x) \dot{r}(x) f(x) d x + O(t^3) \;.
\end{align}
which is positive since $\mathcal{C}(x)$ is positive semidefinite.

\subsection{Proposition \ref{t:flow}} \label{a:largechange}

We next show that assignment $\tau_t$ is the equilibrium assignment and $w_t$ is the equilibrium earnings schedule given technology $t$.

\vspace{0.4 cm}
\noindent Equation (\ref{eq:flow}) implies that $(\nabla_1 y_t)(x, \tau_t(x)) = (\nabla w_t)(x)$ holds for each technology $t$, which is a necessary condition for optimality. A sufficient condition is that $w_t$ comes from a dual pair $(w_t, v_t)$. Given the Ma-Trudinger-Wang condition and the convexity assumption on the support, it suffices to verify that the second order condition $(\nabla_{11} y_t)(x, \tau_t(x))-(\nabla^2 w_t)(x)>0$ holds, as shown by Theorem $12.46$ in \citet{Villani:2009}. Since the inequality holds at time $t=0$ by assumption, and the function is continuous in $t$, we conclude that for $T$ sufficiently small, the solution $(w_t, \tau_t)$ is in fact optimal on $[0,T]$.

\subsection{Unidimensional Case} \label{a:unidimensional}


We discuss the special environment where the distributions of workers and jobs are unidimensional $(d=1)$. In this case, the restrictions on the reallocation given by (\ref{eq:rearr0}) and (\ref{eq:rearr0b}) simplify to:
\begin{align}
0 & = \frac{d}{dx} (\dot{r}(x) f(x)) \hspace{3.65 cm} \text{$ x \in X$} \tag{\ref{eq:rearr0}} \\
0 & = \dot{r}(x) f(x) \cdot n(x)  \hspace{3.4 cm} \text{$x \in \partial X$} \tag{\ref{eq:rearr0b}}
\end{align}

\noindent The reallocation condition on the interior of worker skills (\ref{eq:rearr0}) implies that $\dot{r}(x) f(x)$ is constant. In the unidimensional case, workers either all replace workers with more skills or all replace workers with fewer skills. By the boundary condition, it follows that the constant has to be zero. The only way in which workers either all replace workers with more skill or all replace workers with fewer skills, which is consistent with the boundary condition, is that no worker moves. That is, sorting does not change in response to a technology change. Alternatively, technological change is always symmetric in the unidimensional setting, and hence by Proposition \ref{p:symmetricse} it follows that there is no reallocation in response to technological change.

How does the earnings schedule change in the setting with unidimensional workers and jobs? From the marginal earnings condition (\ref{e:equation2}) it follows that technological change completely passes through to changes in marginal earnings, or $(\nabla \dot{w})(x ) = (\nabla \dot{\mathsf{y}}_1) (x , x)$. For example, in case of the bilinear technology $y_t(x,z) = \gamma_t x z$ with $\gamma_t > 0$, sorting is positive since the production technology is supermodular, and the corresponding change in marginal earnings is $(\nabla \dot{w})(x ) = \dot{\gamma} \tau(x)$, showing that technological change $\dot{\gamma} > 0$ increases earnings differences as it increases the earnings gradient for all worker types $x$.


In sum, this section shows that smooth reallocation in response to technological change is a multidimensional phenomenon. In the unidimensional case, the feasibility condition forces $\dot{r} = 0$, implying only earnings respond to technological progress. In contrast, for dimensions $d\geq2$ the space of feasible reallocations is non-trivial.\footnote{Put differently, there is a sharp contrast between unidimensional assignment problems, where } This implies that in one dimension technological change can only change earnings but in higher dimensions it can also change sorting, reallocating workers across jobs along the many directions allowed by feasibility.

\subsection{Computational Approach} \label{a:computation}

In this appendix, we discuss the computational approach to characterizing equilibrium comparative statics for multidimensional sorting. We use this numerical approach for our quantitative analysis of the consequences of technological change in Section \ref{s:quant}.\footnote{We validate our numerical implementation by comparing our numerical solution against the closed-form solution in Section \ref{s:closedform}.}

\vspace{0.4 cm}
\noindent One advantage of the characterization of the change in earnings in Theorem \ref{t:alternchar} is that it provides a direct approach to computationally characterize the earnings change. Specifically, we formulate the quadratic minimization problem (\ref{eq:quadproblem}) as a constrained minimization problem solving:
\begin{equation*}
\min_{v} \int ( \dot{\mathcal{A}}(x) - v(x) )^\top \mathcal{C}_f(x) ( \dot{\mathcal{A}}(x) - v(x) )  dx  \label{eq:quadproblemnum} ,
\end{equation*}
where $v$ is restricted to be the gradient of some function. The condition that a vector $v(x)$ is the gradient of a function $\varphi(x)$ is the restriction that $\frac{\partial v_i}{\partial x_j} (x) = \frac{\partial v_j}{\partial x_i} (x)$ as $v_i(x) = \frac{\partial \varphi}{\partial x_i}$ and $v_j(x) = \frac{\partial \varphi}{\partial x_j}$. Hence the restriction is that the cross-derivatives are equal (by Lemma \ref{l:poincare} in \Cref{a:gradient}). In order to implement this numerically, we introduce a penalized problem. We construct an auxiliary problem of choosing $v$ to solve: 
\begin{equation}
\min\limits_{v} \frac{1}{2}  \int ( \dot{\mathcal{A}}(x) - v(x) )^\top \mathcal{C}_f(x) ( \dot{\mathcal{A}}(x) - v(x) )  dx + \frac{\Psi}{2}  \int \sum_{i,j} \left( \frac{\partial v_i}{\partial x_j} (x) - \frac{\partial v_j}{\partial x_i} (x) \right)^2 dx \label{eq:augquadpap} ,
\end{equation}
for a large penalization parameter $\Psi > 0$. Reallocation is then constructed using equation (\ref{e:equation2}). 

An alternative way to formulate the restriction that the derivatives of the vector $v$ are equal is by introducing the matrix $\Xi$ and the Jacobian matrix for the vector $v$:
\begin{equation*}
\Xi = \left[ \begin{matrix} \phantom{-}0\phantom{-} & 1 \\ -1\phantom{-} & 0 \end{matrix} \right] \hspace{2.4 cm} \text{ and } \hspace{2.4 cm} J_v = \left[ \begin{matrix} \frac{\partial v_1}{\partial x_1} & \frac{\partial v_1}{\partial x_2} \\ \frac{\partial v_2}{\partial x_1} & \frac{\partial v_2}{\partial x_2} \end{matrix} \right] 
\end{equation*}
so that the penalty term can be written as: $\frac{\beta}{2}  \int | \text{Tr}\left( \Xi J_v \right) |^2 dx$.

In order to illustrate the numerical implementation, we consider a square grid of worker types. In both dimensions, we consider $n$ distinct types, such the number of grid points is $N = n^2$. We want to characterize the minimizer $z_v$:
\begin{equation*}
z_v = \big( v(x_{1,1}),v(x_{1,2}), \dots, v(x_{1,n}), \dots, v(x_{n,n}) \big),
\end{equation*}
where $x_{i,j}$ refers to the $i$-th point in the $x$ dimension and the $j$-th point in the $y$-dimension. Each $v(x_{i,j})$ is a two-dimensional vector that consists of $v(x_{i,j}) = (v_1(x_{i,j}),v_2(x_{i,j}) )$. Hence, the vector $z_v$ is of dimension $2N \times 1$. We construct the analogous object with respect to the vector $\dot{\mathcal{A}}$ so that $z_{\dot{\mathcal{A}}}$ is also of dimension $2 N \times 1$.

We discretize the augmented objective function (\ref{eq:augquadpap}). We separately discretize the original objective function and the penalty function. We first discretize the original objective:
\begin{align*}
\frac{1}{2}  \int_X ( \dot{\mathcal{A}}(x) - v(x) )^\top \mathcal{C}_f(x) ( \dot{\mathcal{A}}(x) - v(x) )  dx & \approx \frac{1}{2} \sum_{i,j = 1}^n ( \dot{\mathcal{A}}(x_{i,j}) - v(x_{i,j}) )^\top \mathcal{C}_f(x_{i,j}) ( \dot{\mathcal{A}}(x_{i,j}) - v(x_{i,j}) )  \\ & = \frac{1}{2} (z_{\dot{\mathcal{A}}} - z_v)^\top \mathcal{B}  (z_{\dot{\mathcal{A}}} - z_v) 
\end{align*}
where the final expression follows from defining the matrix $\mathcal{B}$ as:
\begin{equation*}
\mathcal{B} := \left[
\begin{tikzpicture}[baseline=0cm]   
  \matrix [mymatrix] (m)  
    {
   \mathcal{C}_f(x_{1,1}) & 0  & 0 & \cdots  \\
   0 & \mathcal{C}_f(x_{1,2}) & 0 & \cdots  \\
    0 & 0 & \ddots & \cdots  \\
    \cdots & \cdots & \cdots & \mathcal{C}_f(x_{n,n})  \\
    }; ,
\end{tikzpicture} \right]
\end{equation*}
which is a symmetric matrix since $\mathcal{C}_f$ is a symmetric matrix.

We next consider the discretization of the penalty function. In order to discretize the penalty function we discretize the derivatives. We approximate the derivatives using upward derivatives:
\begin{equation*}
\frac{\partial v_1}{\partial x_2}(x_{i,j}) = \frac{v_1(x_{i,j+1}) - v_1(x_{i,j})}{\Delta_2} \hspace{1.5 cm} \text{ and } \hspace{1.5 cm} \frac{\partial v_2}{\partial x_1}(x_{i,j}) = \frac{v_2(x_{i+1,j}) - v_2(x_{i,j})}{\Delta_1} \;. 
\end{equation*}

We next discretize the penalty function given by $\frac{\beta}{2}  \int | \text{Tr}\left( \Xi J_v \right) |^2 dx$. The observation that the cross-derivatives need to be equal imposes on the two-dimensional grid that the change in the function value is identical across different paths, or:
\begin{equation*}
v_1(x_{i,j}) + v_2(x_{i+1,j}) = v_2(x_{i,j}) + v_1(x_{i,j+1}) .
\end{equation*}
As a result, we can approximate the penalty function as:
\begin{equation*}
\int_X | \text{Tr}\left( \Gamma J_v \right) |^2 dx \approx \sum_{i,j = 1}^n | v_1(x_{i,j+1}) - v_1(x_{i,j}) - v_2(x_{i+1,j}) + v_2(x_{i,j})  |^2 ,
\end{equation*}
a sum over $N$ different equations. 

We translate the penalty function to matrix notation by introducing the matrix $\mathcal{Q}$. Matrix $\mathcal{Q}$ is dimension $N \times 2N$ such that we retrieve all $N$ summands in $\mathcal{Q} \times z_v$. Matrix $\mathcal{Q}$ consists of $N \times N$ blocks of $1 \times 2$ row vectors. Each row vector $\mathcal{Q}_{(i,j),(k,l)}$ describes how $x_{k,l}$ features in the summand corresponding to $x_{i,j}$. From the summand for $x_{i,j}$ given by $v_1(x_{i,j+1}) - v_1(x_{i,j}) - v_2(x_{i+1,j}) + v_2(x_{i,j})$ it follows that $\mathcal{Q}$ is very sparse as the only non-zero elements in the row corresponding to $x_{i,j}$ are given by:
\begin{align*}
\mathcal{Q}_{(i,j),(i,j)}^{1} & = -1  \\
\mathcal{Q}_{(i,j),(i,j)}^{2} & = +1 \\
\mathcal{Q}_{(i,j),(i,j+1)}^{1} & = +1 \\
\mathcal{Q}_{(i,j),(i+1,j)}^{2} & = -1 ,
\end{align*}
where the superscript indicates the position in the row vector $\mathcal{Q}_{(i,j),(k,l)}$. Matrix $\mathcal{Q}$ does not contain non-zero elements for the rows corresponding to points at the boundary of the grid, that is, the rows corresponding to $x_{i,n}$ and $x_{n,j}$ for all $i$ and $j$. In summary, the matrix $\mathcal{Q}$ is given by:
\begin{equation*}
\mathcal{Q} = \left[
\begin{tikzpicture}[baseline=0cm]   
  \matrix [mymatrix] (m)  
    {
   \mathcal{Q}_{(1,1),(1,1)} & \mathcal{Q}_{(1,1),(1,2)}  & \mathcal{Q}_{(1,1),(1,3)} & \cdots  \\
   \mathcal{Q}_{(1,2),(1,1)} & \mathcal{Q}_{(1,2),(1,2)}  & \mathcal{Q}_{(1,2),(1,3)} & \cdots  \\
    \cdots & \cdots & \cdots & \cdots  \\
    \cdots & \cdots & \cdots & \mathcal{Q}_{(n,n),(n,n)}  \\
    };
\end{tikzpicture} \right] .
\end{equation*}
As a result, we can write the penalty function as:
\begin{equation*}
\int_X | \text{Tr}\left( \Gamma J_v \right) |^2 dx \approx \sum_{i,j = 1}^n | v_1(x_{i,j+1}) - v_1(x_{i,j}) - v_2(x_{i+1,j}) + v_2(x_{i,j})  |^2 = z_v^\top \mathcal{Q}^\top \mathcal{Q} z_v .
\end{equation*}

Finally, following our approximation of the objective function, we remain to choose the vector $z_v$ to minimize: 
\begin{equation*}
\frac{1}{2} (z_{\dot{\mathcal{A}}} - z_v )^\top \mathcal{B}^\top \mathcal{B} (z_{\dot{\mathcal{A}}} - z_v ) + \frac{\beta}{2} z_v^\top \mathcal{Q}^\top \mathcal{Q} z_v \hspace{0.04 cm}.
\end{equation*}
The optimality condition with respect to $z_v$ is:
\begin{equation*}
0 = \mathcal{B}^2 (z_v - z_{\dot{\mathcal{A}}}) + \beta \mathcal{Q}^\top \mathcal{Q} z_v   \hspace{0.04 cm}.
\end{equation*}
so that the final solution can be written as:
\begin{equation*}
z_{\dot{\mathcal{A}}} = \big( I + \beta \mathcal{B}^{-2} \mathcal{Q}^\top \mathcal{Q} \big) z_{v} \hspace{0.04 cm} ,
\end{equation*}
which is what we characterize numerically.

\end{document}